\documentclass[12pt]{article}
\usepackage[paper=a4paper,margin=1in]{geometry}
\usepackage{t1enc}      

\usepackage{amsmath}
\usepackage{amsfonts}
\usepackage{amssymb}
\usepackage{amsthm}
\usepackage{graphicx}
\usepackage{mathrsfs}  

\newcommand{\edth} {\mbox{\symbol{'360}}}

\newcommand{\ua}{\underline a \,}

\newcommand{\uA}{\underline A \,}
\newcommand{\uB}{\underline B \,}

\newcommand{\bi}{\bf i}
\newcommand{\bj}{\bf j}
\newcommand{\bk}{\bf k}

\newcommand{\bA}{\bf A}
\newcommand{\bB}{\bf B}
\newcommand{\bC}{\bf C}
\newcommand{\bD}{\bf D}

\newcommand{\supp}{{\rm supp\,}}

\newtheorem{theorem}{Theorem}[section]
\newtheorem{proposition}{Proposition}[section]
\newtheorem{lemma}{Lemma}[section]
\newtheorem{corollary}[theorem]{Corollary}

\numberwithin{equation}{section}

\begin{document}
\bibliographystyle{unsrt}

\title{A positive Bondi--type mass in asymptotically de Sitter spacetimes}

\author{L\'aszl\'o B. Szabados \\
  Wigner Research Centre for Physics, \\
  H--1525 Budapest 114, P. O. Box 49, Hungary
\and
Paul Tod\\
  Mathematical Institute, Oxford University,\\
  Oxford OX2 6GG}

\maketitle

\begin{abstract}
The general structure of the conformal boundary $\mathscr{I}^+$ of 
asymptotically de Sitter spacetimes is investigated. First we show that 
Penrose's quasi-local mass, associated with a cut ${\cal S}$ of the 
conformal boundary, can be zero even in the presence of outgoing 
gravitational radiation. On the other hand, following a Witten--type 
spinorial proof, we show that an analogous expression based on the 
Nester--Witten form is finite only if the Witten spinor field solves the 
2-surface twistor equation on ${\cal S}$, and it yields a positive 
functional on the 2-surface twistor space on ${\cal S}$, provided the 
matter fields satisfy the dominant energy condition. Moreover, this 
functional is vanishing if and only if the domain of dependence of the 
spacelike hypersurface which intersects $\mathscr{I}^+$ in the cut 
${\cal S}$ is locally isometric to the de Sitter spacetime. For 
non-contorted cuts this functional yields an invariant analogous to the 
Bondi mass. 
\end{abstract}


\section{Introduction}
\label{sec-1}

The simplest explanation of the root of the deviation of the observed red 
shift vs. luminosity diagram of distant type Ia supernovae from the expected 
one is probably the strict positivity of the cosmological constant 
\cite{Riessetal,Perlmutteretal,Pe10}. Thus the history of our observed 
universe should be modeled by asymptotically de Sitter spacetimes, and hence 
a systematic study of these spacetimes, e.g. their asymptotic properties, is 
physically justified. The conformally cyclic cosmological model (or shortly 
CCC model) of Penrose \cite{Pe10} is based on the positivity of the 
cosmological constant. In this model the crossover hypersurface is just the 
(spacelike) future timelike infinity of the previous aeon and the big bang 
singularity of the present aeon. 

In the study of asymptotic properties of spacetime one of the most 
important questions is that of the `conserved' quantities, e.g. the 
energy-momentum. In asymptotically flat or asymptotically anti-de Sitter 
spacetimes the `total' energy-momentum is thought of as associated with a 
\emph{localized} gravitating source at an arbitrary but \emph{finite} 
instant, represented by a spacelike hypersurface in the spacetime (or, 
rather, by a `cut' of the null or timelike conformal boundary, respectively, 
at \emph{finite} retarded time). This may be strictly conserved and total 
(like the ADM energy-momentum \cite{ADM}), or may change in time but still be  
total (like the total mass associated with \emph{closed} Cauchy surfaces in 
closed universes \cite{Sz12,Sz13}), or may change in time and be associated 
only with an (infinite) portion of spacetime (like the Bondi--Sachs 
\cite{Bondietal,Sachs,ENP} energy-momentum at the future null infinity, or 
the Abbott--Deser \cite{AD} energy-momentum at the conformal boundary of 
asymptotically anti de Sitter spacetime, which are changing from cut to cut 
in the presence of outgoing radiation or of in- or outgoing radiation, 
respectively). However, the analogous `total' energy-momentum in 
asymptotically de Sitter spacetimes would be associated with the 
\emph{spacelike} future conformal boundary (or a part of this boundary), and 
hence it would be interpreted as being associated with the \emph{asymptotic 
final state} of the universe (or a part of it). Though it would not have any 
dynamical content (in contrast to the total mass of closed universes, the 
Bondi--Sachs or the Abbott--Deser energy-momenta), it could provide a 
good control on the spacetime geomery near the conformal boundary. 

The idea of the energy-momentum associated with a cut of the future 
conformal boundary has already been raised by Penrose \cite{Pe11}, where two 
potentially viable strategies were also discussed. One of these approaches is 
based on the charge integral of the curvature and the use of twistorial 
methods, and the second is analogous to the ideas behind the Bondi--Sachs 
energy-momentum. (The asymptotic properties of the spacetime geometry, 
including the `conserved' quantities, are already well known when the 
spacetime is asymptotically flat \cite{PR2}, and also when the cosmological 
constant is \emph{negative} \cite{PR2,AM,Ke}. For an analogous recent 
investigation in the presence of a \emph{positive} cosmological constant, see 
\cite{ABK}. For a different concept of `total' mass in the presence of a 
positive $\Lambda$, see \cite{AD}, which mass can, however, be negative 
\cite{NaShiMa,KaTr,Luoetal,LiZh}. The second strategy of Penrose was also 
discussed in a nutshell by Frauendiener \cite{Joerg}. The conserved 
quantities of the linearized theory on de Sitter background are discussed 
in \cite{Tod15}.) 

In the present paper we investigate the general structure of the conformal 
boundary, and the notion of `total' energy-momentum associated with a cut of 
the future conformal boundary and based both on the charge integral of the 
curvature and of the Nester--Witten 2-form, too. We refine (and at certain 
points improve) the previous analyses of the general structure of the 
conformal boundary itself, construct a coordinate system and complex null 
tetrad (that are analogous to the Bondi and Newman--Penrose ones, 
respectively, near the null infinity of asymptotically flat spacetimes), and 
determine the asymptotic geometry of the smooth spacelike hypersurfaces that 
extend to the (spacelike) conformal boundary. These provide the technical 
background for the investigation of the energy-momentum. We show that the 
construction for the energy-momentum based on the charge integral of the 
curvature and the use of the 2-surface twistors does \emph{not} have the 
rigidity property: It may be vanishing even for non-trivial spacetime 
configurations, e.g. in vacuum spacetimes with non-vanishing rescaled 
conformal electric curvature (representing for example pure outgoing 
gravitational radiation). Hence it does not seem to provide an appropriate 
measure of the `strength' of the gravitational `field'. 

On the other hand, we show that the expression based on the Nester--Witten 
2-form has the positivity and rigidity properties: Following a Witten-type 
argument on a spacelike hypersurface $\Sigma$ that intersects the future 
conformal boundary in a spacelike cut ${\cal S}$, the integral of the 
Nester--Witten 2-form on ${\cal S}$ defines a non-negative functional on the 
space of the boundary values for the Dirac spinor solution of the Witten 
equation on $\Sigma$, provided the matter fields satisfy the dominant energy 
condition; and this functional is vanishing if and only if the domain of 
dependence of $\Sigma$ is locally isometric to the de Sitter spacetime. 
Interestingly enough, we may \emph{not} impose any boundary condition for 
the Witten spinors at the cut by hand. All the boundary conditions are 
determined by the Witten equation itself and the requirement of the finiteness 
of the resulting functional; and the boundary condition is that the spinor 
field on the cut \emph{must solve the 2-surface twistor equation} of 
Penrose. On the other hand, the interpretation of the resulting (positive 
definite) functional is not trivial: What we could consider to be the 
energy-momentum 4-vector is only a part of a bigger multiplet of quantities 
which are `mixed' by the symmetry group of the 2-surface twistor space. 
Energy-momentum, or at least mass, could be defined in an invariant way only 
in the presence of extra structures on the 2-surface twistor space. In 
particular, if the cut is non-contorted (e.g. when the whole conformal 
boundary is intrinsically locally conformally flat), then a positive measure 
of the strength of the gravitational `field', which could be interpreted as 
mass, is found whose vanishing is equivalent to the local de Sitter nature of 
the domain of dependence of the hypersurface $\Sigma$ above. 

In the proof of the existence of solutions of the Witten equation we 
use functional analytic techniques. It turns out that, on asymptotically 
hyperboloidal hypersurfaces, the solutions of the Witten equations and their 
derivatives fall off with the \emph{same} rate. Thus, the weighted Sobolev 
spaces do \emph{not} seem to be the appropriate function spaces, and the 
classical Sobolev spaces with an overall weight function in front of the 
volume element should be used. The necessary technical details are also 
developed here. 

The structure of the paper follows the logic above: In section \ref{sec-2} 
we discuss the structure of the conformal boundary, section \ref{sec-3} is 
devoted to the discussion of the properties of an expression based on the 
integral of the curvature on the cut. Then, in section \ref{sec-4}, we 
consider the expression based on the Nester--Witten 2-form, prove the 
positivity and rigidity properties, and introduce the mass on non-contorted 
cuts. The Appendix is devoted to the introduction of the functional 
analytic tools and statements needed in the proof of the existence and 
uniqueness of solutions of the Witten equation on asymptotically hyperboloidal 
hypersurfaces. 

We use the abstract index formalism and the sign conventions of \cite{PR2}. 
In particular, the signature of the spacetime metric is $(+,-,-,-)$, and the 
Riemann tensor is defined according to $-R^a{}_{bcd}X^bV^cW^d:=V^c\nabla_c
(W^d\nabla_dX^a)-W^c\nabla_c(V^d\nabla_dX^a)-[V,W]^c\nabla_cX^a$ for any vector 
fields $X^a$, $V^a$ and $W^a$. Thus, Einstein's equations take the form $G_{ab}
:=R_{ab}-\frac{1}{2}Rg_{ab}=-\varkappa T_{ab}-\Lambda g_{ab}$, where $\varkappa:=
8\pi G$ and $G$ is Newton's gravitational constant, and $\Lambda>0$. 


\section{The general framework}
\label{sec-2}

\subsection{The conformal boundary of spacetimes with positive
$\Lambda$}
\label{sub-2.1}

Let the physical spacetime be denoted by $(\hat M,\hat g_{ab})$, which is 
assumed to admit a nontrivial smooth conformal completion $(M,g_{ab},\Omega)$; 
i.e. there exist a manifold $M$ with nonempty boundary $\partial M$, a 
Lorentzian metric $g_{ab}$ on $M$ and a function $\Omega:M\rightarrow[0,
\infty)$, all of them smooth, such that 

\begin{description}
\item(i) $M-\partial M$ is diffeomorphic to (and hence identified with) 
    $\hat M$;
\item(ii) $g_{ab}=\Omega^2\hat g_{ab}$ on $\hat M$; 
\item(iii) the boundary is just $\partial M=\{\Omega=0\}$, and $\nabla_a
    \Omega$ is nowhere vanishing on $\partial M$; 
\item(iv) $\hat g_{ab}$ solves Einstein's equations $\hat R_{ab}-\frac{1}{2}
    \hat R\hat g_{ab}=-\varkappa\hat T_{ab}-\Lambda\hat g_{ab}$ with 
    \emph{positive} cosmological constant $\Lambda$ and energy-momentum 
    tensor $\hat T_{ab}$ satisfying the dominant energy condition; 
\item(v) the physical energy-momentum tensor satisfies the fall-off condition 
    $\hat T^a{}_b=\Omega^3T^a{}_b$, where $T^a{}_b$ is some smooth tensor field 
    on $M$.\footnote{It might be worth noting that this fall-off condition 
    is stronger than is needed to retain the smoothness of the conformal 
    boundary (cf. \cite{Fr13}).} 
\end{description}
It is known \cite{PR2} that $(\nabla_e\Omega)(\nabla^e\Omega)\approx\Lambda
/3$, where $\approx$ means `equal at the points of $\partial M$'. Thus the 
positivity of $\Lambda$ is equivalent to the \emph{spacelike} nature of 
$\partial M$, and we denote by $\mathscr{I}^+$ the part of $\partial M$ 
whose points are \emph{future} endpoints of inextensible non-spacelike 
curves in $\hat M$. (In the CCC model such a boundary hypersurface represents 
the crossover hypersurface between two successive aeons.) The hat on 
quantities and objects is referring to the physical spacetime, while the 
unhatted ones are in the unphysical $(M,g_{ab})$. (See e.g. \cite{ABK}, or 
for the analogous definition in the $\Lambda<0$ case, see \cite{AM}.) 

A number of facts follow from this definition. In particular, 

\begin{description}
\item(1) $\hat C^a{}_{bcd}=C^a{}_{bcd}\approx0$, i.e. the Weyl tensors 
    are vanishing on $\mathscr{I}^+$ \cite{PR2}; 
\item(2) there is a conformal factor $\Omega$ such that the extrinsic 
    curvature $\chi_{ab}$ and the acceleration $a_e$ of $\mathscr{I}^+$ in 
    $(M,g_{ab})$ are vanishing (`conformal Bondi gauge'). The remaining 
    conformal gauge freedom is $g_{ab}\mapsto\omega^2g_{ab}$, where $\omega=
    \omega_0+\Omega^2\Theta$ is a strictly positive function on $M$ in which 
    $\Theta$ is arbitrary on $M$, $\omega_0$ is arbitrary but positive on 
    $\mathscr{I}^+$ and is constant along its normals (at least in a 
    neighbourhood of $\mathscr{I}^+$). Thus we can still rescale the 
    \emph{intrinsic} induced metric of $\mathscr{I}^+$ freely \cite{PR2}; 
\item(3) if $N^a$ denotes the future pointing, $g_{ab}$--unit normal to the 
    $\Omega={\rm const}$ hypersurfaces (denoted henceforth by ${\cal H}
    _\Omega$), and $P^a_b:=\delta^a_b-N^aN_b$, the $g_{ab}$--orthogonal 
    projection to ${\cal H}_\Omega$, then the part $P^a_b\hat T^b{}_cN^c$ of 
    the physical energy-momentum tensor tends to zero in the $\Omega
    \rightarrow0$ limit \emph{faster} than $\Omega^3$, i.e. with $\Omega^4$ 
    \cite{AM}. 
\end{description}
However, the analysis behind these results can be refined further, yielding 
a slightly more detailed characterization of the structure of the conformal 
boundary and minor corrections of previous results. In particular, the same 
general formulae from which the faster fall-off $P^a_b\hat T^b{}_cN^c=
O(\Omega^4)$ was derived in \cite{AM} imply that $P^a_b\hat T^b{}_cP^c_d=
O(\Omega^4)$ also holds; or that the divergence equation (10) of \cite{AM} 
for the electric part of the rescaled Weyl tensor should be corrected. Thus, 
in the next subsection, a more detailed discussion of the structure of the 
conformal boundary will be given.

\subsection{The structure of the conformal boundary}
\label{sub-2.2}

\subsubsection{Consequences of Einstein's equations}
\label{sub-2.2.1}

By Einstein's equations, the assumption on the asymptotic form of the 
physical energy-momentum tensor and the conformal rescaling formulae, the 
Einstein tensor of the unphysical spacetime is 

\begin{eqnarray}
R_{ab}\!\!\!\!&-\!\!\!\!&\frac{1}{2}Rg_{ab}=\hat R_{ab}-\frac{1}{2}\hat R
  \hat g_{ab}+2\Omega^{-1}\bigl(\nabla_a\nabla_b\Omega-g_{ab}\nabla_c
  \nabla^c\Omega\bigr)+3\Omega^{-2}g_{ab}\bigl(\nabla_c\Omega\bigr)\bigl(
  \nabla^c\Omega\bigr) \nonumber \\
=\!\!\!\!&-\!\!\!\!&\varkappa\Omega T_{ab}+\Omega^{-2}g_{ab}\Bigl(3\bigl(
  \nabla_c\Omega\bigr)\bigl(\nabla^c\Omega\bigr)-\Lambda\Bigr)+
  2\Omega^{-1}\bigl(\nabla_a\nabla_b\Omega-g_{ab}\nabla_c\nabla^c\Omega\bigr).
  \label{eq:2.2.1a}
\end{eqnarray}
Multiplying by $\Omega^2$ and evaluating at $\Omega=0$, and multiplying 
by $\Omega$, taking the trace and evaluating at $\Omega=0$, respectively, 
we obtain that $3(\nabla_c\Omega)(\nabla^c\Omega)\approx\Lambda$ and $4\nabla
_a\nabla_b\Omega\approx g_{ab}(\nabla_c\nabla^c\Omega)$ (see also \cite{PR2}). 
Thus, smoothness of the extension implies that there exist a smooth function 
$\phi$ and a tensor field $\Psi_{ab}$ on $M$ such that 

\begin{equation}
\bigl(\nabla_c\Omega\bigr)\bigl(\nabla^c\Omega\bigr)=\frac{1}{3}\Lambda+
 \Omega\phi, \hskip 20pt
\nabla_a\nabla_b\Omega=\frac{1}{4}g_{ab}\bigl(\nabla_c\nabla^c\Omega\bigr)+
 \Omega\Psi_{ab}. \label{eq:2.2.2a}
\end{equation}
Taking the trace of the second we find that $\Psi_{ab}g^{ab}=0$. It is known 
(see e.g. \cite{PR2}) that the conformal factor can always be chosen such 
that $\nabla_c\nabla^c\Omega\approx0$ (the `Bondi conformal gauge'), i.e. 
$\nabla_c\nabla^c\Omega=\Omega\Psi$ for some smooth $\Psi$ on $M$. In the 
rest of the present paper we assume that $\Omega$ is such a conformal factor. 
Hence 

\begin{equation}
\nabla_a\nabla_b\Omega=\Omega\bigl(\frac{1}{4}g_{ab}\Psi+\Psi_{ab}\bigr).
\label{eq:2.2.2b}
\end{equation}
Substituting the first of (\ref{eq:2.2.2a}) and (\ref{eq:2.2.2b}) into 
(\ref{eq:2.2.1a}) we obtain 

\begin{equation*}
R_{ab}-\frac{1}{2}Rg_{ab}=-\varkappa\Omega T_{ab}+2\Psi_{ab}-\frac{3}{2}\Psi g_{ab}
+3\Omega^{-1}\phi g_{ab}.
\end{equation*}
Multiplying by $\Omega$ and evaluating at $\Omega=0$, we obtain that $\phi
\approx0$, i.e. 

\begin{equation}
\bigl(\nabla_a\Omega\bigr)\bigl(\nabla^a\Omega\bigr)=\frac{1}{3}\Lambda+
\Omega^2\Phi \label{eq:2.2.2c}
\end{equation}
for some smooth function $\Phi$ on $M$. 

However, $\Psi_{ab}$, $\Psi$ and $\Phi$ are not independent. To see this, let 
us define the lapse $N$ of the foliation by the level sets ${\cal H}_\Omega$ 
of $\Omega$ by $1=:-NN^a\nabla_a\Omega$, by means of which the future pointing 
unit normal of the leaves is $N^a=-Ng^{ab}\nabla_b\Omega$ and $N=1/\vert\nabla
_c\Omega\vert\approx\sqrt{3/\Lambda}$ (by (\ref{eq:2.2.2c})). (Since $\Omega$ 
is \emph{decreasing} in the future direction and we want the lapse to be 
\emph{positive}, it should be defined with the minus sign.) For later use, we 
need the acceleration of the leaves ${\cal H}_\Omega$. By the definition of 
the lapse and (\ref{eq:2.2.2c}) this is 

\begin{equation}
a_e:=N^a\nabla_aN_e=-D_e(\ln N)=\frac{1}{2}N^2D_e\frac{1}{N^2}=\frac{1}{2}
\Omega^2N^2D_e\Phi, \label{eq:2.2.2d}
\end{equation}
where $D_e$ denotes the intrinsic Levi-Civita derivative operator on ${\cal H}
_\Omega$ determined by the induced (negative definite) metric $h_{ab}:=P^c_a
P^d_bg_{cd}$. Also, we can calculate the extrinsic curvature and its trace: 

\begin{eqnarray}
&{}&\chi_{ab}:=P^c_aP^d_b\nabla_cN_d=-\Omega N\bigl(P^c_aP^d_b\Psi_{cd}+
  \frac{1}{4}\Psi h_{ab}\bigr), \label{eq:2.2.3a} \\
&{}&\chi:=\chi_{ab}h^{ab}=-\Omega N\bigl(h^{ab}\Psi_{ab}+\frac{3}{4}\Psi\bigr).
  \label{eq:2.2.3b}
\end{eqnarray}
Note that $0=\Psi_{ab}g^{ab}=\Psi_{ab}h^{ab}+\Psi_{ab}N^aN^b$. In particular, 
(\ref{eq:2.2.3a}) shows that $\mathscr{I}^+$ is extrinsically flat. On the 
other hand, by (\ref{eq:2.2.2b}) and (\ref{eq:2.2.2c}), we can calculate the 
mean curvature $\chi$ in an alternative way: 

\begin{equation*}
\Omega\Psi=\nabla_a\nabla^a\Omega=\nabla_a\bigl(-\frac{1}{N}N^a\bigr)=
-\frac{1}{N}\chi-\frac{1}{2}NN^a\nabla_a\frac{1}{N^2}=-\frac{1}{N}\chi+\Omega
\Phi-\frac{1}{2}\Omega^2NN^a\nabla_a\Phi,
\end{equation*}
from which 

\begin{equation*}
\chi=-\Omega N\bigl(\Psi-\Phi\bigr)-\frac{1}{2}\Omega^2N^2N^a\nabla_a\Phi
\end{equation*}
follows. Comparing this with (\ref{eq:2.2.3b}) we see that 

\begin{equation}
\frac{1}{4}\Psi=h^{ab}\Psi_{ab}+\Phi-\frac{1}{2}\Omega NN^a\nabla_a\Phi\approx
h^{ab}\Psi_{ab}+\Phi. \label{eq:2.2.4}
\end{equation}
Thus, $\Psi$ is determined by $\Phi$ and the three-dimensional trace of 
$\Psi_{ab}$, and, in particular, on $\mathscr{I}^+$, it is an \emph{algebraic} 
expression of them. 

Substituting (\ref{eq:2.2.2b}) and (\ref{eq:2.2.2c}) into (\ref{eq:2.2.1a}) 
and taking into account (\ref{eq:2.2.4}) we obtain 

\begin{equation}
R_{ab}-\frac{1}{2}Rg_{ab}=-\Omega\bigl(\varkappa T_{ab}+3g_{ab}NN^c\nabla_c\Phi
\bigr)+2\Psi_{ab}-6g_{ab}h^{cd}\Psi_{cd}-3\Phi g_{ab}. \label{eq:2.2.5}
\end{equation}
On the other hand, it is known (e.g. from the initial value formulation of 
general relativity), that the various 3+1 pieces of the (unphysical) Einstein 
tensor are 

\begin{eqnarray}
N^cN^d\bigl(R_{cd}\!\!\!\!&-\!\!\!\!&\frac{1}{2}Rg_{cd}\bigr)=-\frac{1}{2}
  \bigl({\cal R}+\chi^2-\chi_{cd}\chi^{cd}\bigr), \label{eq:2.2.6a} \\
P^c_aN^d\bigl(R_{cd}\!\!\!\!&-\!\!\!\!&\frac{1}{2}Rg_{cd}\bigr)=-D_b\bigl(
  \chi^b{}_a-\chi\delta^b_a\bigr), \label{eq:2.2.6b} \\
P^c_aP^d_b\bigl(R_{cd}\!\!\!\!&-\!\!\!\!&\frac{1}{2}Rg_{cd}\bigr)={\cal R}_{ab}
  -\frac{1}{2}h_{ab}\bigl({\cal R}+\chi^2-\chi_{cd}\chi^{cd}\bigr)+\bigl(
  {\pounds}_N\chi_{cd}\bigr)P^c_aP^d_b+\chi\chi_{ab} \nonumber \\
\!\!\!\!&-\!\!\!\!&2\chi_{ac}\chi^c{}_b+\frac{1}{N}D_aD_bN-h_{ab}\bigl(h^{cd}
  {\pounds}_N\chi_{cd}-\chi_{cd}\chi^{cd}+\frac{1}{N}D_cD^cN\bigr).
  \label{eq:2.2.6c}
\end{eqnarray}
Here ${\cal R}_{ab}$ is the Ricci tensor and ${\cal R}$ the intrinsic scalar 
curvature of $D_a$, and ${\pounds}_N$ denotes Lie derivative along the unit 
normal vector field $N^a$. Comparing the various 3+1 parts of (\ref{eq:2.2.5}) 
with (\ref{eq:2.2.6a}), (\ref{eq:2.2.6b}) and (\ref{eq:2.2.6c}), respectively, 
and using the definitions, equations (\ref{eq:2.2.3a})-(\ref{eq:2.2.4}) and 
the fact $\Psi_{ab}N^aN^b=-h^{ab}\Psi_{ab}$, we find 

\begin{eqnarray}
&{}&{\cal R}\approx 16h^{ab}\Psi_{ab}+6\Phi, \label{eq:2.2.7a} \\
&{}&P^c_aN^d\Psi_{cd}\approx 0, \label{eq:2.2.7b} \\
&{}&{\cal R}_{ab}\approx P^c_aP^d_b\Psi_{cd}+5h_{ab}h^{cd}\Psi_{cd}+2\Phi h_{ab}.
 \label{eq:2.2.7c}
\end{eqnarray}
Comparing (\ref{eq:2.2.7c}) with (\ref{eq:2.2.3a}), by (\ref{eq:2.2.4}) we 
obtain that 

\begin{equation}
\sqrt{\frac{3}{\Lambda}}\lim_{\Omega\rightarrow0}\bigl(\Omega^{-1}\chi_{ab}\bigr)
\approx-\bigl({\cal R}_{ab}-\frac{1}{4}{\cal R}h_{ab}\bigr)-\frac{1}{2}\Phi
h_{ab}. \label{eq:2.2.8}
\end{equation}
Note that the first term on the right is just the Schouten tensor of the 
(three dimensional) Riemannian geometry $(\mathscr{I}^+,h_{ab})$. 

\subsubsection{Consequences of the Bianchi identity}
\label{sub-2.2.2}

By Einstein's equations and the assumption on the asymptotic form of the 
physical energy-momentum tensor, the Schouten tensor of the (four 
dimensional) physical spacetime is 

\begin{equation}
\hat S_{ab}:=-\bigl(\hat R_{ab}-\frac{1}{6}\hat R\hat g_{ab}\bigr)=\varkappa
\Omega\bigl(T_{ab}-\frac{1}{3}Tg_{ab}\bigr)-\frac{1}{3}\Omega^{-2}\Lambda
g_{ab}. \label{eq:2.2.9}
\end{equation}
Hence, by the definition of the Weyl tensor and the conformal rescaling 
formulae, the contracted Bianchi identity for the physical curvature tensor 
yields 

\begin{eqnarray}
0\!\!\!\!&=\!\!\!\!&\hat\nabla_a\hat C^a{}_{bcd}+\frac{1}{2}\bigl(\hat\nabla_c
  \hat S_{db}-\hat\nabla_d\hat S_{cb}\bigr) \label{eq:2.2.10} \\
\!\!\!\!&=\!\!\!\!&\nabla_aC^a{}_{bcd}-\Omega^{-1}\bigl(\nabla_a\Omega\bigr)
  C^a{}_{bcd}+\frac{1}{2}\varkappa\Omega\bigl(\nabla_cT_{db}-\nabla_dT_{cb}\bigr)-
  \frac{1}{6}\varkappa\Omega\bigl(g_{bd}\nabla_cT-g_{bc}\nabla_dT\bigr)
  \nonumber \\
\!\!\!\!&-\!\!\!\!&\varkappa\bigl(T_{bc}\nabla_d\Omega-T_{bd}\nabla_c\Omega\bigr)-
  \frac{1}{2}\varkappa\bigl(g_{bc}T_d{}^e-g_{bd}T_c{}^e\bigr)\nabla_e\Omega+
  \frac{1}{2}\varkappa T\bigl(g_{bc}\nabla_d\Omega-g_{bd}\nabla_c\Omega\bigr).
  \nonumber
\end{eqnarray}
Multiplying this equation by $\Omega$ and evaluating it at $\Omega=0$, we 
obtain $(\nabla_a\Omega)C^a{}_{bcd}\approx0$, which is known to imply $C_{abcd}
\approx0$ \cite{PR2}. Thus we can write $C_{abcd}=\Omega K_{abcd}$ for some 
smooth tensor field $K_{abcd}$ on $M$. Substituting this form of the 
non-physical Weyl tensor back into (\ref{eq:2.2.10}), we obtain 

\begin{eqnarray}
2\Omega\nabla^aK_{abcd}\!\!\!\!&=\!\!\!\!&\varkappa\Omega\bigl(\nabla_dT_{cb}-
  \nabla_cT_{db}\bigr)+\frac{1}{3}\varkappa\Omega\bigl(g_{bd}\nabla_cT-g_{bc}
  \nabla_dT\bigr) \label{eq:2.2.11} \\
\!\!\!\!&+\!\!\!\!&2\varkappa\bigl(T_{bc}\nabla_d\Omega-T_{bd}\nabla_c\Omega
  \bigr)+\varkappa\bigl(g_{bc}T_d{}^e-g_{bd}T_c{}^e\bigr)\nabla_e\Omega+\varkappa
  T\bigl(g_{bd}\nabla_c\Omega-g_{bc}\nabla_d\Omega\bigr). \nonumber
\end{eqnarray}
Evaluating this at $\Omega=0$ and then contracting with $N^b$ and $P^b_aN^c$, 
respectively, we obtain 

\begin{equation}
T_{bc}P^b_aN^c\approx0, \hskip 25pt
T_{cd}P^c_aP^d_b\approx0. \label{eq:2.2.12}
\end{equation}
The first of these has already been derived in \cite{AM} (in the presence of 
a negative cosmological constant). Therefore, the asymptotic form of the 
various 3+1 pieces of $T_{ab}$ is 

\begin{equation}
T_{ab}N^aN^b=\nu+\Omega\mu, \hskip 15pt
T_{bc}P^b_aN^c=\Omega J_a, \hskip 15pt
T_{cd}P^c_aP^d_b=\Omega\Sigma_{ab}, \label{eq:2.2.13}
\end{equation}
for some smooth $\nu$, $\mu$, $J_a$ and $\Sigma_{ab}$ on $M$. However, these 
quantities are not quite independent, because they should satisfy the local 
conservation law $\hat\nabla_a\hat T^a{}_b=0$. Next we evaluate the 
consequences of this restriction. 

By the conformal rescaling formulae and the asymptotic form of the 
energy-momentum tensor we have that 

\begin{equation}
0=\Omega^{-2}\hat\nabla_a\hat T^a{}_b=\Omega\nabla_aT^a{}_b-\nabla_a\Omega
\bigl(T^a{}_b-T\delta^a_b\bigr). \label{eq:2.2.14}
\end{equation}
(It might be worth noting that the contraction of this equation with $NP^b_c$ 
and $NN^b$, respectively, yields $T_{ab}N^aP^b_c\approx0$ and $T_{ab}h^{ab}
\approx0$. The former is just the first of (\ref{eq:2.2.12}), but the latter 
is only the trace of the second.) Then substituting (\ref{eq:2.2.13}) into 
(\ref{eq:2.2.14}) and contracting with $N^b$ and $P^b_c$, respectively, and 
evaluating the resulting equations on $\mathscr{I}^+$, we obtain 

\begin{equation}
\mu+\Sigma_{ab}h^{ab}\approx\sqrt{3/\Lambda}N^a\nabla_a\nu,
  \hskip 20pt
\sqrt{3/\Lambda}D^b\Sigma_{ba}\approx-\bigl({\pounds}_NJ_b\bigr)P^b_a,
\label{eq:2.2.15}
\end{equation}
where, by (\ref{eq:2.2.2d}), we used $D_aN=O(\Omega^2)$. Thus, if $\nu$ were 
constant in time at $\mathscr{I}^+$ (e.g. when the whole physical 
energy-momentum tensor fell off as $\Omega^4T^a{}_b$, in which case $\nu$ 
itself would be zero), then by the first of (\ref{eq:2.2.15}) the $\Omega^4$ 
order part of the physical energy-momentum tensor $\hat T^a{}_b$ would be 
asymptotically trace-free. 

Taking into account (\ref{eq:2.2.13}), equation (\ref{eq:2.2.11}) can be 
rewritten into the form 

\begin{eqnarray}
2\nabla^aK_{abcd}\!\!\!\!&=\!\!\!\!&\varkappa\bigl(\nabla_dT_{cb}-\nabla_cT_{db}
  \bigr)+\frac{1}{3}\varkappa\bigl(g_{bd}\nabla_cT-g_{bc}\nabla_dT\bigr)+
  \varkappa\frac{1}{N}\Bigl(N_b\bigl(J_dN_c-J_cN_d\bigr) \nonumber \\
\!\!\!\!&{}\!\!\!\!&+\bigl(h_{bd}J_c-h_{bc}J_d\bigr)+2\bigl(\Sigma_{bd}N_c-
  \Sigma_{bc}N_d\bigr)-\Sigma_{ef}h^{ef}\bigl(h_{bd}N_c-h_{bc}N_d\bigr)\Bigr).
  \label{eq:2.2.16}
\end{eqnarray}
Recalling that the electric and magnetic parts of the rescaled Weyl tensor 
$K_{abcd}$ are defined by ${\cal E}_{ab}:=K_{acbd}N^cN^d$ and ${\cal B}_{ab}:=
\frac{1}{2}K_{acef}\varepsilon^{ef}{}_{bd}N^cN^d$, respectively, from 
(\ref{eq:2.2.16}) we obtain that 

\begin{eqnarray}
D^a{\cal E}_{ab}\!\!\!\!&=\!\!\!\!&P^e_a\nabla_e\bigl(K^a{}_{cfd}N^cN^d\bigr)
  P^f_b=\bigl(\nabla^aK_{acfd}\bigr)N^cN^dP^f_b-N^aK_{acdf}\chi^{cd}P^f_b
  \nonumber \\
\!\!\!\!&=\!\!\!\!&-\frac{\varkappa}{N}J_b-\frac{1}{3}\varkappa D_b\nu-
  \frac{1}{2}\frac{\varkappa}{N}\nu D_bN-N^aK_{acdf}\chi^{cd}P^f_b+\frac{1}{2}
  \varkappa\Omega\Bigl(-\mu\bigl(D_b\ln N\bigr)\nonumber \\
\!\!\!\!&{}\!\!\!\!&+\bigl(D^a\ln N\bigr)\Sigma_{ab}-\frac{2}{3}D_b\mu+2J^a
  \chi_{ab}+\frac{1}{3}\bigl(D_b\Sigma_{cd}\bigr)h^{cd}+N^a\bigl(\nabla_aJ_c
  \bigr)P^c_b\Bigr) \nonumber \\
\!\!\!\!&\approx\!\!\!\!&-\varkappa\sqrt{\Lambda/3}J_b-\frac{1}{3}\varkappa D_b
  \nu \label{eq:2.2.17a}
\end{eqnarray}
and 

\begin{eqnarray}
D^a{\cal B}_{ab}\!\!\!\!&=\!\!\!\!&P^h_a\nabla_h\bigl(\frac{1}{2}K^a{}_{ecd}
  \varepsilon^{cd}{}_{gf}N^eN^f\bigr)P^g_b \nonumber \\
\!\!\!\!&=\!\!\!\!&\frac{1}{2}\bigl(\nabla^aK_{aecd}\bigr)\varepsilon^{cd}
  {}_{bf}N^eN^f-\frac{1}{2}K_{aecd}\varepsilon^{cd}{}_{fg}\chi^{af}N^eP^g_b-\bigl(
  D^a\ln N\bigr){\cal B}_{ab} \label{eq:2.2.17b} \\
\!\!\!\!&=\!\!\!\!&\frac{1}{2}\varkappa\Omega\bigl(\chi_c{}^e\Sigma_{ed}-D_cJ_d
  \bigr)\varepsilon^{cd}{}_{af}N^f+\frac{1}{2}N^a\chi^{cd}K_{acef}\varepsilon
  ^{ef}{}_{da}P^a_b-\bigl(D^a\ln N\bigr){\cal B}_{ab}\approx0, \nonumber
\end{eqnarray}
where we used that $D_bN\approx0$ and $\chi_{cd}\approx0$. (\ref{eq:2.2.17a}) 
shows that, in addition to $J_b$, the gradient of the slow fall-off (i.e. 
$\Omega^3$ order) part $\nu$ of the energy density also contributes to the 
divergence of the electric part of the rescaled Weyl curvature. An analogous 
term on the right hand side of equation (10) of \cite{AM} should also be 
present. 

Finally, by the definition of the Weyl tensor, equations 
(\ref{eq:2.2.6a})-(\ref{eq:2.2.6c}) and the analogous decomposition of the 
Riemann tensor, 

\begin{eqnarray}
R_{efgh}P^e_aP^f_bP^g_cP^h_d\!\!\!\!&=\!\!\!\!&{\cal R}_{abcd}+\chi_{ac}
  \chi_{bd}-\chi_{ad}\chi_{bc}, \label{eq:2.2.18a} \\
R_{efgh}N^eP^f_bP^g_cP^h_d\!\!\!\!&=\!\!\!\!&D_c\chi_{db}-D_d\chi_{cb},
  \label{eq:2.2.18b} \\
R_{efgh}N^eN^gP^f_bP^h_d\!\!\!\!&=\!\!\!\!&P^a_bP^c_d\bigl({\pounds}_N\chi
  _{ac}\bigr)-\chi_{be}\chi^e{}_d-D_ba_d+a_ba_d, \label{eq:2.2.18c}
\end{eqnarray}
and of the scalar curvature, 

\begin{equation}
R={\cal R}+2h^{ab}\bigl({\pounds}_N\chi_{ab}\bigr)+\chi^2-3\chi_{ab}\chi^{ab}
+\frac{1}{N}D_aD^aN, \label{eq:2.2.19}
\end{equation}
we can determine the explicit form of the electric and magnetic parts of the 
Weyl tensor themselves, $E_{ab}:=C_{acbd}N^cN^d$ and $B_{ab}:=\frac{1}{2}C_{acef}
\varepsilon^{ef}{}_{bd}N^cN^d$, respectively. We obtain 

\begin{eqnarray}
E_{ab}\!\!\!\!&=\!\!\!\!&\frac{1}{2}\Bigl(P^c_aP^d_b\bigl({\pounds}_N\chi_{cd}
  \bigr)-{\cal R}_{ab}-\chi\chi_{ab}+\frac{1}{N}D_aD_bN \nonumber \\
\!\!\!\!&{}\!\!\!\!&-\frac{1}{3}h_{ab}\bigl(h^{cd}{\pounds}_N\chi_{cd}-{\cal
  R}-\chi^2+\frac{1}{N}D_cD^cN\bigr)\Bigr), \label{eq:2.2.20a} \\
B_{ab}\!\!\!\!&=\!\!\!\!&\bigl(D_c\chi_{d(a}\bigr)\varepsilon^{cd}{}_{b)},
  \label{eq:2.2.20b}
\end{eqnarray}
where $\varepsilon_{abc}:=N^e\varepsilon_{eabc}$ is the induced volume 3-form 
on ${\cal H}_\Omega$. In the rest of the paper, we need the explicit form of 
${\cal B}_{ab}=\Omega^{-1}B_{ab}$ at $\mathscr{I}^+$ only, for which, by 
(\ref{eq:2.2.8}), we obtain 

\begin{equation}
\sqrt{\frac{3}{\Lambda}}{\cal B}_{ab}\approx D_c\bigl(-{\cal R}_{d(a}+
\frac{1}{4}{\cal R}h_{d(a}\bigr)\varepsilon^{cd}{}_{b)}=:Y_{ab}, \label{eq:2.2.21}
\end{equation}
where $Y_{ab}$ is known as the Cotton--York tensor: Its vanishing is known 
to be equivalent to the \emph{local} conformal flatness of $(\mathscr{I}^+,
h_{ab})$. 

If the physical energy-momentum tensor $\hat T^a{}_b$ falls off purely at 
order $\Omega^3$, then by (\ref{eq:2.2.13}) this matter is a dust, and by the 
first of (\ref{eq:2.2.15}) the rescaled energy density $\nu=\Omega^{-3}\hat
T_{ab}\hat N^a\hat N^b$ is asymptotically constant. In particular, in the dust 
filled FRW spacetime with positive $\Lambda$ this $\nu$ is just the conserved 
first integral of one of the two Friedman equations. 

On the other hand, if $\hat T^a{}_b=\Omega^4 T^a{}_b$ (i.e. there is no 
$\Omega^3$ order term in $\hat T^a{}_b$), then by the first of 
(\ref{eq:2.2.15}) $\hat T^a{}_b$ is asymptotically trace-free, which is a 
characteristic property of conformally invariant (e.g. Maxwell or Yang--Mills) 
fields. For example, in the radiation filled FRW spacetime with positive 
$\Lambda$ theenergy-momentum tensor is trace-free, the energy density and 
isotropicpressure with respect to $\hat N^a$ fall off as $\Omega^4$, and the 
momentumdensity $J_a$ is vanishing, being compatible with the homogeneity of 
theisotropic pressure and the energy density, just according to the second of 
(\ref{eq:2.2.15}). Thus, the existence of a (spacelike) $\mathscr{I}^+$ 
restricts the form of the matter fields near $\mathscr{I}^+$. 


\section{On the Penrose mass at $\mathscr{I}^+$}
\label{sec-3}

If ${\cal S}$ is a closed, orientable spacelike 2-surface in $\hat M$, then 
for any smooth symmetric spinor field $\omega^{AB}$ we can form the integral 
of the complex 2-form $\omega^{AB}\hat R_{ABcd}$ on ${\cal S}$, where $\hat R
_{ABcd}$ is the anti-self-dual part of the physical spacetime curvature tensor. 
However, even in the de Sitter spacetime this integral diverges if we allow 
the 2-surface to tend to a cut of $\mathscr{I}^+$. Therefore, to get finite 
value, it seems reasonable to use in the charge integral the `renormalized' 
curvature 

\begin{equation}
\tilde R_{ABcd}:=\hat R_{ABcd}-\frac{1}{3}\Lambda\varepsilon_{A(C}\varepsilon
_{D)B}\varepsilon_{C'D'}, \label{eq:3.1}
\end{equation}
rather than $\hat R_{ABcd}$. (Clearly, $\tilde R_{ABcd}$  is the anti-self-dual 
part of the \emph{renormalized} curvature tensor $\tilde R_{abcd}:=\hat R_{abcd}
-\frac{1}{3}\Lambda(\hat g_{ac}\hat g_{bd}-\hat g_{ad}\hat g_{bc})$.) Then by 
the definition of $K^a{}_{bcd}$ and the Weyl tensor, Einstein's equation and 
the conformal rescaling formulae, the corresponding complex 2-form is 

\begin{eqnarray*}
\omega^{AB}\tilde R_{ABcd}\!\!\!\!&=\!\!\!\!&-\frac{1}{2}\omega^{AB}\hat
  \varepsilon^{A'B'}\Bigl(\hat R_{abcd}-\frac{1}{3}\Lambda\bigl(\hat g_{ac}
  \hat g_{bd}-\hat g_{ad}\hat g_{bc}\bigr)\Bigr) \\
\!\!\!\!&=\!\!\!\!&-\frac{1}{2}\omega^{AB}\hat\varepsilon^{A'B'}\Bigl(\hat C
  _{abcd}-\varkappa\bigl(\hat g_{ac}\hat T_{bd}-\hat g_{ad}\hat T_{bc}\bigr)+
  \frac{2}{3}\varkappa\hat T\hat g_{ac}\hat g_{bd}\Bigr)\\
\!\!\!\!&=\!\!\!\!&-\frac{1}{2}\omega^{AB}\varepsilon^{A'B'}\Bigl(K_{abcd}-
  \varkappa\bigl(g_{ac}T_{bd}-g_{ad}T_{bc}\bigr)+\frac{2}{3}\varkappa Tg_{ac}g_{bd}
  \Bigr).
\end{eqnarray*}
Thus the cosmological constant has, in fact, been canceled, and the 
expression on the right is well defined and has a finite integral on any 
closed spacelike 2-surface in $M$, even on a cut of $\mathscr{I}^+$. 

Hence, let ${\cal S}$ be a closed orientable 2-surface in $\mathscr{I}^+$, 
and denote its outward pointing $g_{ab}$-unit normal in $\mathscr{I}^+$ by 
$V^a$. Then, by the fall-off properties (\ref{eq:2.2.13}) of the various 
pieces of $T_{ab}$ and the definition of ${\cal E}_{ab}$ and ${\cal B}_{ab}$, 
we obtain 

\begin{eqnarray}
\tilde{\tt A}\bigl[\omega_{AB}\bigr]\!\!\!\!&:=\!\!\!\!&\frac{\rm i}{\varkappa}
  \oint_{\cal S}\omega^{AB}\tilde R_{ABcd} \nonumber \\
\!\!\!\!&=\!\!\!\!&-\frac{\rm i}{\varkappa}\oint_{\cal S}\omega^{AB}\varepsilon
  ^{A'B'}\Bigl(\bigl(N_a{\cal B}_{bc}-N_b{\cal B}_{ac}\bigr)V^c+\varepsilon_{abc}
  \bigl(\frac{1}{3}\varkappa\nu V^c+{\cal E}^c{}_dV^d\bigr)\Bigr){\rm d}
  {\cal S}. \label{eq:3.2}
\end{eqnarray}
Therefore, a general charge integral of the curvature is built from ${\cal
E}_{ab}$, ${\cal B}_{ab}$ and the slow fall-off (i.e. $\Omega^3$ order) part 
$\nu$ of the energy density. 

If $D\subset\mathscr{I}^+$ is an open domain with compact closure such that 
$\partial D={\cal S}$ and the spinor field $\omega_{AB}$ is well defined on 
$D$, then, by equations (\ref{eq:2.2.17a}) and (\ref{eq:2.2.17b}), the charge 
integral $\tilde{\tt A}[\omega_{AB}]$ can be rewritten as a 3-surface integral 
on $D$. In fact, if 

\begin{equation*}
Q_{ab}:=\frac{1}{2}\varepsilon_{abcd}\omega^{CD}\varepsilon^{C'D'}=-{\rm i}
\omega_{AB}\varepsilon_{A'B'},
\end{equation*}
then 

\begin{equation}
\tilde{\tt A}\bigl[\omega_{AB}\bigr]=\frac{2{\rm i}}{\varkappa}\int_D\Bigl\{D_e
Q_{ab}\Bigl(N^a\bigl(\frac{1}{3}\varkappa\nu h^{be}+{\cal E}^{be}\bigr)-
\varepsilon^{ab}{}_c{\cal B}^{ce}\Bigr)-\varkappa\sqrt{\frac{\Lambda}{3}}Q_{ab}
N^aJ^b\Bigr\}{\rm d}{\cal H}_0. \label{eq:3.3}
\end{equation}
Here ${\rm d}{\cal H}_0$ is the induced volume element on $\mathscr{I}^+$, 
and we extended the action of the intrinsic derivative operator $D_e$ from 
purely spatial tensors to arbitrary ones on $\mathscr{I}^+$ by $D_eN^a=0$. 

It is known that on a spacelike hypersurface that can be embedded with its 
first and second fundamental forms into some conformal Minkowski spacetimes 
the 3-surface twistor equation admits four (i.e. maximal number of) linearly 
independent solutions \cite{Tod84}. In particular, since $\mathscr{I}^+$ is 
extrinsically flat, this embeddability is equivalent to the local conformal 
flatness of $(\mathscr{I}^+,h_{ab})$. Thus, in the special case when its 
Cotton--York tensor is vanishing, $Y_{ab}=0$, the (already Riemannian) 
3-surface twistor equation, $D_{(AB}\lambda_{C)}=0$, admits four linearly 
independent solutions. Here $D_{AB}:=\sqrt{2}N_B{}^{A'}D_{AA'}=D_{(AB)}$ is the 
unitary spinor form of the intrinsic Levi-Civita derivative operator. These 
solutions are globally defined on $D$ if $D$ is homeomorphic to the 3-ball. 
Then, a direct calculation shows that $N^aQ_{ab}$ is a (complex) conformal 
Killing field on $D$ if $\omega_{AB}$ is a linear combination of symmetrized 
products $\lambda_{(A}\mu_{B)}$ of solutions of the 3-surface twistor equation. 
(The result that $N^aQ_{ab}$ is a conformal Killing vector was proven in the 
$\Lambda<0$ case by Kelly \cite{Ke} by considering the conformal boundary to 
be embedded in a conformal Minkowski spacetime as a hyperplane and using the 
solutions of the 1-valence twistor equation of that spacetime.) 

Since, however, ${\cal B}_{ab}=\sqrt{\Lambda/3}\,Y_{ab}$ on $\mathscr{I}^+$, 
by (\ref{eq:3.3}) for such spinor fields it is only $\nu$ and $J_a$ that 
contribute to the integral even if ${\cal E}_{ab}\not=0$. In particular, for 
the Maxwell (or Yang--Mills) field it is only $J_a$ that contributes to 
$\tilde{\tt A}[\lambda_{(A}\mu_{B)}]$ but its energy density does not; while 
in vacuum $\tilde{\tt A}[\lambda_{(A}\mu_{B)}]$ is zero even if ${\cal E}_{ab}
\not=0$. For example, in vacuum spacetimes with metric of Starobinskii form 
\cite{Starobinski} and intrinsically conformally flat $\mathscr{I}^+$ the 
electric part of the rescaled Weyl tensor is \emph{not} zero, but the charge 
integral with spinor fields solving the 3-surface twistor equation on a 
contractible $D$ is vanishing. Such an exact solution is the vacuum Kasner 
solution with positive $\Lambda$ (see \cite{exactsol}, Sect. 13.3.3). 

It is known that if $\lambda_A$ is a solution of the 3-surface twistor 
equation on a spacelike hypersurface, then its restriction to a 2-surface in 
this hypersurface solves the 2-surface twistor equation \cite{Tod84}. Thus by 
the general results above, when $(\mathscr{I}^+,h_{ab})$ is intrinsically 
locally conformally flat, the matter field is conformally invariant and its 
local energy flow is vanishing asymptotically (i.e. $J_a\approx0$), then the 
Penrose mass is zero even if the rescaled conformal electric curvature 
${\cal E}_{ab}$ is not. Therefore, the kinematical twistor (and, in 
particular, the Penrose mass) does not have the rigidity property: Its 
vanishing does \emph{not} imply the triviality of the spacetime geometry, 
i.e. actually that the past domain of dependence of $D$ is (locally) 
isometric to the de Sitter spacetime.


\section{Energy-momentum based on the Nester--Witten form}
\label{sec-4}

Another strategy to associate energy-momentum to 2-surfaces in $\mathscr{I}
^+$ might be based on the use of the integral of the Nester--Witten 2-form. 
In fact, this formalism was successfully used to give a unified spinorial 
reformulation of the ADM and Bondi--Sachs energy-momenta in asymptotically 
flat spacetimes \cite{HT}, and of the Abbott--Deser energy in asymptotically 
anti-de Sitter spacetimes \cite{Gibbonsetal}. Moreover, probably the simplest 
proof of the positivity of these energies is given in this formalism. 
Therefore, it seems natural to try to introduce energy-momentum in this 
manner in the presence of a positive cosmological constant, too.

\subsection{The Nester--Witten form and the Sen--Witten identity}
\label{sub-4.1}

In the present subsection we work exclusively in the \emph{physical} 
spacetime, but, for the sake of simplicity, \emph{in this subsection}, we 
leave the `hats' off of the quantities and objects. 

Let $\Sigma$ be a spacelike hypersurface in $M$ which extends to a 
hypersurface-with-boundary in the unphysical spacetime intersecting 
$\mathscr{I}^+$ in a smooth 2-surface homeomorphic to $S^2$ (a cut). Let 
$t^a$ denote its future pointing $g_{ab}$-unit timelike normal, $P^a_b:=
\delta^a_b-t^at_b$ is the $g_{ab}$-orthogonal projection to $\Sigma$ and 
$h_{ab}:=P^c_aP^d_bg_{cd}$ is the induced metric. (Although we use the same 
symbols $P^a_b$ and $h_{ab}$, and $\chi_{ab}$ below for the extrinsic curvature, 
they should not be confused with those introduced on the $\Omega={\rm const}$ 
hypersurfaces in section \ref{sec-2}.) 

For any pair $\lambda_A$, $\mu_A$ of spinor fields the general Nester--Witten 
form is defined by 

\begin{equation}
u\bigl(\lambda,\bar\mu\bigr)_{ab}:=\frac{\rm i}{2}\bigl(\bar\mu_{A'}
\nabla_{BB'}\lambda_A-\bar\mu_{B'}\nabla_{AA'}\lambda_B\bigr). \label{eq:4.1}
\end{equation}
Then the components of the energy-momentum 4-vector associated with a closed 
orientable 2-surface ${\cal S}$, which in the present case is the cut of the 
conformal boundary, will be defined by its integral, 

\begin{eqnarray}
H\bigl[\lambda,\bar\mu\bigr]:=\frac{2}{\varkappa}\oint_{\cal S}u\bigl(\lambda,
\bar\mu\bigr)_{cd}, \label{eq:4.1a}
\end{eqnarray}
for spinor fields $\lambda^A$ and $\mu^A$ belonging to some appropriately 
chosen two dimensional subspace ${\bf S}_{\bA}\subset C^\infty({\cal S},
\mathbb{S}^A)$ of the infinite dimensional space of the smooth spinor fields 
on ${\cal S}$. Thus, the boldface index is referring to this (still not 
specified) two dimensional space of spinor fields. Since a basis in this space 
is a pair of spinor fields, $\lambda^A_{\bA}=(\lambda^A_{\bf 0},\lambda^A
_{\bf 1})$, the boldface index can also be considered as a name index, too, 
taking numerical values: ${\bA}={\bf 0},{\bf 1}$. 
It is a simple calculation to show that this integral, as a bilinear map $H:
{\bf S}_{\bA}\times\bar{\bf S}_{{\bA}'}\rightarrow\mathbb{C}$, is Hermitian in 
the sense that $H[\lambda,\bar\mu]=\overline{H[\mu,\bar\lambda]}$ for any 
$\lambda^A,\mu^A\in{\bf S}_{\bA}$, and overline denotes complex conjugation. 
However, by the so-called polarization formula, 

\begin{equation}
H\bigl[\lambda,\bar\mu\bigr]=\frac{1}{2}\Bigl(H\bigl[\lambda+\mu,\overline{
\lambda+\mu}\bigr]+{\rm i}H\bigl[\lambda+{\rm i}\mu,\overline{\lambda+{\rm i}
\mu}\bigr]-(1+{\rm i})H\bigl[\lambda,\bar\lambda\bigr]-(1+{\rm i})H\bigl[
\mu,\bar\mu\bigr]\Bigr), \label{eq:4.1p}
\end{equation}
the bilinear form $H[\lambda,\bar\mu]$ is completely determined by the 
quadratic form $H[\alpha,\bar\alpha]$ on ${\bf S}_{\bA}$. Moreover, it would 
be enough to prove $H[\alpha,\bar\alpha]\geq0$ for any $\alpha^A\in{\bf S}
_{\bA}$, because this would already imply that $H$ as a Hermitian bilinear 
form is positive. 

The standard Witten-type spinorial proof of the positivity of the total 
energy is based on the integrated form of the Sen--Witten identity on a 
spacelike hypersurface $\Sigma$ whose boundary $\partial\Sigma$ at infinity 
is the 2-surface ${\cal S}$ in question (see e.g. \cite{RT}): 

\begin{equation*}
\oint_{\partial\Sigma}u(\alpha,\bar\alpha)_{cd}=\int_\Sigma B(\alpha){\rm d}\Sigma,
\end{equation*}
where $\alpha_A$ is defined on $\Sigma$ and 

\begin{equation*}
B(\alpha):=-2t^{AA'}\bigl({\cal D}_{A'B}\alpha^B\bigr)\bigl({\cal D}_{AB'}\bar
\alpha^{B'}\bigr)-h^{ef}t^{AA'}\bigl({\cal D}_e\alpha_A\bigr)\bigl({\cal D}_f
\bar\alpha_{A'}\bigr)-\frac{1}{2}t^aG_{aBB'}\alpha^B\bar\alpha^{B'}.
\end{equation*}
Here ${\cal D}_a:=P^b_a\nabla_b$ is the derivative operator of the so-called 
Sen connection, and, in terms of the intrinsic Levi-Civita derivative operator 
$D_e$ and the extrinsic curvature $\chi_{ab}$, its action on the spinor field 
$\alpha^A$ is given by ${\cal D}_e\alpha^A=D_e\alpha^A-\chi_e{}^A{}_{A'}t^{A'}
{}_B\alpha^B$. However, by Einstein's equations in the presence of a 
cosmological constant the last term on the right contains $\frac{1}{2}
\Lambda t_{AA'}\alpha^A\bar\alpha^{A'}$, whose integral on an infinite $\Sigma$ 
is \emph{a priori diverging} for spinor fields for which $t_{AA'}\alpha^A\bar
\alpha^{A'}$ does not fall off appropriately (e.g. when the spinor fields tend 
to a nonzero asymptotic value), independently of the fall-off properties of 
the matter fields. Therefore, we must `renormalize' the derivative operators 
in our integrals. 

Thus, following \cite{Gibbonsetal}, for any pair $(\alpha_A,\bar\beta_{A'})$ 
of spinor fields (or, equivalently, a Dirac spinor $\Psi^\alpha$ with Weyl 
spinor constituents $\alpha^A$ and $\bar\beta^{A'}$ and with $\alpha=A\oplus
A'$) we define the renormalized spacetime connection by 

\begin{equation}
\tilde\nabla_{AA'}\alpha_B:=\nabla_{AA'}\alpha_B+K\varepsilon_{AB}
 \bar\beta_{A'}, \hskip 20pt
\tilde\nabla_{AA'}\bar\beta_{B'}:=\nabla_{AA'}\bar\beta_{B'}+\tilde K
 \varepsilon_{A'B'}\alpha_A, \label{eq:4.2}
\end{equation}
for some complex constants $K$ and $\tilde K$. It is a straightforward 
calculation to show that the curvature of $\tilde\nabla_e$, acting on the 
bundle of Dirac spinors, is just the direct sum of $\tilde R^A{}_{Bcd}$, 
given by (\ref{eq:3.1}), and its complex conjugate $\bar{\tilde R}^{A'}{}
_{B'cd}$ precisely when $6K\tilde K=-\Lambda$. The corresponding 
renormalized Sen connection on the spacelike hypersurface $\Sigma$ is 

\begin{equation}
\tilde{\cal D}_{AA'}\alpha_B:={\cal D}_{AA'}\alpha_B+KP^{DD'}_{AA'}\varepsilon
 _{DB}\bar\beta_{D'}, \hskip 20pt
\tilde{\cal D}_{AA'}\bar\beta_{B'}:={\cal D}_{AA'}\bar\beta_{B'}+\tilde K
 P_{AA'}^{DD'}\varepsilon_{D'B'}\alpha_D. \label{eq:4.3}
\end{equation}
Denoting by tilde the quantities built from the renormalized connection 
(\ref{eq:4.2}), the Nester--Witten form is 

\begin{equation}
u(\alpha,\bar\alpha)_{ab}=\tilde u(\alpha,\bar\alpha)_{ab}+{\rm i}K
 \varepsilon_{AB}\bar\alpha_{(A'}\bar\beta_{B')}, \hskip 20pt
u(\beta,\bar\beta)_{ab}=\tilde u(\beta,\bar\beta)_{ab}+{\rm i}\bar{\tilde K}
 \varepsilon_{AB}\bar\alpha_{(A'}\bar\beta_{B')}. \label{eq:4.4}
\end{equation}
Note that here $\tilde u(\alpha,\bar\alpha)_{cd}$ already depends on $\bar
\beta_{A'}$, and $\tilde u(\beta,\bar\beta)_{cd}$ on $\alpha_A$, too. Then, a 
straightforward calculation yields that 

\begin{eqnarray}
\oint_{\partial\Sigma}\bigl(\tilde u(\alpha,\bar\alpha)_{ab}+\tilde u(\beta,
  \bar\beta)_{ab}\bigr)=\int_\Sigma\Bigl(B(\alpha)\!\!\!\!&+\!\!\!\!&B(\beta)
  \label{eq:4.5} \\
+3(K\!\!\!\!&+\!\!\!\!&\bar{\tilde K})t^{AA'}\bigl(\tilde K\alpha_A
  \bar\alpha_{A'}+\bar K\beta_A\bar\beta_{A'}\bigr) \nonumber \\
-2(K\!\!\!\!&+\!\!\!\!&\bar{\tilde K})t^{AA'}\varepsilon^{B'D'}\bigl(\bar
  \beta_{A'}\overline{\tilde{\cal D}}_{AB'}\bar\alpha_{D'}+\bar\alpha_{A'}
  \tilde{\cal D}_{AB'}\bar\beta_{D'}\bigr)\Bigr){\rm d}\Sigma, \nonumber
\end{eqnarray}
where ${\rm d}\Sigma$ is the natural volume element on $\Sigma$, and 
$B(\alpha)+B(\beta)$ in terms of $\tilde{\cal D}_e$ is 

\begin{eqnarray}
B(\alpha)+B(\beta)=\!\!\!\!&-\!\!\!\!&t^{AA'}h^{ef}
 \bigl((\tilde{\cal D}_e\alpha_A)(\overline{\tilde{\cal D}}_f\bar\alpha_{A'})+
 (\overline{\tilde{\cal D}}_e\beta_A)(\tilde{\cal D}_f\bar\beta_{A'})\bigr)
 \nonumber \\
\!\!\!\!&+\!\!\!\!&\frac{1}{2}\varkappa t^aT_{ab}(\alpha^B\bar\alpha^{B'}+
 \beta^B\bar\beta^{B'}) \nonumber \\
\!\!\!\!&+\!\!\!\!&\frac{1}{2}(\Lambda-6K\bar K)t_{AA'}\beta^A\bar\beta^{A'}
 +\frac{1}{2}(\Lambda-6\tilde K\bar{\tilde K})t_{AA'}\alpha^A\bar\alpha^{A'}
 \nonumber \\
\!\!\!\!&-\!\!\!\!&2t^{AA'}\varepsilon^{BD}\varepsilon^{B'D'}\Bigl(\bigl(\tilde
 {\cal D}_{A'B}\alpha_D\bigr)\bigl(\overline{\tilde{\cal D}}_{AB'}\bar\alpha_{D'}
 \bigr)+\bigl(\overline{\tilde{\cal D}}_{A'B}\beta_D\bigr)\bigl(\tilde{\cal D}
 _{AB'}\bar\beta_{D'}\bigr)\Bigr) \nonumber \\
\!\!\!\!&+\!\!\!\!&4\Bigl(K\bar\beta_{A'}t^{A'A}\varepsilon^{B'D'}\bigl(
 \overline{\tilde{\cal D}}_{AB'}\bar\alpha_{D'}\bigr)+\bar K\beta_At^{AA'}
 \varepsilon^{BD}\bigl(\tilde{\cal D}_{A'B}\alpha_D\bigr) \nonumber \\
\!\!\!\!&{}\!\!\!\!&\hskip 8pt +\tilde K\alpha_At^{AA'}\varepsilon^{BD}\bigl(
 \overline{\tilde{\cal D}}_{A'B}\beta_D\bigr)+\bar{\tilde K}\bar\alpha_{A'}t^{A'A}
 \varepsilon^{B'D'}\bigl(\tilde{\cal D}_{AB'}\bar\beta_{D'}\bigr)\Bigr),
 \label{eq:4.5B}
\end{eqnarray}
where we used Einstein's equations. 

Next, following \cite{Gibbonsetal}, we require the spinor fields $\alpha_A$, 
$\beta_A$ on $\Sigma$ to solve Witten's gauge condition (with still 
unspecified boundary conditions) with the renormalized connection: 

\begin{equation}
\varepsilon^{AB}\tilde{\cal D}_{A'A}\alpha_B=0, \hskip 20pt
\varepsilon^{A'B'}\tilde{\cal D}_{AA'}\bar\beta_{B'}=0. \label{eq:4.6}
\end{equation}
In this gauge the last three lines of (\ref{eq:4.5B}) and the third line of 
(\ref{eq:4.5}) are vanishing, and hence the integral of the renormalized 
Nester--Witten forms is 

\begin{eqnarray}
\oint_{\partial\Sigma}\Bigl(\tilde u(\alpha,\bar\alpha)_{ab}\!\!\!\!&+\!\!\!\!&
 \tilde u(\beta,\bar\beta)_{ab}\Bigr)=\int_\Sigma\Bigl\{-t^{AA'}h^{ef}\bigl(
 (\tilde{\cal D}_e\alpha_A)(\overline{\tilde{\cal D}}_f\bar\alpha_{A'})+
 (\overline{\tilde{\cal D}}_e\beta_A)(\tilde{\cal D}_f\bar\beta_{A'})\bigr)
 \nonumber \\
\!\!\!\!&+\!\!\!\!&\frac{1}{2}\varkappa t^aT_{ab}(\alpha^B\bar\alpha^{B'}+
 \beta^B\bar\beta^{B'}) \nonumber \\
\!\!\!\!&+\!\!\!\!&\frac{1}{2}(\Lambda+6\bar K\bar{\tilde K})t_{AA'}\beta^A
 \bar\beta^{A'}+\frac{1}{2}(\Lambda+6K\tilde K)t_{AA'}\alpha^A\bar\alpha^{A'}
 \Bigr\}{\rm d}\Sigma. \label{eq:4.7}
\end{eqnarray}
To kill the last two (\emph{a priori} diverging) terms of the integrand, we 
should require that $6K\tilde K=-\Lambda$, just the condition that we already 
obtained above. A particularly convenient choice (that we make) is that 
$\tilde K=K$ and $6K^2=-\Lambda$. Then, if the spinor fields solve the 
renormalized Witten type gauge condition (\ref{eq:4.6}) and the matter 
fields satisfy the dominant energy condition, then the integral of the 
renormalized Nester-Witten form is non-negative, and not \emph{a priori} 
diverging. 

Since $\Lambda$ is positive and hence $K$ is imaginary, by (\ref{eq:4.4}) 
$u(\alpha,\bar\alpha)_{ab}+u(\beta,\bar\beta)_{ab}=\tilde u(\alpha,\bar\alpha)
_{ab}+\tilde u(\beta,\bar\beta)_{ab}$ holds. Thus, under the same conditions, 
the left hand side of (\ref{eq:4.7}) would be finite even if it were built 
from the `un-renormalized' connection. We use this observation in subsection 
\ref{sub-4.6}. 

Our aim is to find the boundary conditions for the spinor fields $\alpha_A$ 
and $\beta_A$ on the cut of the conformal boundary such that 
(i.) the renormalized Witten type gauge conditions (\ref{eq:4.6}) admit 
a non-trivial solution on the spacelike hypersurface with the cut as its 
boundary, and 
(ii.) the resulting integrals could be interpreted as the components of a 
finite, real energy-momentum 4-vector, or at least could yield an invariant 
that could be interpreted as the mass. 

The finiteness of the energy-momentum 4-vector can be ensured by the 
requirement that the integral of the various terms on the right hand side 
of (\ref{eq:4.7}) be finite. In particular, we need to specify the fall-off 
properties of the physical energy-momentum tensor, the induced metric and 
the extrinsic curvature on the hypersurface $\Sigma$. We carry out these 
investigations in a coordinate system near the asymptotic end of $\Sigma$ 
that is analogous to the Bondi type coordinates in asymptotically flat 
spacetimes.

\subsection{Bondi type coordinates near $\mathscr{I}^+$}
\label{sub-4.2}

Let ${\cal S}:=\Sigma\cap\mathscr{I}^+\approx S^2$, the cut of the future 
conformal boundary defined by the hypersurface $\Sigma$, and let ${\cal S}
_u$, $u\in(-1,1)$, be a foliation of a neighbourhood of ${\cal S}$ in 
$\mathscr{I}^+$ by smooth topological 2-spheres such that ${\cal S}_0={\cal
S}$ and $u$ is increasing in the \emph{inward} direction. We define the 
lapse $n$ of this foliation in the usual way by $1=:-nV^a\nabla_au$, where 
$V^a$, as in section \ref{sec-3}, is the \emph{outward} pointing unit 
normal of the surfaces ${\cal S}_u$; and introduce the `evolution vector 
field' by $(\partial/\partial u)^a:=-nV^a$. Thus, in particular, its shift 
part will be chosen to be vanishing. Let $x^\mu=(x^2,x^3)$ be a local 
coordinate system on ${\cal S}_0\approx S^2$ (e.g. the $(\zeta,\bar\zeta)$ 
complex stereographic or the $(\theta,\phi)$ angle coordinates), and extend 
them from ${\cal S}_0$ to the other surfaces of the foliation by $V^a\nabla_a
x^\mu=0$. Thus we obtained a local coordinate system $(u,x^\mu)$ on a 
neighbourhood of ${\cal S}$ in $\mathscr{I}^+$. Since the shift part of the 
evolution vector field was chosen to be vanishing, the metric induced from 
$g_{ab}$ on the conformal boundary $\mathscr{I}^+$ (`boundary metric') takes 
the form $db^2=-A^2du^2+q_{\mu\nu}dx^\mu dx^\nu$, where $A=A(u,x^\mu)$ is a 
strictly positive function and $q_{\mu\nu}=q_{\mu\nu}(u,x^\rho)$ is negative 
definite. The freedom in the definition of this coordinate system is to 
choose a different foliation ${\cal S}_{\tilde u}$ in a neighbourhood of 
${\cal S}$, e.g. by the level sets of a new function $\tilde u:=Hu$, where 
$H$ is some strictly positive function on this neighbourhood, and to choose 
different coordinates $(x^2,x^3)$ on ${\cal S}$. 

Next we complete $(u,x^\mu)$ to be a local coordinate system on a 
neighbourhood of the `asymptotic end' of $\Sigma$ in $M$. Thus, let us fix 
$(u,x^\mu)$ and let ${\cal N}_u$ denote the past directed, ingoing null 
hypersurface emanated from the 2-surface ${\cal S}_u$. In a neighbourhood of 
$\mathscr{I}^+$ this is smooth, generated by past directed ingoing null 
geodesics $\gamma$ with future end points on ${\cal S}_u$ with coordinates 
$x^\mu$, and yields an extension of $u$ from $\mathscr{I}^+$ to a 
neighbourhood of $\mathscr{I}^+$ in $M$. Let us define $l^a:=g^{ab}\nabla_bu$, 
which is a future directed null normal of the hypersurfaces ${\cal N}_u$ and, 
by $l^a\nabla_al_b=l^a\nabla_bl_a=0$, it is a tangent of the \emph{affine 
parametrized} null geodesic generators $\gamma$. Let $w$ denote the affine 
parameter measured from ${\cal S}_u$, and hence we write $l^a=-(\partial/
\partial w)^a$. (Note that $w$ is \emph{decreasing} in the future direction.) 
Then the coordinates of a point $p$ in a neighbourhood of ${\cal S}_0$ in $M$ 
are defined to be $u$, $w$ and $x^\mu$ if $p=\gamma(w)$ and the coordinates 
of the future end point of $\gamma$ on $\mathscr{I}^+$ are $(u,x^\mu)$. The 
resulting coordinate system $(u,w,x^\mu)$ is analogous to that used to 
analyze asymptotically flat spacetimes near their future null infinity (see 
e.g. \cite{NT80}). 

Since by the definitions $\mathscr{I}^+=\{\Omega=0\}=\{w=0\}$ holds and 
$\Omega$ is smooth, we can write that $\Omega=aw+bw^2+O(w^3)$ for some 
functions $a$ and $b$ on $\mathscr{I}^+$, where $a$ is positive. Substituting 
this expression into $\nabla_a\nabla_b\Omega\approx0$ (see equation 
(\ref{eq:2.2.2b})), we obtain 

\begin{equation*}
\bigl(\nabla_aw\bigr)\bigl(\nabla_ba\bigr)+\bigl(\nabla_aa\bigr)\bigl(
\nabla_bw\bigr)+2b\bigl(\nabla_aw\bigr)\bigl(\nabla_bw\bigr)\approx0.
\end{equation*}
Contracting this with $(\partial/\partial w)^a(\partial/\partial w)^b$ and 
using that $(\partial/\partial w)^a$ is the tangent of affine parametrized 
geodesics and that $a$ does not depend on $w$, we find that $b=0$. 

However, it is only the \emph{conformal class} of $db^2$ that is physically 
determined. Thus we can use the conformal gauge freedom $g_{ab}\mapsto\omega^2
g_{ab}$ mentioned in subsection \ref{sub-2.1} with $\omega=a^{-1}$ to yield 
$\Omega=w+O(w^3)$. Hence, with this choice $\Omega$ is an asymptotic affine 
parameter in the unphysical metric along the null geodesic generators 
$\gamma$ even in the first two orders. Thus, we have fixed the conformal 
factor on $\mathscr{I}^+$, i.e. the remaining conformal gauge freedom is the 
rescaling of $g_{ab}$ with conformal factors of the form $1+\Omega^2\Theta$. 
We call this conformal gauge a \emph{special conformal Bondi gauge}. In the 
rest of the paper we assume that in our unphysical spacetime $(M,g_{ab},
\Omega)$ the conformal gauge is such a special conformal Bondi gauge. 

(Since the surfaces ${\cal S}_u$ in $\mathscr{I}^+$ are homeomorphic to 
2-spheres, there exists a conformal factor $R=R(u,x^\rho)$ such that $q
_{\mu\nu}(u,x^\rho)=R^2(u,x^\rho){}_0q_{\mu\nu}(u,x^\rho)$, where ${}_0
q_{\mu\nu}(u,x^\rho)$ are the components of the unit sphere metric on 
${\cal S}_u$ in the coordinates $(x^2,x^3)$. However, note that although
${}_0q_{\mu\nu}dx^\mu dx^\nu$ is a unit sphere metric, its components 
${}_0q_{\mu\nu}$ take some simple special form, e.g. ${\rm diag}(-1,-\sin^2
\theta)$ in the angle coordinates, \emph{only on a single surface}, e.g. 
on ${\cal S}_0$. In fact, $(x^2,x^3)$ were specified freely only on 
${\cal S}_0$, but their extension to the other surfaces was fixed essentially 
by the requirement of the vanishing of the $du\,dx^\mu$ components in the 
line element $db^2$. Hence, in general, the components ${}_0q_{\mu\nu}$ of 
the unit sphere metric depend on the coordinate $u$, too.) 

In the coordinate system $(u,w,x^\mu)$, 

\begin{eqnarray*}
&{}&g_{01}:=g\bigl(\frac{\partial}{\partial u},\frac{\partial}{\partial w}
  \bigr)=-\bigl(\frac{\partial}{\partial u}\bigr)^a\nabla_au=-1, \\
&{}&g_{11}:=g\bigl(\frac{\partial}{\partial w},\frac{\partial}{\partial w}
  \bigr)=g_{ab}l^al^b=0, \\
&{}&g_{1\mu}:=g\bigl(\frac{\partial}{\partial w},\frac{\partial}{\partial
  x^\mu}\bigr)=-\bigl(\frac{\partial}{\partial x^\mu}\bigr)^a\nabla_a
  u=0.
\end{eqnarray*}
Thus, in these coordinates, the form of the conformal metric is 

\begin{equation}
ds^2=g_{00}du^2-2du\,dw+2g_{0\mu}du\,dx^\mu+g_{\mu\nu}dx^\mu\,dx^\nu.
\label{eq:4.8}
\end{equation}
Comparing its pull back to $\mathscr{I}^+=\{\,w=0\,\}$ with $db^2$ we see that 

\begin{equation}
g_{00}+A^2,\,\,\, g_{0\mu},\,\,\, g_{\mu\nu}-q_{\mu\nu}=O(w).
\label{eq:4.9}
\end{equation}
The contravariant form of the unphysical metric is 

\begin{equation}
g^{ab}=\left(\begin{array}{ccc}
           0 & -1 & 0\\
          -1 & -g_{00}+g^{\nu\rho}g_{\nu0}g_{\rho0} & g^{\mu\rho}g_{\rho0}\\
           0 & g^{\mu\rho}g_{\rho0} & g^{\mu\nu}\\
\end{array}\right), \label{eq:4.10}
\end{equation}
where $g^{\mu\nu}$ is the inverse of the $2\times2$ matrix $g_{\mu\nu}$, 
i.e. defined by $g^{\mu\rho}g_{\rho\nu}:=\delta^\mu_\nu$. Since $\Lambda/3
\approx(\nabla_a\Omega)(\nabla^a\Omega)$ holds by equation (\ref{eq:2.2.2c}), 
the component $g^{11}$ of the conformal metric is asymptotically constant: 
$g^{11}=\frac{1}{3}\Lambda+O(w^2)$. Comparing this with (\ref{eq:4.10}) and 
taking into account that $g_{\mu0}=O(w)$ (see equation (\ref{eq:4.9})) we 
find that 

\begin{equation}
A^2=\frac{1}{3}\Lambda, \hskip 20pt g_{00}=-\frac{1}{3}\Lambda+O(w^2).
\label{eq:4.11}
\end{equation}
Moreover, contracting $0\approx\nabla_a\nabla_b\Omega\approx\nabla_a\nabla_bw$ 
with the various coordinate vectors, for the Christoffel symbols we obtain 
that $\Gamma^1_{ab}=O(w)$, where $a,b=0,...,3$. (N.B.: $\Gamma^a_{11}=0$ 
identically, because the generators of ${\cal N}_u$ are affine parametrized 
geodesics.) In particular, 

\begin{eqnarray}
O(w)\!\!\!\!&=\!\!\!\!&2\Gamma^1_{0\mu}=-\partial_\mu g_{00}
  -g^{11}\partial_1g_{0\mu}+g^{1\rho}\bigl(-\partial_\rho g_{0\mu}+
  \partial_0g_{\rho\mu}+\partial_\mu g_{0\rho}\bigr) \nonumber \\
\!\!\!\!&=\!\!\!\!&-\frac{1}{3}\Lambda\partial_1g_{0\mu}+O(w),
  \label{eq:4.12a}\\
O(w)\!\!\!\!&=\!\!\!\!&2\Gamma^1_{\mu\nu}=\partial_0g_{\mu\nu}-
  \partial_\mu g_{0\nu}-\partial_\nu g_{0\mu}-g^{11}\partial_1g_{\mu\nu}+g^{1\rho}
  \bigl(-\partial_\rho g_{\mu\nu}+\partial_\mu g_{\rho\nu}+\partial_\nu g_{\rho
  \mu}\bigr) \nonumber \\
\!\!\!\!&=\!\!\!\!&\partial_0g_{\mu\nu}-\frac{1}{3}\Lambda\partial_1
  g_{\mu\nu}+O(w). \label{eq:4.12b}
\end{eqnarray}
The first implies that $\partial_1g_{0\mu}=O(w)$. To evaluate the second, 
let us write $g_{\mu\nu}=q_{\mu\nu}+r_{\mu\nu}w+O(w^2)$, where $r_{\mu\nu}=
r_{\mu\nu}(u,x^\rho)$ (see equation (\ref{eq:4.9})). Substituting this into 
(\ref{eq:4.12b}) we find that $\partial_0q_{\mu\nu}=\frac{1}{3}\Lambda
r_{\mu\nu}$. Thus, from five of the equations $\nabla_a\nabla_b\Omega\approx0$ 
we obtain that 

\begin{equation}
g_{0\mu}=O(w^2), \hskip 20pt
g_{\mu\nu}=q_{\mu\nu}+\frac{3}{\Lambda}\bigl(\partial_0q_{\mu\nu}\bigr)w+O(w^2).
\label{eq:4.13}
\end{equation}
Equations (\ref{eq:4.11}) and (\ref{eq:4.13}) provide a refinement of the 
fall-off properties (\ref{eq:4.9}). The second of (\ref{eq:4.13}) implies 
that $g^{\mu\nu}=q^{\mu\nu}+\frac{3}{\Lambda}(\partial_0q^{\mu\nu})w+O(w^2)$, 
where $q^{\mu\nu}$ is the inverse of $q_{\mu\nu}$. The remaining components of 
$\nabla_a\nabla_b\Omega\approx0$ do not yield any further restriction on 
the asymptotic form of the metric. 

Introducing the new coordinate $r:=1/w$, the asymptotic form of the physical 
spacetime metric $d\hat s^2=\Omega^{-2}ds^2$ is 

\begin{eqnarray}
d\hat s^2=\!\!\!\!&-\!\!\!\!&\frac{1}{3}\Lambda r^2du^2+2du\,dr+r^2\Bigl(
  q_{\mu\nu}+\frac{3}{\Lambda}\bigl(\partial_0q_{\mu\nu}\bigr)\frac{1}{r}
  \Bigr)dx^\mu dx^\nu \nonumber \\
\!\!\!\!&+\!\!\!\!&O(1)du^2+O(\frac{1}{r^2})du\,dr+O(1)du\,dx^\mu+O(1)
dx^\mu\,dx^\nu.
\label{eq:4.14}
\end{eqnarray}
In particular, the metric of the de Sitter spacetime also has this form. To 
see this, first let us rewrite the line element of the Einstein universe, 
$ds_E^2=d\tau^2-d\bar r^2-\sin^2\bar r\, d\omega^2$, in the new coordinates 

\begin{equation*}
u:=\sqrt{\frac{3}{\Lambda}}\bigl(\tau-\bar r\bigr), \hskip 20pt
w:=\sqrt{\frac{\Lambda}{3}}\bigl(\tau_0-\tau\bigr).
\end{equation*}
(Here $\tau_0$ is some constant and $d\omega^2:=d\theta^2+\sin^2\theta d\phi
^2=4(1+\zeta\bar\zeta)^{-2}d\zeta d\bar\zeta$, the line element of the unit 
sphere metric in the angle and the complex null coordinates, respectively, 
and $\bar r\in[0,\pi]$, $\tau\in\mathbb{R}$.) We obtain 

\begin{eqnarray*}
ds_E^2=\!\!\!\!&-\!\!\!\!&\frac{1}{3}\Lambda du^2-2du\,dw-\sin^2\Bigl(
  \sqrt{\frac{\Lambda}{3}}u-\tau_0+\sqrt{\frac{3}{\Lambda}}w\Bigr)d\omega^2\\
=\!\!\!\!&-\!\!\!\!&\frac{1}{3}\Lambda du^2-2du\,dw-\sin^2\Bigl(
  \sqrt{\frac{\Lambda}{3}}u-\tau_0\Bigr)d\omega^2-\frac{3}{\Lambda}
  \frac{\partial}{\partial u}\Bigl(\sin^2\Bigl(\sqrt{\frac{\Lambda}{3}}u-
  \tau_0\Bigr)d\omega^2\Bigr)w\\
\!\!\!\!&+\!\!\!\!&O(w^2),
\end{eqnarray*}
which has the expected asymptotic form near the $w=0$ hypersurface. Thus, to 
see that the de Sitter metric has indeed the form (\ref{eq:4.14}), it is 
already enough to recall that the de Sitter spacetime is conformal e.g. to 
the $-\pi/2<\tau<\pi/2$ part of the Einstein universe. Its future conformal 
boundary $\mathscr{I}^+=\{\tau=\pi/2\}$ coincides with the $w=0$ hypersurface 
precisely when $\tau_0=\pi/2$, in which case $u=\sqrt{\frac{3}{\Lambda}}
\frac{\pi}{2}$, $0$ and $-\sqrt{\frac{3}{\Lambda}}\frac{\pi}{2}$ correspond, 
respectively, to the $\bar r=0$ origin, the $\bar r=\pi/2$ maximal 2-surface 
and the $\bar r=\pi$ anti-podal point; while the conformal factor is $\Omega
=\sqrt{\frac{\Lambda}{3}}\sin(\tau+\pi/2)=\sqrt{\frac{\Lambda}{3}}\sin(\pi-
\sqrt{\frac{3}{\Lambda}}w)=w+O(w^3)$.

\subsection{The Newman--Penrose tetrad}
\label{sub-4.3}

Let ${\cal S}_{u,w}$ denote the $w={\rm const}$ topological 2-sphere in the 
null hypersurface ${\cal N}_u$, and adapt a Newman--Penrose (NP) complex null 
tetrad $\{l^a,n^a,m^a,\bar m^a\}$ to these surfaces: First, recall that $l^a$ 
is an outgoing null normal to ${\cal S}_{u,w}$, and let us choose $n^a$ to be 
the future pointing ingoing null normal and normalized by $n^al_a=1$; and let 
us choose $m^a$ to be a complex null tangent to ${\cal S}_{u,w}$, $\bar m^a$ 
to be its complex conjugate and they are normalized according to $m^a\bar m_a
=-1$. This basis is fixed up to the change $m^a\mapsto e^{{\rm i}\alpha}m^a$ for 
any real function $\alpha$ of the coordinates. Since $l_a=\nabla_au$, $n^al_a
=1$, $m^a\nabla_au=0$ and $m^a\nabla_aw=0$, the vectors of the tetrad have 
the form 

\begin{equation}
l^a=-\bigl(\frac{\partial}{\partial w}\bigr)^a, \hskip 20pt
n^a=\bigl(\frac{\partial}{\partial u}\bigr)^a+B\bigl(\frac{\partial}
{\partial w}\bigr)^a+C^\mu\bigl(\frac{\partial}{\partial x^\mu}\bigr)^a,
\hskip 20pt
m^a=D^\mu\bigl(\frac{\partial}{\partial x^\mu}\bigr)^a \label{eq:4.15}
\end{equation}
for some functions $B$, $C^\mu$ and $D^\mu$ of the coordinates. Here $B$ 
is real and, if $(x^2,x^3)$ are complex, then $C^2=\bar C^3$ holds. Comparing 
$g^{ab}=l^an^b+n^al^b-m^a\bar m^b-\bar m^am^b$ with (\ref{eq:4.10}) we see 
that 

\begin{equation}
B=-\frac{1}{6}\Lambda+O(w^2), \hskip 20pt
C^\mu=-g^{\mu\rho}g_{\rho 0}=O(w^2), \hskip 20pt
D^\mu\bar D^\nu+\bar D^\mu D^\nu=g^{\mu\nu}. \label{eq:4.16}
\end{equation}
Since the future pointing unit timelike normal of $\mathscr{I}^+$ is $N^a=-
\sqrt{3/\Lambda}g^{ab}\nabla_bw$, its contraction with $l_a$ is constant on 
the whole $\mathscr{I}^+$: $N^al_a=\sqrt{3/\Lambda}$. This implies that on 
$\mathscr{I}^+$ 

\begin{equation}
l_a=\sqrt{\frac{3}{\Lambda}}\bigl(N_a+V_a\bigr), \hskip 20pt
n_a=\sqrt{\frac{\Lambda}{12}}\bigl(N_a-V_a\bigr). \label{eq:4.17}
\end{equation}
Thus, the real null normals $l_a$ and $n_a$ of the 2-surface ${\cal S}$ are 
boosted with respect to $\frac{1}{\sqrt{2}}(N_a\pm V_a)$, the ones built 
from the timelike and spacelike normals \emph{symmetrically}. Hence, \emph{we 
adapt our normalized GHP spinor dyad $\{o^A,\iota^A\}$ to the un-boosted null 
normals} $\frac{1}{\sqrt{2}}(N_a\pm V_a)$, i.e. $o^A\bar o^{A'}=\sqrt{\Lambda
/6}l^a$, $\iota^A\bar\iota^{A'}=\sqrt{6/\Lambda}n^a$, $o^A\bar\iota^{A'}=m^a$ 
and $\iota^A\bar o^{A'}=\bar m^a$. Also, \emph{we define the GHP spin 
coefficients \cite{GHP} in this un-boosted frame}, rather than in $\{l^a,n^a,
m^a,\bar m^a\}$. 

To find the asymptotic form of the spin coefficients in the \emph{physical} 
spacetime, we conformally rescale the tetrad $\{l^a,n^a,m^a,\bar m^a\}$. Since, 
however, neither $l^a$ nor $n^a$ is distinguished physically over the other, 
we rescale them, and the spinor dyad also, symmetrically: 

\begin{eqnarray}
&{}&l^a=\Omega^{-1}\hat l^a, \hskip 25pt
n^a=\Omega^{-1}\hat n^a, \hskip 25pt
m^a=\Omega^{-1}\hat m^a;  \label{eq:4.18a} \\
&{}&o^A=\Omega^{-\frac{1}{2}}\hat o^A, \hskip 15pt
\iota^A=\Omega^{-\frac{1}{2}}\hat\iota^A. \label{eq:4.18b}
\end{eqnarray}
These imply the general formulae how the various GHP spin coefficients 
change under such a symmetric rescaling: 

\begin{eqnarray*}
&{}&\hat\kappa=\Omega\kappa, \hskip 97pt
  \hat\kappa^\prime=\Omega\kappa^\prime, \\
&{}&\hat\sigma=\Omega\sigma, \hskip 97pt
  \hat\sigma^\prime=\Omega\sigma^\prime, \\
&{}&\hat\rho=\Omega\rho+o^A\bar o^{A'}\nabla_a\Omega, \hskip 35pt
  \hat\rho^\prime=\Omega\rho^\prime+\iota^A\bar\iota^{A'}\nabla_a\Omega, \\
&{}&\hat\tau=\Omega\tau+o^A\bar\iota^{A'}\nabla_a\Omega, \hskip35pt
  \hat\tau^\prime=\Omega\tau^\prime+\iota^A\bar o^{A'}\nabla_a\Omega, \\
&{}&\hat\beta=\Omega\beta-\frac{1}{2}o^A\bar\iota^{A'}\nabla_a\Omega,
  \hskip 25pt
  \hat\beta^\prime=\Omega\beta^\prime-\frac{1}{2}\iota^A\bar o^{A'}\nabla_a
  \Omega, \\
&{}&\hat\varepsilon=\Omega\varepsilon-\frac{1}{2}o^A\bar o^{A'}\nabla_a\Omega,
  \hskip 30pt
  \hat\varepsilon^\prime=\Omega\varepsilon^\prime-\frac{1}{2}\iota^A\bar\iota
  ^{A'}\nabla_a\Omega.
\end{eqnarray*}
Since $l_a$ is a gradient and $n_a$ is hypersurface orthogonal, certain GHP 
spin coefficients take special value: $\kappa=0$, $\bar\rho=\rho$, $\bar\rho
^\prime=\rho^\prime$, $\varepsilon+\bar\varepsilon=0$ and $\tau=\beta-\bar\beta
^\prime$ hold. In addition, by (\ref{eq:4.17}) at $\mathscr{I}^+$ 

\begin{equation*}
\frac{1}{\sqrt{2}}\Bigl(\nabla_a\bigl(o_B\bar o_{B'}\bigr)+\nabla_a\bigl(
\iota_B\bar\iota_{B'}\bigr)\Bigr)=\nabla_aN_b=0,
\end{equation*}
whose contraction with the various tetrad vectors yields that 

\begin{equation*}
\kappa^\prime, \,\,\,
\tau^\prime, \,\,\,
\tau=\beta-\bar\beta^\prime, \,\,\,
\varepsilon^\prime+\bar\varepsilon^\prime, \,\,\,
\sigma^\prime+\bar\sigma, \,\,\,
\rho^\prime+\rho=O(w).
\end{equation*}
Hence, by $\Omega=1/r+O(1/r^3)$, the asymptotic form of the GHP spin 
coefficients in the physical spacetime is 

\begin{eqnarray}
\hat\kappa\!\!\!\!&=\!\!\!\!&0, \hskip 140pt
  \hat\kappa^\prime=O(\frac{1}{r^2}), \label{eq:4.19a} \\
\hat\sigma\!\!\!\!&=\!\!\!\!&\frac{1}{r}\sigma^0+O(\frac{1}{r^2}),
  \hskip 78pt
  \hat\sigma^\prime=-\frac{1}{r}\bar\sigma^0+O(\frac{1}{r^2}),
  \label{eq:4.19b} \\
\hat\rho\!\!\!\!&=\!\!\!\!&-\sqrt{\frac{\Lambda}{6}}+\frac{1}{r}\rho^0+
  O(\frac{1}{r^2}), \hskip 35pt
  \hat\rho^\prime=-\sqrt{\frac{\Lambda}{6}}-\frac{1}{r}\rho^0+O(\frac{1}{r^2}),
  \label{eq:4.19c} \\
\hat\tau\!\!\!\!&=\!\!\!\!&O(\frac{1}{r^2}), \hskip 115pt
  \hat\tau^\prime=O(\frac{1}{r^2}), \label{eq:4.19d} \\
\hat\beta+\bar{\hat\beta}^\prime\!\!\!\!&=\!\!\!\!&\frac{2}{r}\beta^0+O(
  \frac{1}{r^2}), \hskip 60pt
\hat\beta-\bar{\hat\beta}^\prime=O(\frac{1}{r^2}), \label{eq:4.19e} \\
\hat\varepsilon+\bar{\hat\varepsilon}\!\!\!\!&=\!\!\!\!&\sqrt{\frac{\Lambda}
  {6}}+O(\frac{1}{r^2}), \hskip 65pt
\hat\varepsilon-\bar{\hat\varepsilon}=\frac{2{\rm i}}{r}\varepsilon^0+
  O(\frac{1}{r^2}), \label{eq:4.19f} \\
\hat\varepsilon^\prime+\bar{\hat\varepsilon}^\prime\!\!\!\!&=\!\!\!\!&
  \sqrt{\frac{\Lambda}{6}}+O(\frac{1}{r^2}), \hskip 65pt
  \hat\varepsilon^\prime-\bar{\hat\varepsilon}^\prime=\frac{2{\rm i}}{r}
  \varepsilon^{\prime\, 0}+O(\frac{1}{r^2}), \label{eq:4.19g}
\end{eqnarray}
where $\sigma^0$ and $\beta^0$ are complex, while $\rho^0$, $\varepsilon^0$ 
and $\varepsilon^{\prime\,0}$ are real functions on $\mathscr{I}^+$. The first 
three of them are linked to the intrinsic geometry of the conformal boundary. 
In fact, 

\begin{equation}
\sigma^0=\frac{1}{\sqrt{2}}m^a\bigl(\nabla_aV_b\bigr)m^b, \hskip 20pt
\rho^0=\frac{1}{\sqrt{2}}\bar m^a\bigl(\nabla_aV_b\bigr)m^b, \hskip 20pt
\beta^0=-\frac{1}{2}m^a\bigl(\nabla_am_b\bigr)\bar m^b \label{eq:4.19h}
\end{equation}
represent, respectively, the trace-free part of the extrinsic curvature of 
the 2-surfaces ${\cal S}_u$ in $\mathscr{I}^+$, the trace of this extrinsic 
curvature, and the rest of the connection 1-form of the intrinsic geometry 
of $\mathscr{I}^+$. On the other hand, on $\mathscr{I}^+$ 

\begin{equation}
\varepsilon^0=\frac{\rm i}{2}\sqrt{\frac{\Lambda}{6}}l^a\bigl(\nabla_am_b
  \bigr)\bar m^b, \hskip 20pt
\varepsilon^{\prime\,0}=\frac{\rm i}{2}\sqrt{\frac{6}{\Lambda}}n^a\bigl(
  \nabla_a\bar m_b\bigr)m^b, \label{eq:4.19i}
\end{equation}
which specify how the complex null vectors $m^a$ and $\bar m^a$ are extended 
off the conformal boundary. Since, however, the complex null vectors $m^a$ 
and $\bar m^a$ are fixed only up to the phase transformation $m^a\mapsto\exp(
{\rm i}\alpha)m^a$ with an arbitrary smooth function $\alpha=\alpha(u,w,
x^\mu)$, the functions $\varepsilon^0$ and $\varepsilon^{\prime 0}$ can be 
chosen to be vanishing on $\mathscr{I}^+$. In fact, since on a neighbourhood 
of $\mathscr{I}^+$ they are given by (\ref{eq:4.19i}) up to $O(\Omega)$ terms, 
under such a transformation they change according to $\varepsilon^0\mapsto
\varepsilon^0+\frac{1}{2}\sqrt{\frac{\Lambda}{6}}l^a\nabla_a\alpha$ and 
$\varepsilon^{\prime 0}\mapsto\varepsilon^{\prime 0}-\frac{1}{2}\sqrt{\frac{6}
{\Lambda}}n^a\nabla_a\alpha$. Thus, if $\alpha$ is chosen to be 

\begin{equation*}
\alpha\bigl(u,w,x^\mu\bigr)=2\sqrt{\frac{6}{\Lambda}}\int^w_0\varepsilon^0
\bigl(u,w',x^\mu\bigr)dw',
\end{equation*}
then the new $\varepsilon^0$ in the transformed frame is vanishing even on a 
neighbourhood of $\mathscr{I}^+$. The NP frame is still not fixed, phase 
transformations of the complex null vectors with $w$-independent phase are 
still allowed. Thus, if this phase in such a further transformation is chosen 
to be 

\begin{equation*}
\alpha\bigl(u,x^\mu\bigr)=2\sqrt{\frac{\Lambda}{6}}\int^u_0\varepsilon
^{\prime 0}\bigl(u',0,x^\mu\bigr)du',
\end{equation*}
then the new $\varepsilon^{\prime 0}$ in the transformed frame is vanishing on 
$\mathscr{I}^+$. Therefore, the $\varepsilon^0$ and $\varepsilon^{\prime 0}$ in 
(\ref{eq:4.19f}) and (\ref{eq:4.19g}), respectively, can be chosen to be 
vanishing. Note also that, by the first of (\ref{eq:4.19f}), $r$ is \emph{not} 
an affine parameter along the null geodesic generators of the null 
hypersurfaces ${\cal N}_u$ in the physical spacetime.

\subsection{The asymptotic properties of the geometry of $\Sigma$}
\label{sub-4.4}

To determine the asymptotic form of the solutions of the renormalized 
Witten equation, and also to find the appropriate function spaces in which the 
renormalized Witten equation can be proven to admit a solution, we need to 
know the detailed asymptotic structure of the spacelike hypersurface $\Sigma$. 
Thus, suppose that in a neighbourhood of $\mathscr{I}^+$ in $M$ the 
hypersurface $\Sigma$ is given by $u-U(w,x^\mu)=0$ for some smooth function 
$U$ of $w$ and $x^\mu$, where, by $\Sigma\cap\mathscr{I}^+={\cal S}_0$, $U(0,
x^\mu)=0$ holds. (More generally, a 1-parameter family $\Sigma_t$ of spacelike 
hypersurfaces that intersect $\mathscr{I}^+$ in the 2-surfaces ${\cal S}_u$ 
with $u=t$ is given by the level sets $t:=u-U(w,x^\mu)={\rm const}$.) Hence 
we can write $U(w,x^\mu)=Ww+O(w^2)$ for some smooth function $W$ of the 
coordinates $x^\mu$. Since $\Sigma$ is spacelike (even at $\mathscr{I}^+$) and 
the hypersurfaces ${\cal N}_u$ are null, the $u$ coordinate along $\Sigma$ 
must be increasing with increasing $w$, and hence $W$ must be strictly 
positive. Since the components of the normal $T_a$ of $\Sigma$ in the 
coordinate system $(u,w,x^\mu)$ are $(1,-\frac{\partial U}{\partial w},
-\frac{\partial U}{\partial x^\mu})$, its norm is $\vert T_e\vert^2:=g^{ab}
T_aT_b=\frac{1}{3}\Lambda W^2+2W+O(w)$, its scalar product with the unit 
normal of the $\Omega={\rm const}$ hypersurfaces is $N^aT_a=\sqrt{\frac{3}
{\Lambda}}(1+\frac{1}{3}\Lambda W)+O(w)$, and the function $W$ is completely 
determined by $\vert T_e\vert^2$ on $\mathscr{I}^+$. Note that $\Sigma$ would 
be asymptotically null (e.g. the null hypersurface ${\cal N}_0$ itself) 
precisely when $W$ were vanishing. Then it is straightforward to derive the 
asymptotic form of the induced physical metric $\hat h_{ab}=\Omega^{-2}h_{ab}$ 
on $\Sigma$. It is 

\begin{equation*}
d\hat h^2=-\frac{1}{r^2}\Bigl(\vert T_e\vert^2+O(\frac{1}{r})\Bigr)dr^2+
O(\frac{1}{r})dr\,dx^\mu+r^2\Bigl(q_{\mu\nu}+O(\frac{1}{r})\Bigr)dx^\mu\,dx^\nu.
\end{equation*}
If $\vert T_e\vert$ (and hence $W$, too) were constant and $q_{\mu\nu}$ were 
the unit sphere metric, then this would be just the asymptotic form of the 
standard hyperboloidal metric 

\begin{equation*}
d\hat h^2_H=-\frac{\vert T_e\vert^2}{\vert T_e\vert^2+r^2}dr^2-r^2d\omega^2
\end{equation*}
with constant curvature (and curvature scalar $\hat{\cal R}=-6/\vert T_e\vert
^2$). Therefore, the induced intrinsic metric on $\Sigma$ is some `deformed', 
or asymptotically hyperboloidal one, characterized asymptotically by the 
function $W$ and the 2-metric $q_{\mu\nu}$, in which $r$ is an asymptotic 
\emph{areal} (rather than a radial distance) coordinate. The function $W$ 
plays the role of the local boost parameters, characterizing the relative 
direction of the normal of $\Sigma$ with respect to that of $\mathscr{I}^+$ 
at $\mathscr{I}^+$. 

However, it seems useful to rewrite the induced metric in a slightly 
different, intrinsic coordinate system on $\Sigma$. Thus, let us foliate the 
asymptotic end of $\Sigma$ by the level sets of the conformal factor, $\hat
{\cal S}_{\Omega}:=\Sigma\cap\{\Omega={\rm const}\}$. In general, for $\Omega
>0$, these surfaces do not coincide with any ${\cal S}_{u,w}:={\cal N}_u\cap
\{w={\rm const}\}$, but in the $\Omega\rightarrow0$ limit $\hat{\cal S}
_{\Omega}\rightarrow{\cal S}_0\subset\mathscr{I}^+$. Let $v_a:=D_a\Omega/\vert
D_e\Omega\vert$, the $g_{ab}$-unit normal to $\hat{\cal S}_{\Omega}$ which is 
tangent to $\Sigma$, where $\vert D_e\Omega\vert^2:=-g^{ab}(D_a\Omega)(D_b
\Omega)$. This $v^a$ points `outward' to the conformal boundary, and the 
lapse $\tilde n$ of this foliation, defined by $1=:-\tilde nv^aD_a\Omega$, is 
just $\tilde n=1/\vert D_e\Omega\vert$. 

Let us complete this $v^a$ to be a frame field $\{v^a,M^a,\bar M^a\}$ on 
$\Sigma$. Here $M^a$ and $\bar M^a$ are complex null tangents of the surfaces 
$\hat{\cal S}_{\Omega}$, orthogonal to $v^a$, and normalized with respect to 
$g_{ab}$ by $M^a\bar M_a=-1$. A simple calculation yields that the $g_{ab}$-unit 
normal to $\Sigma$ and the vectors of this frame field on $\Sigma$ are given 
by 

\begin{eqnarray}
t^a\!\!\!\!&=\!\!\!\!&\frac{1}{\vert T_e\vert}\Bigl(\bigl(1+\frac{1}{6}
  \Lambda W\bigr)l^a+Wn^a\Bigr)+O(w), \label{eq:4.4.1a} \\
v^a\!\!\!\!&=\!\!\!\!&\frac{1}{\vert T_e\vert}\Bigl(\bigl(1+\frac{1}{6}
  \Lambda W\bigr)l^a-Wn^a\Bigr)+O(w), \label{eq:4.4.1b} \\
M^a\!\!\!\!&=\!\!\!\!&\exp({\rm i}\alpha)\,m^a+O(w)
 \label{eq:4.4.1c}
\end{eqnarray}
with an irrelevant phase $\alpha$, which will be chosen to be zero. By means 
of the last two it is straightforward to give the explicit form of the 
projection $P^a_b=-v^av_b-M^a\bar M_a-\bar M^aM_b=-v^av_b-m^a\bar m_a-\bar m^a
m_b+O(w)$. By (\ref{eq:4.4.1b}), (\ref{eq:4.15}) and (\ref{eq:4.16}) the 
integral curves of $v^a$ (with parameter $\tilde w$) in the coordinates 
$(u,w,x^\mu)$ are 

\begin{equation}
u(\tilde w)=-\frac{W}{\vert T_e\vert}\tilde w +O(\tilde w^2), \hskip 20pt
w(\tilde w)=-\frac{1}{\vert T_e\vert}\tilde w +O(\tilde w^2), \hskip 20pt
x^\mu(\tilde w)=x^\mu(0)+O(\tilde w^2); \label{eq:4.4.1d}
\end{equation}
and their end points on $\mathscr{I}^+$ are at $\tilde w=0$. Hence these 
integral curves define a diffeomorphism between ${\cal S}_0$ and the surfaces 
$\hat{\cal S}_\Omega$. Moreover, $\tilde w$ coincides with the affine parameter 
$w$ in the first order up to a scale transformation (though this scale factor 
depends on the coordinate $x^\mu$ of the end point of the integral curves on 
$\mathscr{I}^+$). Thus, the coordinates $x^\mu$ of the end points of the 
integral curves can be extended from ${\cal S}_0$ to the whole asymptotic end 
of $\Sigma$ by $v^aD_a\tilde x^\mu=0$ with the initial condition $\tilde x^\mu
\approx x^\mu$. By the third of (\ref{eq:4.4.1d}) $\tilde x^\mu=x^\mu+O(\tilde
w^2)$. Therefore, $(\Omega,\tilde x^\mu)$, or rather $(\tilde r,\tilde x^\mu)$ 
with $\tilde r:=1/\Omega$, form a coordinate system on the asymptotic end of 
$\Sigma$. The radial coordinate $\tilde r$ coincides with $r$ in the first 
two orders: $\tilde r=r+O(r^{-1})$. The coordinate vector $(\partial/\partial
\Omega)^a$ is just the lapse times of the unit normal of the surfaces, $-v^a/
\vert D_e\Omega\vert$, with vanishing `shift part'. But by (\ref{eq:4.4.1a}) 

\begin{eqnarray}
\vert D_e\Omega\vert^2\!\!\!\!&=\!\!\!\!&-g^{ab}(D_a\Omega)(D_b\Omega)=
 -g^{11}+t^1t^1+O(w^2) \nonumber \\
\!\!\!\!&=\!\!\!\!&-\frac{1}{3}\Lambda+\frac{1}{\vert T_e\vert^2}\bigl(1+
 \frac{1}{3}\Lambda W+O(w)\bigr)^2+O(w^2)=\frac{1}{\vert T_e\vert^2}+O(w),
\label{eq:4.4.DO}
\end{eqnarray}
i.e. the lapse is $\tilde n=\vert T_e\vert+O(w)$. Therefore, in these 
coordinates, $(\partial/\partial\Omega)^a=-\vert T_e\vert v^a+O(\Omega)$, and 
the asymptotic form of the induced physical metric $\hat h_{ab}$ is 

\begin{equation}
d\hat h^2=-\frac{1}{\tilde{r}^2}\Bigl(\vert T_e\vert^2+O(\frac{1}{\tilde
r})\Bigr)d\tilde{r}^2+\tilde{r}^2\Bigl(R^2\,{}_0q_{\mu\nu}+O(\frac{1}{\tilde r})
\Bigr)d\tilde x^\mu\,d\tilde x^\nu, \label{eq:4.4.2}
\end{equation}
where $R$ is the conformal factor such that ${}_0q_{\mu\nu}=R^{-2}q_{\mu\nu}$ is 
the unit sphere metric, and the coordinates $\tilde x^\mu$ can be chosen to 
be the familiar angle or the complex stereographic coordinates in which ${}_0
q_{\mu\nu}$ takes the standard form of the unit sphere metric (see subsection 
\ref{sub-4.2}). We use this form of the metric in appendix \ref{sub-A.2}, 
and this form of the coordinate vector $(\partial/\partial\Omega)^a$ in 
subsections \ref{sub-4.5.2} and \ref{sub-4.6.1}. In particular, the 
asymptotic form of the induced volume element on the hypersurface is ${\rm d}
\hat\Sigma=\tilde rR^2\vert T_e\vert\,{\rm d}{\cal S}_0\,{\rm d}\tilde r$, 
where ${\rm d}{\cal S}_0$ is the area element on the unit sphere. 

The most convenient way to calculate the extrinsic curvature of $\Sigma$, 
both in the conformal and in the physical spacetime, is the use of the 
family $\Sigma_t$ of hypersurfaces and the coordinates $(t,w,x^\mu)$. We 
obtain that 

\begin{equation}
\hat\chi_{ab}=\frac{1}{\vert T_e\vert}\bigl(1+\frac{1}{3}\Lambda W\bigr)\hat
h_{ab}+O(\frac{1}{\tilde r}). \label{eq:4.4.3}
\end{equation}
Thus the physical extrinsic curvature of $\Sigma$ is asymptotically 
proportional to its intrinsic physical metric, just like in the case of 
spacelike hypersurfaces extending to the future null infinity of 
asymptotically flat spacetimes. Therefore, the leading terms in the 
asymptotic form of the metric $\hat h_{ab}$ and of the extrinsic curvature 
$\hat\chi_{ab}$ on the single hypersurface $\Sigma$ are determined by two 
functions on the cut ${\cal S}=\mathscr{I}^+\cap\Sigma$: the boost 
`parameter' $W$ and the conformal factor $R$.

\subsection{The boundary conditions from the Witten equation}
\label{sub-4.5}

\subsubsection{The fall-off and the algebraic boundary conditions}
\label{sub-4.5.1}

The boundary conditions of the renormalized Witten equations 
(\ref{eq:4.6}), given explicitly by 

\begin{equation}
\hat{\cal D}_{A'A}\hat\alpha^A+\frac{3}{2}K\bar{\hat\beta}_{A'}=0,
\hskip 20pt
\hat{\cal D}_{AA'}\bar{\hat\beta}^{A'}+\frac{3}{2}K\hat\alpha_A=0,
\label{eq:4.5.1}
\end{equation}
consist of two parts. (Here $K=\pm{\rm i}\sqrt{\Lambda/6}$, see the text 
following equation (\ref{eq:4.7}). Though the sign can be fixed without loss 
of generality, we leave this ambiguity in the formalism. All the sign 
ambiguities in what follows come from this ambiguity.) The first is an 
appropriate fall-off condition specifying how fast the spinor fields tend to 
their own asymptotic value at infinity, while the second is a condition on 
the asymptotic values of the spinor fields. In the present subsection we 
determine the first, and the part of the second that comes from the equations 
(\ref{eq:4.5.1}) themselves. 

To find these conditions, we rewrite (\ref{eq:4.5.1}) in the unphysical 
spacetime. We associate zero conformal weight to the contravariant form of 
the spinor fields, i.e. $\alpha^A=\hat\alpha^A$, $\beta^A=\hat\beta^A$, and, 
for the sake of simplicity, we \emph{a priori assume} that $\alpha^A$ and 
$\beta^A$ are smooth on $M$. Hence we can write their components, defined in 
the unphysical spinor dyad $\varepsilon^A_{\uA}:=\{o^A,\iota^A\}$ e.g. by 
$\alpha_{\uA}:=\alpha_A\varepsilon^A_{\uA}$, as 

\begin{equation}
\alpha_{\uA}=\alpha_{\uA}^{(0)}+\Omega \alpha_{\uA}^{(1)}+O(\Omega^2),
\hskip 20pt
\beta_{\uA}=\beta_{\uA}^{(0)}+\Omega\beta_{\uA}^{(1)}+O(\Omega^2);
\label{eq:4.5.2a}
\end{equation}
where the functions $\alpha_{\uA}^{(0)}, ... ,\beta_{\uA}^{(1)}$ depend only on 
the coordinates $\tilde x^\mu$, i.e. they are functions on ${\cal S}=\Sigma
\cap\mathscr{I}^+$. (The subsequent analysis shows that even a slightly less 
restrictive condition, viz. $\alpha_{\uA}=\alpha^{(0)}_{\uA}+\Omega^k\alpha^{(1)}
_{\uA}+O(\Omega^{k+1})$, $k>1/2$, would already be enough.) Since $\hat\alpha_A
=\Omega^{-1}\alpha_A$ and $\hat\beta_A=\Omega^{-1}\beta_A$ hold, by 
(\ref{eq:4.18b}) the components of the spinor fields in the physical spinor 
dyad $\hat\varepsilon^A_{\uA}=\{\hat o^A,\hat\iota^A\}=\Omega^{\frac{1}{2}}
\varepsilon^A_{\uA}$, ${\uA}=0,1$, have the asymptotic form 

\begin{equation}
\hat\alpha_{\uA}=\Omega^{-\frac{1}{2}}\Bigl(\alpha_{\uA}^{(0)}+\Omega\alpha_{\uA}
 ^{(1)}+O(\Omega^2)\Bigr), \hskip 20pt
\hat\beta_{\uA}=\Omega^{-\frac{1}{2}}\Bigl(\beta_{\uA}^{(0)}+\Omega\beta_{\uA}^{(1)}
 +O(\Omega^2)\Bigr). \label{eq:4.5.2b}
\end{equation}
Thus, the components of the spinor fields in the physical spacetime 
\emph{diverge} as $\sqrt{\tilde r}$. 

In the unphysical spacetime (\ref{eq:4.5.1}) takes the form 

\begin{eqnarray}
&{}&0={\cal D}_{A'A}\alpha^A+\frac{3}{2}\Omega^{-1}\Bigl(K\bar\beta_{A'}-(\nabla
  _{A'A}\Omega)\alpha^A\Bigr), \label{eq:4.5.3a} \\
&{}&0={\cal D}_{AA'}\bar\beta^{A'}+\frac{3}{2}\Omega^{-1}\Bigl(K\alpha_A-
  (\nabla_{AA'}\Omega)\bar\beta^{A'}\Bigr). \label{eq:4.5.3b}
\end{eqnarray}
Let us recall that near $\mathscr{I}^+=\{\Omega=0\}$ the unit normal of the 
$\Omega={\rm const}$ hypersurfaces is $N_a=-\sqrt{3/\Lambda}(\nabla_a\Omega)+
O(\Omega^2)$ (see equation (\ref{eq:2.2.2c})). Thus, multiplying these 
equations by $\Omega$ and evaluating at $\Omega=0$, we obtain 

\begin{equation}
\pm{\rm i}\bar\beta_{A'}\approx-\sqrt{2}N_{A'A}\alpha^A, \hskip 20pt
\pm{\rm i}\alpha_A\approx-\sqrt{2}N_{AA'}\bar\beta^{A'}. \label{eq:4.5.4}
\end{equation}
These are not independent, one implies the other. Therefore, the pair 
$(\alpha^A,\beta^A)$ of spinor fields can solve (\ref{eq:4.5.1}) only if the 
asymptotic value of one of them determines the other at the conformal 
boundary according to (\ref{eq:4.5.4}). In terms of the spinor components 
(\ref{eq:4.5.4}) is equivalent to $\alpha_0^{(0)}=\pm{\rm i}\bar\beta_{1'}
^{(0)}$ and $\alpha_1^{(0)}=\mp{\rm i}\bar\beta_{0'}^{(0)}$. Thus the two spinor 
fields are linked to each other, but they are still not specified at the 
conformal boundary.

\subsubsection{The asymptotic structure of the solution}
\label{sub-4.5.2}

Since by (\ref{eq:4.5.4}) $K\bar\beta_{A'}-(\nabla_{A'A}\Omega)\alpha^A
\approx0$, we may write 

\begin{equation*}
\Omega\bar\gamma_{A'}:=K\bar\beta_{A'}-(\nabla_{A'A}\Omega)\alpha^A
\end{equation*}
for some smooth spinor field $\bar\gamma_{A'}$ on $M$. In terms of $\alpha^A$ 
and $\bar\gamma^{A'}$ the renormalized Witten equations are 

\begin{eqnarray}
0\!\!\!\!&=\!\!\!\!&{\cal D}_{A'A}\alpha^A+\frac{3}{2}\bar\gamma_{A'},
 \label{eq:4.5.6a} \\
0\!\!\!\!&=\!\!\!\!&\Omega{\cal D}_{AA'}\bar\gamma^{A'}+\bar\gamma^{A'}(D
 _{A'A}\Omega)+P^c_{AA'}\varepsilon^{A'B'}\bigl(\nabla_c\nabla_b\Omega\bigr)
 \alpha^B+(\nabla^{A'}{}_B\Omega){\cal D}_{A'A}\alpha^B
 \nonumber \\
\!\!\!\!&+\!\!\!\!&\frac{3}{2}\Omega^{-1}\Bigl(K^2+\frac{1}{2}(\nabla_b
 \Omega)(\nabla^b\Omega)\Bigr)\alpha_A-\frac{3}{2}\bar\gamma^{A'}\nabla_{A'A}
 \Omega. \label{eq:4.5.6b}
\end{eqnarray}
Using the 3+1 decomposition $\nabla_e\Omega=D_e\Omega+t_et^f\nabla_f\Omega$ 
and equations (\ref{eq:2.2.2b}), (\ref{eq:2.2.2c}) and (\ref{eq:4.5.6a}), 
the second of these takes the form 

\begin{equation*}
0=(D^{A'}{}_B\Omega)\Bigl({\cal D}_{A'A}\alpha^B+\frac{1}{2}\bar\gamma_{A'}
\delta^B_A\Bigr)+\Omega\Bigl({\cal D}_{AA'}\bar\gamma^{A'}+P^c_{AA'}\varepsilon
^{A'B'}\bigl(\frac{1}{4}g_{bc}\Psi+\Psi_{bc}\bigr)\alpha^B+\frac{3}{2}\Phi
\alpha_A\Bigr).
\end{equation*}
Evaluating this equation at $\Omega=0$, we find that 

\begin{equation}
v^{A'}{}_B{\cal D}_{A'A}\alpha^B+\frac{1}{2}\bar\gamma_{A'}v^{A'}{}_A\approx0,
\label{eq:4.5.7}
\end{equation}
where $v_a:= D_a\Omega/\vert D_e\Omega\vert$ (see subsection \ref{sub-4.4}). 
Let $\Pi^a_b:=\delta^a_b-t^at_b+v^av_b=P^a_b+v^av_b$, the orthogonal projection 
to the 2-surfaces $\hat{\cal S}_\Omega$, and define $\Delta_a:=\Pi^b_a\nabla_b$, 
the two-dimensional version of the Sen connection on the 2-surfaces. Then the 
2+1 decomposition of the derivative ${\cal D}_a$ by $\Delta_a$ and the 
directional derivative $v^e{\cal D}_e$ yields that (\ref{eq:4.5.7}) has the 
form 

\begin{equation}
\Delta_{A'A}\alpha^A-v^e({\cal D}_e\alpha_A)v^A{}_{A'}+\frac{1}{2}\bar\gamma
_{A'}\approx0. \label{eq:4.5.8}
\end{equation}
Here we used that $v^{A'}{}_A\Delta_{A'B}=v^{A'}{}_B\Delta_{A'A}$ and $2v^{AA'}
v_{AB'}=-\delta^{A'}_{B'}$. On the other hand, after a similar decomposition 
the first of the renormalized Witten equations, equation (\ref{eq:4.5.6a}), 
yields 

\begin{equation}
\Delta_{A'A}\alpha^A+v^e({\cal D}_e\alpha_A)v^A{}_{A'}+\frac{3}{2}\bar\gamma
_{A'}\approx0. \label{eq:4.5.9}
\end{equation}
By (\ref{eq:4.5.8}) and (\ref{eq:4.5.9}) 

\begin{equation}
\Delta_{A'A}\alpha^A+\bar\gamma_{A'}=O(\Omega), \hskip 20pt
v^e{\cal D}_e\alpha^A+v^{AA'}\bar\gamma_{A'}=O(\Omega);
\label{eq:4.5.10a}
\end{equation}
which imply, in particular, that 

\begin{equation}
v^e{\cal D}_e\alpha^A\approx v^{AA'}\Delta_{A'B}\alpha^B. \label{eq:4.5.10}
\end{equation}
Clearly, there is a similar relationship between the tangential and radial 
derivatives of the spinor field $\beta^A$, too. Hence, the radial and 
tangential derivatives of the spinor fields on the cut are linked together. 
To evaluate this, we rewrite it into its GHP form. 

Contracting (\ref{eq:4.5.10}) with the vectors of the spinor dyad, using the 
asymptotic form of the spin coefficients (in the gauge $\varepsilon^0=
\varepsilon^{\prime 0}=0$), equations (\ref{eq:4.4.1c}) and (\ref{eq:4.4.1b}), 
$v^a=-\vert T_e\vert^{-1}(\partial/\partial\Omega)^a+O(\Omega)$ and the 
expansion (\ref{eq:4.5.2a}), we obtain that 

\begin{equation}
\alpha^{(1)}_0=-\sqrt{\frac{\Lambda}{6}}W\Bigl({\edth}\alpha^{(0)}_1-\rho^0
 \alpha^{(0)}_0\Bigr), \hskip 20pt
\alpha^{(1)}_1=\sqrt{\frac{6}{\Lambda}}\bigl(1+\frac{1}{6}\Lambda W\bigr)
 \Bigl({\edth}^\prime\alpha^{(0)}_0+\rho^0\alpha^{(0)}_1\Bigr); \label{eq:4.5.11}
\end{equation}
and there are analogous formulae for the expansion coefficients $\beta^{(1)}_0$  
and $\beta^{(1)}_1$, too. Here ${\edth}$ and ${\edth}'$ are the standard GHP 
edth operators on ${\cal S}$ \cite{GHP}. Thus, \emph{in the solutions of the 
renormalized Witten equation the $\Omega=1/\tilde r$ order terms in their 
asymptotic expansion are determined completely by their boundary value and 
the boost gauge defined by $\Sigma$ (and represented by $W$) at 
$\mathscr{I}^+$}. In particular, on asymptotically null hypersurfaces 
$\alpha^{(1)}_0$ and $\beta^{(1)}_0$ would be vanishing. Therefore, with the 
definitions 

\begin{equation}
{}_0\alpha_A:=\bigl(\alpha^{(0)}_{\uA}+\Omega\alpha^{(1)}_{\uA}\bigr)\varepsilon
 ^{\uA}_A,   \hskip 20pt
{}_0\beta_A:=\bigl(\beta^{(0)}_{\uA}+\Omega\beta^{(1)}_{\uA}\bigr)\varepsilon
 ^{\uA}_A \label{eq:4.5.12}
\end{equation}
the solution of the renormalized Witten equations has the asymptotic form 

\begin{equation}
\hat\alpha^A=\alpha^A={}_0\alpha^A+\hat\sigma^A, \hskip 20pt
\bar{\hat\beta}^{A'}=\bar\beta^{A'}={}_0\bar\beta^{A'}+\bar{\hat\pi}^{A'},
\label{eq:4.5.13}
\end{equation}
where e.g. $\alpha^{(0)}_{\uA}$ determines $\beta^{(0)}_{\uA}$ through the 
algebraic boundary condition (\ref{eq:4.5.4}), the coefficients $\alpha^{(1)}
_{\uA}$ satisfy (\ref{eq:4.5.11}), $\beta^{(1)}_{\uA}$ satisfy the analogous 
equation (and hence also determined by $\alpha^{(0)}_{\uA}$), and the 
components of $\hat\sigma^A$ and $\bar{\hat\pi}^{A'}$ in the unphysical spin 
frame $\varepsilon^A_{\uA}$ are of order $\Omega^2$. (N.B.: The dual spin 
frame is $\varepsilon^{\uA}_A=-\epsilon^{\uA\uB}\varepsilon^B_{\uB}\varepsilon
_{BA}$, where $\epsilon^{\uA\uB}$ is the anti-symmetric Levi-Civita symbol.)

\subsubsection{Example: The de Sitter spacetime}
\label{sub-4.5.3}

In the positivity and rigidity proofs in subsections \ref{sub-4.6.2} and 
\ref{sub-4.6.3} we need to know some of the properties of the de Sitter 
spacetime. It is known that in this spacetime the differential equation 

\begin{equation}
\hat\nabla_e\Psi^\alpha=\pm\frac{\rm i}{\sqrt{2}}\sqrt{\frac{\Lambda}{6}}
\hat\gamma^\alpha_{e\beta}\Psi^\beta \label{eq:4.5.14}
\end{equation}
is completely integrable (for either sign on the right), where $\Psi^\alpha=
(\alpha^A,\bar\beta^{A'})$ (as a column vector) and the Dirac `matrices' are 
given explicitly in terms of the metric spinor by 

\begin{equation*}
\gamma^\alpha_{e\beta}=\sqrt{2}\left(\begin{array}{cc}
          0&\varepsilon_{E'B'}\delta^A_E \\
      \varepsilon_{EB}\delta^{A'}_{E'}&0 \end{array}\right)
\end{equation*}
(see e.g. \cite{PR1}, pp 221). Hence it admits four linearly independent 
solutions and these solutions can be specified by prescribing $\Psi^\alpha$ at 
any given point of the spacetime. In fact, its Weyl spinor constituents solve 
the 1--valence twistor equation such that the primary spinor part of one 
twistor is just the secondary spinor part of the other \cite{Sz13}; and the 
solutions of (\ref{eq:4.5.14}) also solve (\ref{eq:4.5.1}) on any spacelike 
hypersurface $\Sigma$. In flat spacetime ($\Lambda=0$) its solutions are just 
the spinor constituents of the \emph{translational} Killing vectors. Thus the 
solutions of (\ref{eq:4.5.14}) are the spinor constituents of what 
substitutes the translational Killing fields in de Sitter spacetime most 
naturally. (For a more detailed discussion of the geometry of the de Sitter 
spacetime and the twistor equation, see e.g. \cite{Tod15}.) 

To find the explicit solutions, let us rewrite (\ref{eq:4.5.14}) in the GHP 
formalism. In the coordinate system based on a spherically symmetric foliation 
of $\mathscr{I}^+$ and the GHP spin frame (up to phase transformation of the 
complex null vectors) of subsection \ref{sub-4.2} and \ref{sub-4.3}, 
respectively, the only non-zero GHP spin coefficients are 

\begin{eqnarray*}
&{}&\hat\rho=\sqrt{\frac{\Lambda}{6}}\sin\bigl(\sqrt{\frac{3}{\Lambda}}
  w\bigr)\cot\bigl(\sqrt{\frac{\Lambda}{3}}u-\frac{\pi}{2}+\sqrt{\frac{3}
  {\Lambda}}w\bigr)-\sqrt{\frac{\Lambda}{6}}\cos\bigl(\sqrt{\frac{3}
  {\Lambda}}w\bigr), \\
&{}&\hat\rho^\prime=-\sqrt{\frac{\Lambda}{6}}\sin\bigl(\sqrt{\frac{3}
  {\Lambda}}w\bigr)\cot\bigl(\sqrt{\frac{\Lambda}{3}}u-\frac{\pi}{2}+\sqrt{
  \frac{3}{\Lambda}}w\bigr)-\sqrt{\frac{\Lambda}{6}}\cos\bigl(\sqrt{
  \frac{3}{\Lambda}}w\bigr), \\
&{}&\hat\beta=\bar{\hat\beta}^\prime=-\frac{1}{2\sqrt{2}}\sqrt{\frac{\Lambda}
  {3}}\frac{\sin\bigl(\sqrt{\frac{3}{\Lambda}}w\bigr)}{\sin\bigl(\sqrt{
  \frac{\Lambda}{3}}u-\frac{\pi}{2}+\sqrt{\frac{3}{\Lambda}}w\bigr)}\zeta, \\
&{}&\hat\varepsilon={\hat\varepsilon}^\prime=\frac{1}{2}\sqrt{\frac{\Lambda}
  {6}}\cos\bigl(\sqrt{\frac{3}{\Lambda}}w\bigr).
\end{eqnarray*}
Then, contracting (\ref{eq:4.5.14}) with $\hat o^E\bar{\hat o}^{E'}$ we obtain 
the so-called radial equations (i.e. the parts of (\ref{eq:4.5.14}) tangential 
to the null geodesic generators of the null hypersurfaces ${\cal N}_u$), 
whose solution is given by 

\begin{eqnarray*}
&{}&\hat\alpha_0=\frac{\alpha_0^{(0)}}{\sqrt{\sin\bigl(\sqrt{\frac{3}
 {\Lambda}}w\bigr)}}, \hskip 45pt
 \bar{\hat\beta}_{0'}=\frac{\bar\beta_{0'}^{(0)}}{\sqrt{\sin\bigl(
 \sqrt{\frac{3}{\Lambda}}w\bigr)}}, \\
&{}&\hat\alpha_1=\frac{1}{\sqrt{\sin\bigl(\sqrt{\frac{3}{\Lambda}}w\bigr)}}
 \Bigl(\mp{\rm i}\bar\beta_{0'}^{(0)}\cos\bigl(\sqrt{\frac{3}{\Lambda}}w
 \bigr)+\sin\bigl(\sqrt{\frac{3}{\Lambda}}w\bigr)A_1\Bigr), \\
&{}&\bar{\hat\beta}_{1'}=\frac{1}{\sqrt{\sin\bigl(\sqrt{\frac{3}{\Lambda}}w
 \bigr)}}\Bigl(\mp{\rm i}\alpha_0^{(0)}\cos\bigl(\sqrt{\frac{3}{\Lambda}}w
 \bigr)+\sin\bigl(\sqrt{\frac{3}{\Lambda}}w\bigr)\bar B_{1'}\Bigr).
\end{eqnarray*}
Here $\alpha_0^{(0)}$, $\bar\beta_{0'}^{(0)}$, $A_1$ and $\bar B_{1'}$ are still 
to be determined functions of $u$ and $x^\mu$. Thus, in particular, this 
solution is compatible with both the general fall off properties 
(\ref{eq:4.5.2a}) and the algebraic boundary conditions (\ref{eq:4.5.4}). 

The contraction of (\ref{eq:4.5.14}) with $\hat o^E\bar{\hat\iota}^{E'}$ and 
with $\hat\iota^E\bar{\hat o}^{E'}$ give the so-called surface equations, 
i.e. the ones tangential to the $u={\rm const}$, $w={\rm const}$ 2-spheres. 
Substituting the solution of the radial equations here we obtain, in 
particular, that 

\begin{equation*}
{}_0{\edth}\alpha_0^{(0)}=0, \hskip 20pt
{}_0{\edth}^\prime\bar\beta_{0'}^{(0)}=0, \hskip 20pt
{}_0{\edth}^\prime A_1=0, \hskip 20pt
{}_0{\edth}\bar B_{1'}=0,
\end{equation*}
where ${}_0{\edth}$ and ${}_0{\edth}^\prime$ denote the standard edth operators 
on the \emph{unit sphere} \cite{PR1}. The solution of these equations is well 
known to be 

\begin{equation*}
\alpha_0^{(0)}=\sum_ma_m\,{}_{\frac{1}{2}}Y_{\frac{1}{2}\,m}, \hskip 8pt
\bar{\beta}_{0'}^{(0)}=\sum_mb_m\,{}_{-\frac{1}{2}}Y_{\frac{1}{2}\,m}, \hskip 8pt
A_1=\sum_mA_m\,{}_{-\frac{1}{2}}Y_{\frac{1}{2}\,m}, \hskip 8pt
\bar{B}_{1'}=\sum_mB_m\,{}_{\frac{1}{2}}Y_{\frac{1}{2}\,m},
\end{equation*}
where ${}_{\pm\frac{1}{2}}Y_{\frac{1}{2}\,m}$ are the $\pm\frac{1}{2}$ spin weighted 
spherical harmonics, $m=-\frac{1}{2},\frac{1}{2}$, and the coefficients $a_m$, 
$b_m$, $A_m$ and $B_m$ are still not specified functions of $u$. Substituting 
these into the remaining four surface equations and using how the edth 
operators act on the spin weighted spherical harmonics we find that $A_m$ and 
$B_m$ are determined by $a_m$ and $b_m$ according to 

\begin{eqnarray*}
&{}&\sin\bigl(\frac{\pi}{2}+\sqrt{\frac{\Lambda}{3}}u\bigr)\,A_m=a_m\mp
 {\rm i}\cos\bigl(\frac{\pi}{2}+\sqrt{\frac{\Lambda}{3}}u\bigr)\,b_m, \\
&{}&\sin\bigl(\frac{\pi}{2}+\sqrt{\frac{\Lambda}{3}}u\bigr)\,B_m=-b_m\mp
 {\rm i}\cos\bigl(\frac{\pi}{2}+\sqrt{\frac{\Lambda}{3}}u\bigr)\,a_m;
\end{eqnarray*}
but $a_m$ and $b_m$ remain independent. It might be worth noting that these 
are just the conditions (\ref{eq:4.5.11}), in which we substitute $W=0$ 
(since the hypersurface on which the spinor components are expanded is null), 
the radius of the $u={\rm const}$ cut is $R=\sin(\frac{\pi}{2}-\sqrt{
\frac{\Lambda}{3}}u)$ (see the line element of the Einstein universe in 
subsection \ref{sub-4.2}) and $\rho^0=\frac{1}{\sqrt{2}}\cot(\sqrt{\frac{
\Lambda}{3}}u-\frac{\pi}{2})$. The $u$-dependence of $a_m$ and $b_m$ is 
determined by the so-called evolution equations, obtained by contracting 
(\ref{eq:4.5.14}) with $\hat\iota^E\bar{\hat\iota}^{E'}$, and the whole 
solution is completely determined by the value of the four functions $a_m$ 
and $b_m$ e.g. at $u=0$. Therefore, the solution $(\alpha^A,\bar\beta^{A'})$ 
of equation (\ref{eq:4.5.14}) is completely determined by its spinor 
components $\alpha_0=\alpha_Ao^A$ and $\bar\beta_{0'}=\bar\beta_{A'}\bar o^{A'}$ 
on one $u={\rm const}$ cut of the conformal boundary. 

Let ${\cal S}_u$ denote the $u={\rm const}$ cut. This can be considered as 
the intersection of some spherically symmetric spacelike hypersurface 
$\Sigma$ and the conformal infinity, and hence we should ask how the 
solutions $\alpha^A$ and $\beta^A$ on ${\cal S}_u$ could be recovered 
\emph{purely} in terms of the geometry of ${\cal S}_u$. Since on $\mathscr{I}
^+$ $\alpha^A$ determines $\beta^A$ algebraically, it is enough to consider 
$\alpha^A$. For its structure we obtained $\alpha_0=\sum_ma_m\,{}_{\frac{1}{2}}
Y_{\frac{1}{2}m}$ and $\alpha_1=\mp{\rm i}\sum_mb_m\,{}_{-\frac{1}{2}}Y
_{\frac{1}{2}m}$. These are the general solution of 

\begin{equation}
{}_0{\edth}\alpha_0=0, \hskip 20pt
{}_0{\edth}^\prime\alpha_1=0, \label{eq:4.5.15}
\end{equation}
which are the \emph{2-surface twistor equations} on the spherically symmetric 
${\cal S}_u$.

\subsection{The positivity and rigidity of $H$}
\label{sub-4.6}

In this subsection we begin the proof of existence for solutions of the
Witten equation. We show that the requirement of the finiteness of the
functional $H[\alpha,\bar\alpha]+H[\beta,\bar\beta]$ on the solutions of the
renormalized Witten equations (\ref{eq:4.5.1}) yields that the spinor
fields must solve the 2-surface twistor equations on the conformal boundary,
and with these boundary values they have a unique solution, controlled by
the boundary value of e.g. $\alpha^A$, provided the matter fields satisfy the
dominant energy condition. This implies that $H[\alpha,\bar\alpha]+H[\beta,
\bar\beta]$ is non-negative for such boundary values (positivity), and that
it is vanishing precisely when the domain of dependence of the spacelike
hypersurface $\Sigma$ is locally isometric to the de Sitter spacetime
(rigidity).

\subsubsection{Boundary conditions from the finiteness of $H$: The 2-surface
twistor equations}
\label{sub-4.6.1}

Since $K$ in (\ref{eq:4.4}) is imaginary, the finiteness of $H[\alpha,\bar
\alpha]+H[\beta,\bar\beta]$ defined in the physical spacetime is ensured by
the finiteness of the integral on the right of (\ref{eq:4.7}). Since we
associated zero conformal weight to the (contravariant form of the) spinor
fields, moreover under the conformal rescaling of the spacetime metric the
volume element changes according to ${\rm d}\hat\Sigma=\Omega^{-3}{\rm d}
\Sigma$, the integral of the second term on the right hand side of
(\ref{eq:4.7}) is finite precisely when this term falls off \emph{slightly
faster} than $\Omega^2$, i.e. when

\begin{equation}
o\bigl(\Omega^2\bigr)=\hat t_a\hat T^a{}_b\bigl(\hat\alpha^B\bar{\hat\alpha}
^{B'}+\hat\beta^B\bar{\hat\beta}^{B'}\bigr)=\Omega^2t_aT^a{}_b\bigl(\alpha^B\bar
\alpha^{B'}+\beta^B\bar\beta^{B'}\bigr). \label{eq:4.6.1}
\end{equation}
Hence the rescaled energy-momentum tensor $T^a{}_b$ must tend to zero at the
conformal boundary, i.e. the physical energy-momentum tensor must fall off as
$\hat T^a{}_b=o(\Omega^3)$, \emph{slightly faster} than $\Omega^3$. Therefore,
assuming the smoothness of the rescaled energy-momentum tensor on $M$, we
must require $\hat T^a{}_b=O(\Omega^4)$ (rather than only $\hat T^a{}_b=\Omega
^3T^a{}_b$). Thus, in particular, dust-like matter cannot be present near the
conformal boundary.

To determine the condition of the finiteness of the integral of the first
term on the right of (\ref{eq:4.7}), we rewrite the derivative $\tilde{\cal
D}_e\hat\alpha_A$ into a more familiar form. It is known that

\begin{equation}
\sqrt{2}\hat t_F{}^{E'}\hat{\cal D}_{EE'}\hat\alpha_A=\hat{\cal D}_{(EF}\hat
\alpha_{A)}+\sqrt{2}\hat t_F{}^{E'}\frac{2}{3}P_{EE'}^{DD'}\hat\varepsilon_{DA}
\hat{\cal D}_{D'C}\hat\alpha^C \label{eq:4.6.2}
\end{equation}
is the complete algebraically irreducible, $\hat t^{AA'}$--orthogonal
decomposition of the derivative. Here $\hat{\cal D}_{AB}:=\sqrt{2}\hat t_B
{}^{A'}\hat{\cal D}_{AA'}=\hat{\cal D}_{(AB)}$ is the unitary spinor form of
$\hat{\cal D}_{AA'}$, and the totally symmetric part $\hat{\cal D}_{(AB}\hat
\alpha_{C)}$ defines the 3-surface twistor operator \cite{Tod84}. Substituting
this decomposition into the explicit expression of $\tilde{\cal D}_e\hat
\alpha_A$ given by (\ref{eq:4.3}) and using the renormalized Witten
equations (\ref{eq:4.5.1}), we find that \emph{it is precisely the 3-surface
twistor operator acting on $\hat\alpha_A$}, i.e.

\begin{eqnarray}
H\bigl[\alpha,\bar\alpha\bigr]+H\bigl[\beta,\bar\beta\bigr]\!\!\!\!&=
 \!\!\!\!&\int_\Sigma\Bigl\{\frac{4}{\varkappa}\hat t^{AA'}\hat t^{BB'}\hat t^{CC'}
 \Bigl((\hat{\cal D}_{(AB}\hat\alpha_{C)})(\hat{\cal D}_{(A'B'}\bar{\hat\alpha}
 _{C')}) \label{eq:4.6.3} \\
\!\!\!\!&+\!\!\!\!&(\hat{\cal D}_{(AB}\hat\beta_{C)})(\hat{\cal D}_{(A'B'}\bar
 {\hat\beta}_{C')})\Bigr)+\hat t_a\hat T^a{}_b(\hat\alpha^B\bar{\hat\alpha}^{B'}
 +\hat\beta^B\bar{\hat\beta}^{B'})\Bigr\}{\rm d}\hat\Sigma. \nonumber
\end{eqnarray}
It has exactly the same structure that the components of the ADM,
Bondi--Sachs and Abbott--Deser energy-momenta and the total mass of closed
universes (with $\Lambda\geq 0$) have in their spinorial form
\cite{HT,RT,Gibbonsetal,Sz12,Sz13}: It is the sum of the square of the
$L_2$-norm of the 3-surface twistor derivative of the spinor field satisfying
the gauge condition and the integral of the energy-momentum of the matter
fields.

Returning to the question of the finiteness of the integral of the first term
on the right of (\ref{eq:4.7}), (\ref{eq:4.6.3}) shows that $\hat{\cal D}_{(AB}
\hat\alpha_{C)}$ and $\hat{\cal D}_{(AB}\hat\beta_{C)}$ \emph{must be square
integrable} on $\Sigma$ in the physical spacetime. Since the 3-surface twistor
operator is conformally covariant, viz. $\hat{\cal D}_{(AB}(\Omega^{-1}\alpha
_{C)})=\Omega^{-1}{\cal D}_{(AB}\alpha_{C)}$, this condition is equivalent to the

\begin{equation}
t^{AA'}t^{BB'}t^{CC'}\bigl({\cal D}_{(AB}\alpha_{C)}\bigr)\bigl({\cal D}_{(A'B'}
\bar\alpha_{C')}\bigr)=o(\Omega) \label{eq:4.6.4}
\end{equation}
fall-off in the unphysical spacetime, and to an analogous one for $\beta_A$.
By (\ref{eq:4.4.1a}) this is equivalent to

\begin{equation*}
o^Ao^Bo^C\bigl({\cal D}_{AB}\alpha_{C}\bigr), \hskip 5pt
o^Ao^B\iota^C\bigl({\cal D}_{(AB}\alpha_{C)}\bigr), \hskip 5pt
o^A\iota^B\iota^C\bigl({\cal D}_{(AB}\alpha_{C)}\bigr), \hskip 5pt
\iota^A\iota^B\iota^C\bigl({\cal D}_{AB}\alpha_{C}\bigr)=o(\Omega^{\frac{1}{2}}).
\end{equation*}
Then, by (\ref{eq:4.4.1a}) and the asymptotic value of the GHP spin
coefficients in the conformal spacetime, we obtain that

\begin{eqnarray}
o^Ao^Bo^C\bigl({\cal D}_{AB}\alpha_{C}\bigr)\!\!\!\!&=\!\!\!\!&o^Ao^B\sqrt{2}
 t_B{}^{A'}\bigl(\nabla_{AA'}\alpha_C\bigr)o^C=-\sqrt{\frac{\Lambda}{3}}
 \frac{W}{\vert T_e\vert}m^a\bigl(\nabla_a\alpha_C\bigr)o^C+O(\Omega)
 \nonumber \\
\!\!\!\!&=\!\!\!\!&-\sqrt{\frac{\Lambda}{3}}\frac{W}{\vert T_e\vert}\Bigl(
 {\edth}\alpha^{(0)}_0+\sigma^0\alpha^{(0)}_1\Bigr)+O(\Omega),\label{eq:4.6.4a}\\
\iota^A\iota^B\iota^C\bigl({\cal D}_{AB}\alpha_{C}\bigr)\!\!\!\!&=\!\!\!\!&
 \sqrt{\frac{12}{\Lambda}}\frac{1}{\vert T_e\vert}\bigl(1+\frac{1}{6}\Lambda
 W\bigr)\Bigl({\edth}^\prime\alpha^{(0)}_1-\sigma^0\alpha^{(0)}_0\Bigr)+
 O(\Omega). \label{eq:4.6.4b}
\end{eqnarray}
Using (\ref{eq:4.4.1a})-(\ref{eq:4.4.1c}), $\vert T_e\vert^2=\frac{1}{3}
\Lambda W^2+2W+O(\Omega)$ and $v^a=-\vert T_e\vert^{-1}(\partial/\partial
\Omega)^a+O(\Omega)$, we find that

\begin{eqnarray*}
3\!\!\!\!&{}\!\!\!\!&o^Ao^B\iota^C\bigl({\cal D}_{(AB}\alpha_{C)}\bigr)=o^Ao^B
 \sqrt{2}t_B{}^{A'}P^e_a\bigl(\nabla_e\alpha_C\bigr)\iota^C+2o^A\iota^B
 \sqrt{2}t_B{}^{A'}P^e_a\bigl(\nabla_e\alpha_C\bigr)o^C \\
=\!\!\!\!&{}\!\!\!\!&-\sqrt{\frac{\Lambda}{3}}\frac{W}{\vert T_e\vert}m^a\bigl(
 \nabla_a\alpha_C\bigr)\iota^C+\frac{2\sqrt{2}}{\vert T_e\vert}\sqrt{
 \frac{6}{\Lambda}}\bigl(1+\frac{1}{6}\Lambda W\bigr)o^A\bar o^{A'}P^b_a
 \bigl(\nabla_b\alpha_C\bigr)o^C+O(\Omega)\\
=\!\!\!\!&{}\!\!\!\!&-\sqrt{\frac{\Lambda}{3}}\frac{W}{\vert T_e\vert}\bigl(
 {\edth}\alpha^{(0)}_1-\rho^0\alpha^{(0)}_0\bigr)+2\sqrt{2}\frac{W}{\vert T_e
 \vert^2}\bigl(1+\frac{1}{6}\Lambda W\bigr)v^a\bigl({\cal D}_a\alpha_C\bigr)
 o^C+O(\Omega) \\
=\!\!\!\!&{}\!\!\!\!&-\frac{W}{\vert T_e\vert}\Bigl(\sqrt{\frac{\Lambda}{3}}
 \bigl({\edth}\alpha^{(0)}_1-\rho^0\alpha^{(0)}_0\bigr)+\frac{2\sqrt{2}}{\vert
 T_e\vert^2}\bigl(1+\frac{1}{6}\Lambda W\bigr)\alpha^{(1)}_0\Bigr) \\
\!\!\!\!&{}\!\!\!\!&+\frac{2\sqrt{2}W}{\vert T_e\vert^3}\bigl(1+\frac{1}{6}
 \Lambda W\bigr)\Bigl(\sqrt{\frac{6}{\Lambda}}\bigl(1+\frac{1}{6}\Lambda W
 \bigr)\bigl(\kappa\alpha_1-\varepsilon\alpha_0\bigr)-\sqrt{\frac{\Lambda}{6}}
 W\bigl(\tau\alpha_1+\varepsilon^\prime\alpha_0\bigr)\Bigr)+O(\Omega).
\end{eqnarray*}
Then by (\ref{eq:4.5.11}) and the asymptotic form of the GHP spin coefficients
this, and the analogous expression for $o^A\iota^B\iota^C({\cal D}_{(AB}\alpha
_{C)})$, yield that

\begin{equation}
o^Ao^B\iota^C\bigl({\cal D}_{(AB}\alpha_{C)}\bigr)=O(\Omega), \hskip 20pt
o^A\iota^B\iota^C\bigl({\cal D}_{(AB}\alpha_{C)}\bigr)=O(\Omega).
\label{eq:4.6.4c}
\end{equation}
Therefore, by (\ref{eq:4.6.4c}) $o^Ao^B\iota^C({\cal D}_{(AB}\alpha_{C)})$ and
$o^A\iota^B\iota^C({\cal D}_{(AB}\alpha_{C)})$ fall off appropriately, but by
(\ref{eq:4.6.4a}) and (\ref{eq:4.6.4b}) the condition (\ref{eq:4.6.4}) is
satisfied precisely when

\begin{equation}
{\edth}\alpha^{(0)}_0+\sigma^0\alpha^{(0)}_1=0, \hskip 20pt
{\edth}^\prime\alpha^{(0)}_1-\sigma^0\alpha^{(0)}_0=0 \label{eq:4.6.4d}
\end{equation}
also hold, i.e. \emph{if the spinor field $\alpha_A$ on the cut ${\cal S}=
\Sigma\cap\mathscr{I}^+$ solves the 2-surface twistor equations} of Penrose
\cite{PR2}. In particular, by (\ref{eq:4.5.15}) the solution of
(\ref{eq:4.5.14}) in the de Sitter spacetime also satisfies this condition.

\subsubsection{Positivity}
\label{sub-4.6.2}

Assuming that the matter fields satisfy the dominant energy condition, the
proof of the non-negativity of $H[\alpha,\bar\alpha]+H[\beta,\bar\beta]$
reduces to the proof of the existence of solutions of (\ref{eq:4.5.1}) with
the boundary values satisfying (\ref{eq:4.5.4}) and (\ref{eq:4.6.4d}). To
prove this existence, let us use the decomposition
(\ref{eq:4.5.12})-(\ref{eq:4.5.13}) of the spinor fields in (\ref{eq:4.5.1}).
Then the homogeneous renormalized Witten equations take the form of the
system

\begin{eqnarray}
&{}&\hat{\cal D}_{A'A}\hat\sigma^A+\frac{3}{2}K\bar{\hat\pi}_{A'}=-\Bigl(
 \hat{\cal D}_{A'A}\,{}_0\hat\alpha^A+\frac{3}{2}K\,{}_0\bar{\hat\beta}_{A'}
 \Bigr)=:\bar{\hat\omega}_{A'}, \label{eq:4.6.5a} \\
&{}&\hat{\cal D}_{AA'}\bar{\hat\pi}^{A'}+\frac{3}{2}K\hat\sigma_A=-\Bigl(
 \hat{\cal D}_{AA'}\,{}_0\bar{\hat\beta}^{A'}+\frac{3}{2}K\,{}_0\hat\alpha_A
 \Bigr)=:\hat\rho_A \label{eq:4.6.5b}
\end{eqnarray}
of inhomogeneous equations. The advantage of using $(\hat\sigma^A,\bar{\hat
\pi}^{A'})$ instead of $(\hat\alpha^A,\bar{\hat\beta}^{A'})$ is that the spinor
fields satisfying the homogeneous boundary condition form a vector space,
while those satisfying an inhomogeneous one do not. Thus, the techniques of
linear functional analysis can be applied to them more easily. Moreover, the
spinor fields on the right hand side of (\ref{eq:4.6.5a}) and
(\ref{eq:4.6.5b}) are fixed by the boundary conditions, and hence what we
should prove is only the existence of the spinor fields $(\hat\sigma^A,\bar
{\hat\pi}^{A'})$ with appropriate fall-off. Also, we will need the uniqueness
of this solution.

First we show that the \emph{homogeneous} equations corresponding to
(\ref{eq:4.6.5a})-(\ref{eq:4.6.5b}), i.e.

\begin{equation}
\hat{\cal D}_{A'A}\hat\sigma^A+\frac{3}{2}K\bar{\hat\pi}_{A'}=0, \hskip 20pt
\hat{\cal D}_{AA'}\bar{\hat\pi}^{A'}+\frac{3}{2}K\hat\sigma_A=0, \label{eq:4.6.6}
\end{equation}
do not admit any non-trivial smooth solution with the $\hat\sigma^A, \hat\pi^A
=o(\Omega^{3/2})$ fall-off, provided the matter fields satisfy the dominant
energy condition on $\Sigma$. (Note that, by the results of subsection
\ref{sub-4.6.1}, the fall-off $\hat\sigma^A,\hat\pi^A=o(\Omega^{3/2})$ is
needed to ensure the finiteness of $H[\alpha,\bar\alpha]+H[\beta,\bar\beta]$.
This fall-off condition is equivalent to their square integrability, see
below.) Suppose, on the contrary, that $(\hat\sigma^A,\bar{\hat\pi}^{A'})$ is
such a solution, and let us apply the Sen--Witten type identity (\ref{eq:4.7})
to this solution. Then by the dominant energy condition we have that

\begin{eqnarray}
0\leq\int_\Sigma\Bigl\{\!\!\!\!&-\!\!\!\!&\frac{2}{\varkappa}\hat t_{AA'}\hat
 h^{ef}\bigl((\tilde{\cal D}_e\hat\sigma^A)(\overline{\tilde{\cal D}}_f
 \bar{\hat\sigma}^{A'})+(\overline{\tilde{\cal D}}_e\hat\pi^A)(\tilde
 {\cal D}_f\bar{\hat\pi}^{A'})\bigr) \nonumber \\
\!\!\!\!&+\!\!\!\!&\hat t_a\hat T^a{}_b(\hat\sigma^B\bar{\hat\sigma}^{B'}+
 \hat\pi^B\bar{\hat\pi}^{B'})\Bigr\}{\rm d}\hat\Sigma=H\bigl[\hat\sigma,\bar{
 \hat\sigma}\bigr]+H\bigl[\hat\pi,\bar{\hat\pi}\bigr]. \label{eq:4.6.7}
\end{eqnarray}
We show that the right hand side of this inequality is also zero. The GHP
form of $H[\hat\sigma,\bar{\hat\sigma}]$ is $2/\varkappa$ times the
$\Omega\rightarrow0$ limit of

\begin{eqnarray}
&{}&\oint_{\hat{\cal S}_\Omega}\Bigl(\bar{\hat\sigma}_{1'}\bigl(\hat{\edth}
 {}^\prime\hat\sigma_0+\hat\rho\hat\sigma_1\bigr)-\bar{\hat\lambda}_{0'}
 \bigl(\hat{\edth}\hat\sigma_1+\hat\rho{}^\prime\hat\sigma_0\bigr)\Bigr)
 {\rm d}\hat{\cal S} \label{eq:4.6.7b} \\
&=&\oint_{\hat{\cal S}_\Omega}\Bigl(\bar{\hat\sigma}_{1'}\bigl(\hat{\edth}{}
 ^\prime\hat\sigma_0-\sqrt{\frac{\Lambda}{6}}\hat\sigma_1+\hat\sigma_1O
 (\Omega)\bigr)-\bar{\hat\sigma}_{0'}\bigl(\hat{\edth}\hat\sigma_1-
 \sqrt{\frac{\Lambda}{6}}\hat\sigma_0+\hat\sigma_0O(\Omega)\bigr)\Bigr)
 \Omega^{-2}{\rm d}{\cal S}, \nonumber
\end{eqnarray}
where we used the asymptotic form (\ref{eq:4.19c}) of the GHP convergences.
Thus, if $\hat\sigma^A=\Omega^l\,{}_0\sigma^A$ with $l>3/2$ and some bounded
spinor field ${}_0\sigma^A$ near $\mathscr{I}^+$, then $\hat\sigma_0$ and
$\hat\sigma_1$ fall off \emph{faster} than $\Omega$. This, together with the
same argument for $\hat\pi^A$, imply that the right hand side of
(\ref{eq:4.6.7}) is indeed zero. Hence, the integrand of the middle term of
(\ref{eq:4.6.7}) is vanishing, i.e.

\begin{eqnarray}
&{}&\hat{\cal D}_e\hat\sigma^A+KP_e^{AA'}\bar{\hat\pi}_{A'}=0, \hskip 20pt
\hat{\cal D}_e\bar{\hat\pi}^{A'}+KP_e^{A'A}\hat\sigma_A=0, \label{eq:4.6.8a} \\
&{}&\hat t_a\hat T^a{}_b\bigl(\hat\sigma^B\bar{\hat\sigma}^{B'}+\hat\pi^B
\bar{\hat\pi}^{B'}\bigr)=0. \label{eq:4.6.8b}
\end{eqnarray}
Now we show that the spinor fields $\hat\sigma^A$ and $\hat\pi^A$ cannot
be proportional to each other on any open subset of $\Sigma$. Thus,
suppose, on the contrary, that $\hat\pi^A=F\hat\sigma^A$ for some smooth
complex function $F$ on some open subset $U\subset\Sigma$. Then by
(\ref{eq:4.6.8a})

\begin{equation*}
0=\hat D_e\bigl(\hat\sigma_A\hat\pi^A\bigr)=-\bigl(\hat{\cal D}_e\hat\sigma
^A\bigr)\hat\pi_A+\hat\sigma_A\bigl(\hat{\cal D}_e\hat\pi^A\bigr)=K\bigl(1+F
\bar F\bigr)P_e^{AA'}\hat\sigma_A\bar{\hat\sigma}_{A'}.
\end{equation*}
Since $\Sigma$ is spacelike and $\hat\sigma_A\bar{\hat\sigma}_{A'}$ is null,
this implies the vanishing of $\hat\sigma_A$ on $U$. However, by
(\ref{eq:4.6.8a}) $\hat\sigma^A$ solves the eigenvalue equation $\hat{\cal
D}^{AA'}\hat{\cal D}_{A'B}\hat\sigma^B=\frac{3}{8}\Lambda\hat\sigma^A$, and
hence, by an appropriate modification of the proof of Aronszajn's theorem
for its eigenspinors, the spinor field $\hat\sigma_A$ cannot be vanishing on
any open set $U\subset\Sigma$. (For the details, see the appendix of
\cite{Sz13}.) Therefore, $\hat\sigma^A$ and $\hat\pi^A$ cannot be proportional
to each other on any open subset of $\Sigma$, and hence $\hat\sigma^A\bar
{\hat\sigma}^{A'}+\hat\pi^A\bar{\hat\pi}^{A'}$ is future pointing and
\emph{timelike} on an open dense subset of $\Sigma$. But then by
(\ref{eq:4.6.8b}) and the dominant energy condition it follows that $\hat
T_{ab}=0$ on $\Sigma$.

Evaluating the integrability condition of the system (\ref{eq:4.6.8a}) and
using that $\hat\sigma^A$ and $\hat\pi^A$ can be proportional with each
other only on closed subsets of $\Sigma$ with empty interior, we find that
the curvature $\hat R_{ABcd}$ of the spacetime at the points of $\Sigma$ is
that of the de Sitter spacetime (see \cite{Sz13}). Foliating the domain of
dependence of $\Sigma$ by smooth Cauchy surfaces $\Sigma_s$ by Lie dragging
$\Sigma$ along its own timelike $\hat g_{ab}$-unit normals, by the Bianchi
identities (written in their 3+1 form with respect to this foliation by
Friedrich \cite{Fr}) we obtain that the geometry of the domain of dependence
is locally isometric to the de Sitter spacetime. (Note that these Cauchy
surfaces $\Sigma_s$ for the domain of dependence of $\Sigma$ are \emph{not}
the hypersurfaces $\Sigma_t$ of subsection \ref{sub-4.4}. In fact, while all
the $\Sigma_s$ cut $\mathscr{I}^+$ in the \emph{same} 2-surface ${\cal S}$,
the surfaces $\Sigma_t\cap\mathscr{I}^+$ foliate a neighbourhood of ${\cal S}$
in $\mathscr{I}^+$.)

Finally, since the domain of dependence of $\Sigma$ is locally de Sitter, the
spinor fields $\hat\sigma^A$ and $\bar{\hat\pi}^{A'}$ provide a correct
initial condition for (\ref{eq:4.5.14}) on $\Sigma$. However, by the results
of subsection \ref{sub-4.5.3}, its solution is completely determined by the
value of the solution e.g. at a point of $\Sigma\cap\mathscr{I}^+$. Since
both $\hat\sigma^A$ and $\hat\pi^A$ are vanishing there, the whole solution
on $\Sigma$ must be vanishing. Hence, the differential operator

\begin{equation}
\tilde{\cal D}:C^\infty\bigl(\Sigma,\mathbb{D}^\alpha\bigr)\cap L_2\bigl(\Sigma,
\mathbb{D}^\alpha\bigr)\rightarrow C^\infty\bigl(\Sigma,\mathbb{D}^\alpha\bigr):
\Phi^\alpha\mapsto\tilde{\cal D}^\alpha{}_\beta\Phi^\beta \label{eq:4.6.9}
\end{equation}
is an \emph{injective} linear map on the space of the smooth square integrable
Dirac spinor fields, where $\Phi^\alpha:=(\hat\sigma^A,\bar{\hat\pi}^{A'})$ (as
a column vector) and

\begin{equation}
\tilde{\cal D}^\alpha{}_\beta\Phi^\beta:=\bigl(\hat{\cal D}^A{}_{B'}\bar{\hat\pi}
^{B'}+\frac{3}{2}K\hat\sigma^A,\hat{\cal D}^{A'}{}_B\hat\sigma^B+\frac{3}{2}K
\bar{\hat\pi}^{A'}\bigr). \label{eq:4.6.10}
\end{equation}
Therefore, if the system (\ref{eq:4.6.5a})-(\ref{eq:4.6.5b}) admits a smooth
solution, then that is unique in $C^\infty(\Sigma,\mathbb{D}^\alpha)$.

To show the existence of a solution of the \emph{inhomogeneous}
(\ref{eq:4.6.5a})-(\ref{eq:4.6.5b}), and also that non-smooth solutions of
(\ref{eq:4.6.6}) do not exist either, we should reformulate the problem in
appropriate function spaces and use certain functional analytic techniques.
We start this here and defer the details to the appendix. If $(\hat\sigma^A,
\bar{\hat\pi}^{A'})$ were a solution of (\ref{eq:4.6.5a})-(\ref{eq:4.6.5b})
with \emph{differentiable} extension of the corresponding $\hat\alpha^A$ and
$\hat\beta^A$ to the conformal boundary, then $(\hat\sigma^A,\bar{\hat\pi}
^{A'})=O(\Omega^2)$ would hold (see equation (\ref{eq:4.5.13})). Hence, for
the Dirac spinor $\Phi^\alpha=(\hat\sigma^A,\bar{\hat\pi}^{A'})$ we would have
that

\begin{eqnarray*}
\vert\Phi^\alpha\vert^2\!\!\!\!&:=\!\!\!\!&\sqrt{2}\hat t_{AA'}\bigl(
  \hat\sigma^A\bar{\hat\sigma}^{A'}+\hat\pi^A\bar{\hat\pi}^{A'}\bigr)=
  O(\Omega^3), \\
\vert\hat{\cal D}_e\Phi^\alpha\vert^2\!\!\!\!&:=\!\!\!\!&-\hat h^{ef}\sqrt{2}
  \hat t_{AA'}\Bigl((\hat{\cal D}_e\hat\sigma^A)(\hat{\cal D}_f\bar{\hat\sigma}
  ^{A'})+(\hat{\cal D}_e\hat\pi^A)(\hat{\cal D}_f\bar{\hat\pi}^{A'})\Bigr)=
  O(\Omega^3).
\end{eqnarray*}
(In the second of these, we used how $\hat{\cal D}_e$ is related to the Sen
derivative operator ${\cal D}_e$ in the conformal geometry, and that $\hat
\sigma^A=\Omega^2{}_0\sigma^A$ and $\hat\pi^A=\Omega^2{}_0\pi^A$ hold for some
smooth ${}_0\sigma^A$ and ${}_0\pi^A$ on $M$.) Thus, both $\Omega^{-2\delta}
\vert\Phi^\alpha\vert^2$ and $\Omega^{-2\delta}\vert\hat{\cal D}_e\Phi^\alpha
\vert^2$ would be integrable on $\Sigma$ for the \emph{same} $\delta<
\frac{1}{2}$, and hence

\begin{equation}
\bigl(\Vert\Phi^\alpha\Vert_{1,\delta}\bigr)^2:=\int_\Sigma\Omega^{-2\delta}\bigl(
\vert\Phi^\alpha\vert^2+\vert\hat{\cal D}_e\Phi^\alpha\vert^2\bigr){\rm d}\hat
\Sigma<\infty. \label{eq:4.6.11}
\end{equation}
However, the norm that (\ref{eq:4.6.11}) defines is \emph{not} the weighted
Sobolev norm (for the latter see e.g. \cite{ChCh}). This is the classical
Sobolev norm\footnote{Strictly speaking, the dimensionally correct norm
would be the square root of the integral of $\Omega^{-2\delta}(\vert\Phi
^\alpha\vert^2+L^2\vert\hat{\cal D}_e\Phi^\alpha\vert^2)$, where $L$ is a
positive constant with \emph{length} physical dimension, e.g. $L=1/\sqrt{
\Lambda}$. Since, however, it is the \emph{topology} of the Banach spaces
that the norm defines that has significance (but not the norm itself), we
adopt the standard (but physically incorrect) definition of the Sobolev norms.
This yields formally incorrect sums of quantities with different physical
dimension in certain estimates.} with respect to the \emph{weighted volume
element} $\Omega^{-2\delta}{\rm d}\hat\Sigma$. The weighted Sobolev spaces do
not appear to be the natural function spaces on the asymptotically
hyperboloidal $\Sigma$ because the fields and their derivatives have the
\emph{same} fall-off properties. However, this fall-off rate cannot be
arbitrary: The investigations in subsection \ref{sub-4.6.1} show that the
spinor fields $\hat\sigma^A$ and $\hat\pi^A$ \emph{must} fall off
\emph{faster} than $\Omega^{\frac{3}{2}}$, i.e. they must be square integrable
(with $\delta=0$), otherwise $H[\alpha,\bar\alpha]+H[\beta,\bar\beta]$ would
not be finite. The spaces $H_{s,\delta}(\Sigma,\mathbb{D}^\alpha)$ (or simply
$H_{s,\delta}$) of the Dirac spinor fields for which the norm has the structure
(\ref{eq:4.6.11}) with the number of derivatives $s=0,1,2,...$ will be
discussed in Appendix \ref{sub-A.1}.

Now we show that the Dirac spinor $(\hat\rho^A,\bar{\hat\omega}^{A'})$ belongs
to the weighted Lebesgue spaces $L^\delta_2=H_{0,\delta}$ for $\delta<\frac{1}
{2}$. Since ${}_0\hat\alpha^A={}_0\alpha^A$ and ${}_0\bar{\hat\beta}^{A'}={}_0
\bar\beta^{A'}$ satisfy the algebraic boundary condition (\ref{eq:4.5.4}) and
since they were constructed from the solution of (\ref{eq:4.5.8}) and
(\ref{eq:4.5.9}), by equation (\ref{eq:4.5.10a}) (with the notations of
subsection \ref{sub-4.5.2}), we have that

\begin{eqnarray*}
\hat{\cal D}_{A'A}\,{}_0\hat\alpha^A+\frac{3}{2}\!\!\!\!&{}\!\!\!\!&K\,{}_0\bar
 {\hat\beta}_{A'}={\cal D}_{A'A}\,{}_0\alpha^A+\frac{3}{2}\,{}_0\bar\gamma_{A'}=
 \Delta_{A'A}\,{}_0\alpha^A-v_{A'A}v^e{\cal D}_e\,{}_0\alpha^A+
 \frac{3}{2}\,{}_0\bar\gamma_{A'}=O(\Omega), \\
K\hat{\cal D}_{AA'}\,{}_0\bar{\hat\beta}^{A'}+\!\!\!\!&{}\!\!\!\!&\frac{3}{2}
 K^2\,{}_0\hat\alpha_A=\vert D_e\Omega\vert\, v_B{}^{A'}\bigl({\cal D}_{A'A}\,
 {}_0\alpha^B+\frac{1}{2}\,{}_0\bar\gamma_{A'}\delta^B_A\bigr)+O(\Omega)\\
=\!\!\!\!&{}\!\!\!\!&\vert D_e\Omega\vert\,\Bigl(v_A{}^{A'}\Delta_{A'B}\,{}_0
 \alpha^B-\frac{1}{2}v^e({\cal D}_e\,{}_0\alpha_A)-\frac{1}{2}v_{AA'}\,{}_0\bar
 \gamma^{A'}\Bigr)+O(\Omega)=O(\Omega).
\end{eqnarray*}
However, this fall-off means, in fact, that $(\hat\rho^A,\bar{\hat\omega}^{A'})
\in L^\delta_2$ for any $\delta<\frac{1}{2}$. Hence, it seems natural to expect
that (\ref{eq:4.6.10}) defines a bounded linear operator $\tilde{\cal D}$
from $H_{1,\delta}$ into $L^\delta_2$ with $\delta<\frac{1}{2}$, and we need to
show only that $(\hat\rho^A,\bar{\hat\omega}^{A'})\in{\rm Im}\,\tilde{\cal D}$
and that $\ker\tilde{\cal D}=\emptyset$. The former would imply the existence,
the latter the uniqueness of the solution of the renormalized Witten equation
(even among the non-smooth solutions). We complete the proof in the Appendix.

In fact, in Appendix \ref{sub-A.3} we show that the extension of $\tilde
{\cal D}$ from the space of the square integrable smooth Dirac spinor fields
to the first Sobolev space of the spinor fields, i.e. $\tilde{\cal D}:H_{1,0}
\rightarrow L_2$, is a topological vector space \emph{isomorphism} if the
hypersurface $\Sigma$ is chosen such that its `boost parameter function' $W$
satisfies $\frac{1}{3}\Lambda W<1$ (Theorem \ref{t:A.3.1}). Thus, in
particular, $\ker\tilde{\cal D}\subset H_{1,0}$ is empty, i.e. the homogeneous
equations (\ref{eq:4.6.6}) do not have even non-smooth square integrable
solutions with square integrable first derivative. Also, $(\hat\rho^A,\bar
{\hat\omega}^{A'})\in{\rm Im}\,\tilde{\cal D}$ holds, and hence
(\ref{eq:4.5.1}) has a unique square integrable smooth solution. However,
since $(\hat\rho^A,\bar{\hat\omega}^{A'})$ is not only square integrable but
belongs to the weighted Lebesgue spaces $L^\delta_2=H_{0,\delta}$ for $\delta<
\frac{1}{2}$, moreover the solution of (\ref{eq:4.6.5a})-(\ref{eq:4.6.5b})
is unique, the solution is not only square integrable, but belongs to $H
_{1,\delta}$ for any $\delta<\frac{1}{2}$. Therefore, the renormalized Witten
equation (\ref{eq:4.5.1}) with the algebraic boundary condition
(\ref{eq:4.5.4}) has a unique solution on such a $\Sigma$, proving that
$H[\alpha,\bar\alpha]+H[\beta,\bar\beta]$ is finite and non-negative. Since
$H[\alpha,\bar\alpha]+H[\beta,\bar\beta]$ depends only on the cut, its
finiteness and non-negativity are independent of the choice of $\Sigma$.

Finally, it could be worth noting that the positivity proof can be extended
to hypersurfaces with more than one asymptotic end; and also with inner
boundaries representing future marginally trapped surfaces, where the spinor
fields are subject to the chiral boundary conditions of \cite{Gibbonsetal}.

\subsubsection{Rigidity}
\label{sub-4.6.3}

It is easy to see that $H[\alpha,\bar\alpha]+H[\beta,\bar\beta]$ is vanishing
for any smooth cut of the conformal boundary of the de Sitter spacetime. In
fact, we saw in subsection \ref{sub-4.5.3} that in de Sitter spacetime
equation (\ref{eq:4.5.14}) admits four linearly independent solutions. Then
the restriction of the Weyl spinor constituents $\hat\alpha_A$ and $\hat\beta
_A$ of such a solution to any smooth spacelike hypersurface $\Sigma$ extending
to the conformal boundary $\mathscr{I}^+$ solve the 3-surface twistor equation
on $\Sigma$. Therefore, by (\ref{eq:4.6.3}), $H[\alpha,\bar\alpha]+H[\beta,
\bar\beta]$ is vanishing.

Now we show that the converse of this statement is, in some sense, also true:
If $H[\alpha,\bar\alpha]+H[\beta,\bar\beta]$ is zero, then the domain of
dependence of the spacelike hypersurface $\Sigma$ is isometric to an open
neighbourhood of a piece of the conformal boundary of the de Sitter spacetime.
Thus, the vanishing of $H[\alpha,\bar\alpha]+H[\beta,\bar\beta]$ is equivalent
to the local de Sitter nature of the spacetime near its future conformal
boundary. The present proof is an adaptation of the proof of an analogous
statement in \cite{Sz13}, and its logic is essentially the same that we
followed in proving the non-existence of smooth solutions of the homogeneous
equations (\ref{eq:4.6.6}) in the previous subsection. Thus here we only
sketch the key points of the proof.

Thus, let us suppose that $H[\alpha,\bar\alpha]+H[\beta,\bar\beta]=0$ for
some solution $(\hat\alpha_A,\bar{\hat\beta}_{A'})$ of (\ref{eq:4.5.1}). Then
both $\hat\alpha_A$ and $\hat\beta_A$ solve the 3-surface twistor equation,
$\hat{\cal D}_{(AB}\hat\alpha_{C)}=0$ and $\hat{\cal D}_{(AB}\hat\beta_{C)}=0$,
too. Hence, by the Witten equations (\ref{eq:4.5.1}) and the decomposition
(\ref{eq:4.6.2}), these satisfy

\begin{equation}
\hat{\cal D}_e\hat\alpha^A+KP^{AA'}_e\bar{\hat\beta}_{A'}=0, \hskip 20pt
\hat{\cal D}_e\bar{\hat\beta}^{A'}+KP^{AA'}\hat\alpha_A=0.
\label{eq:4.6.12}
\end{equation}
Like in the previous subsection, their solutions $\hat\alpha^A$ and $\hat
\beta^A$ cannot be proportional with each other on any open subset of
$\Sigma$, and hence, by the dominant energy condition, (\ref{eq:4.6.3})
gives that $\hat T_{ab}=0$ on $\Sigma$. Also, the integrability conditions of
(\ref{eq:4.6.12}) yields that the domain of dependence of $\Sigma$ is
locally isometric to the de Sitter spacetime.

\subsection{The total energy-momentum}
\label{sub-4.7}

\subsubsection{The structure of the  2-surface twistor space}
\label{sub-4.7.1}

In subsection \ref{sub-4.6.1} we saw that the functional $H[\alpha,\bar\alpha]
+H[\beta,\bar\beta]$ of the solutions $(\hat\alpha^A,\bar{\hat\beta}^{A'})$ of
the renormalized Witten equation can be finite only if the spinor fields
solve the 2-surface twistor equations on ${\cal S}=\Sigma\cap\mathscr{I}^+$,
i.e.:

\begin{equation}
-{\cal T}^+(\alpha):={\edth}^\prime\alpha_1-\bar\sigma^0\alpha_0=0, \hskip 20pt
{\cal T}^-(\alpha):={\edth}\alpha_0+\sigma^0\alpha_1=0. \label{eq:4.7.1}
\end{equation}
It is known that on topological 2-spheres the 2-surface twistor equations
admit at least four, and in the generic case precisely four linearly
independent solutions \cite{Ba}. However, examples are known for topological
2-spheres on which the 2-surface twistor equations admit five independent
solutions \cite{Je}. We will show that the 2-surface twistor space, i.e. the
space $\ker{\cal T}:=\ker({\cal T}^-\oplus{\cal T}^+)$ of the solutions of
the 2-surface twistor equations on ${\cal S}\subset\mathscr{I}^+$, is even
dimensional.

First, let us observe that the \emph{constant} normal $N_a$ of $\mathscr{I}
^+$ yields a non-trivial \emph{extra structure} on the kernel of a number of
differential operators. Indeed, for any spinor field $\lambda_A$ on
${\cal S}$ let us form the $\mathbb{C}$-anti-linear map

\begin{equation}
\nu:\lambda_A\mapsto\nu(\lambda)_A:=\sqrt{2}N_A{}^{A'}\bar\lambda_{A'},
\label{eq:4.7.2}
\end{equation}
i.e. in terms of spinor components $\nu:(\lambda_0,\lambda_1)\mapsto(-\bar
\lambda_{1'},\bar\lambda_{0'})$. Then the algebraic boundary condition
(\ref{eq:4.5.4}) is simply $\beta_A=\pm{\rm i}\nu(\alpha)_A$. Then it is a
simple calculation to check that this map yields the $\mathbb{C}$-anti-linear
\emph{isomorphisms}

\begin{equation}
\ker{\cal T}\rightarrow\ker{\cal T}, \hskip 15pt
\ker\Delta\rightarrow\ker\Delta, \hskip 15pt
\ker{\cal H}^+\rightarrow\ker{\cal H}^-, \hskip 15pt
\ker{\cal C}^+\rightarrow\ker{\cal C}^-. \label{eq:4.7.3}
\end{equation}
Here $\Delta:=\Delta^+\oplus\Delta^-$, ${\cal H}^\pm:=\Delta^\pm\oplus{\cal T}
^\pm$ and ${\cal C}^\pm:=\Delta^\pm\oplus{\cal T}^\mp$; and where

\begin{equation}
\Delta^+(\lambda):={\edth}^\prime\lambda_0+\rho^0\lambda_1, \hskip 20pt
-\Delta^-(\lambda):={\edth}\lambda_1-\rho^0\lambda_0. \label{eq:4.7.4}
\end{equation}
Thus, $\ker\Delta$ is the kernel of the Dirac operator built from the
2-dimensional Sen connection $\Delta_a:=\Pi^b_a\nabla_b$ on ${\cal S}$;
$\ker{\cal H}^\pm$ is the space of the holomorphic/anti-holomorphic spinor
fields of Dougan and Mason \cite{DM}; while, with the $\sigma^0=0$
substitution, $\ker{\cal C}^\pm$ is the space of Bramson's spinors \cite{Br}
at the future/past null infinity of asymptotically flat spacetimes (where
the relevant shears fall off \emph{faster} then the divergences). (For a
more detailed discussion of these operators, see the appendix of \cite{Sz01}.)

In particular, $\nu$ takes solutions of the 2-surface twistor equation into
solutions. Clearly, $\nu(\alpha)_A$ is not proportional to $\alpha_A$ (and
hence it is linearly independent of $\alpha_A$) because $\nu(\alpha)_A\alpha^A
=-\sqrt{2}N_{AA'}\alpha^A\bar\alpha^{A'}$, which is zero only if $\alpha_A$
itself is vanishing. Moreover, the spinor fields $\nu(\alpha)_A$ and $\alpha
_A$ are orthogonal to each other with respect to $N^{AA'}$, and $\nu^2=-{\rm
Id}$ holds. Hence, each solution $\alpha_A$ has a naturally determined
linearly independent counterpart $\nu(\alpha)_A$ and the $\alpha_A
\leftrightarrow\nu(\alpha)_A$ correspondence is one-to-one. Therefore, on
2-surfaces in $\mathscr{I}^+$, $\ker{\cal T}$ is necessarily even dimensional,
i.e. no odd number of `extra' solutions can exist. We can form the quotient
$\ker{\cal T}/\nu$, which, since $\nu$ is an isomorphism, can be identified
with the space of the complex 2-planes $[\alpha_A]$ of $\ker{\cal T}$ spanned
by $\alpha_A$ and $\nu(\alpha)_A$ for any $\alpha_A\in\ker{\cal T}$. This
$\ker{\cal T}/\nu$ is at least two, and generically is precisely two complex
dimensional. We assume that there are no `extra' solutions of the 2-surface
twistor equations on ${\cal S}$, and hence that $\ker{\cal T}$ is precisely
four, and hence that $\ker{\cal T}/\nu$ is precisely two dimensional.

Let us define $G(\ker{\cal T},\nu)$ to be the set of automorphisms $\Phi$
of $\ker{\cal T}$ for which $\nu\circ\Phi=\Phi\circ\nu$. Clearly, this is a
subgroup of $GL(\ker{\cal T})$, which can be called the symmetry group of
the 2-surface twistor space. Let ${\bf S}_{\bA}\subset\ker{\cal T}$ be a two
dimensional subspace which is not invariant under $\nu$, and hence for
which $\ker{\cal T}={\bf S}_{\bA}\oplus\nu({\bf S}_{\bA})$ holds. (This ${\bf
S}_{\bA}$ is a representative of $\ker{\cal T}/\nu$, and obviously it is
not canonically defined.) Let us fix a basis in ${\bf S}_{\bA}$, say $\{\alpha
^A_{\bA}\}$, ${\bA}={\bf 0},{\bf 1}$. (Thus, note that the \emph{boldface}
name indices refer to a basis in the abstract solution space, while the
\emph{underlined} name indices to a frame \emph{field} on ${\cal S}$. Note
also that since $\nu$ contains complex conjugation and it is only the abstract
index that is converted by $N^A{}_{A'}$ to an unprimed one, the boldface index
in $\nu(\alpha_{{\bA}'})^A$ is, in fact, primed.) Then $\{\alpha^A_{\bA},\nu
(\alpha_{{\bA}'})^A\}$ is a basis in $\ker{\cal T}$, and in this basis $\nu$
takes the form of a $4\times4$ complex matrix with $2\times2$ blocks $0$,
$-\delta^{\bA}_{\bB}$, and, in the second row, $\delta^{{\bA}'}_{{\bB}'}$ and $0$.
(Here $0$ is the zero matrix.) Thus, since $\nu$ is \emph{anti}-linear, its
action in this basis is the matrix multiplication with this matrix, following
the complex conjugation. Hence, in this basis, $\Phi\in G(\ker{\cal T},\nu)$
is a matrix of the form

\begin{equation}
\Phi=\left(\begin{array}{cc}
      A^{\bA}{}_{\bB} & B^{\bA}{}_{{\bB}'} \\
 -\bar{B}^{{\bA}'}{}_{\bB} & \bar{A}^{{\bA}'}{}_{{\bB}'} \\
      \end{array}\right), \label{eq:4.7.5}
\end{equation}
where $A$ and $B$ are $2\times2$ complex matrices and $A$ is nonsingular.
Clearly, these matrices can be factorized in a unique way according to

\begin{equation}
\Phi=\left(\begin{array}{cc}
      A^{\bA}{}_{\bB} & 0 \\
      0 & \bar{A}^{{\bA}'}{}_{{\bB}'} \\
      \end{array}\right)
\left(\begin{array}{cc}
       \delta^{\bB}_{\bC} & C^{\bB}{}_{{\bC}'} \\
 -\bar{C}^{{\bB}'}{}_{\bC} & \delta^{{\bB}'}_{{\bC}'} \\
      \end{array}\right); \label{eq:4.7.5a}
\end{equation}
and the matrices of the form ${\rm diag}(A,\bar A)$ form a subgroup in $G
(\ker{\cal T},\nu)$, which is isomorphic to $GL(1,\mathbb{C})\times SL(2,
\mathbb{C})$. The determinant of these factors is $\vert\det(A)\vert^2$ and
$1+{\rm Tr}(C\bar C)+\vert\det(C)\vert^2\geq0$, respectively. Hence the
latter should be required to be positive. Therefore, the symmetry group
$G(\ker{\cal T},\nu)$ factorized by the multiplicative group of the
determinants $\det(\Phi)$ is a 15 parameter subgroup of $SL(4,\mathbb{C})$.
This subgroup turns out to be isomorphic to $SL(2,\mathbb{H})$, the spin
group of $SO(1,5)$.

To see this, first let us determine its Lie algebra. Let us denote the
standard $SL(2,\mathbb{C})$ Pauli matrices (divided by $\sqrt{2}$) by
$\sigma^{\bA{\bB}'}_{\ua}$, ${\ua}=0,...,3$; and, for ${\bi}=1,2,3$, let
$\sigma^{\bA}_{\bi}{}_{\bB}:=\sqrt{2}\sigma^{{\bA}{\bA}'}_{\bi}\sigma^0_{{\bA}'{\bB}}$,
which are the standard $SU(2)$ Pauli matrices (also divided by $\sqrt{2}$).
Then the basis of the Lie algebra of $G(\ker{\cal T},\nu)$ corresponding to
the factorization (\ref{eq:4.7.5a}) is

\begin{eqnarray*}
&{}&d:=\frac{1}{2}\left(\begin{array}{cc}
      \delta^{\bA}_{\bB} & 0 \\
                    0  & \delta^{{\bA}'}_{{\bB}'} \\
      \end{array}\right), \hskip 55pt
a_0:=\frac{\rm i}{2}\left(\begin{array}{cc}
      \delta^{\bA}_{\bB} & 0 \\
                    0  & -\delta^{{\bA}'}_{{\bB}'} \\
      \end{array}\right), \\
&{}&a_{\bi}:=\frac{1}{\sqrt{2}}\left(\begin{array}{cc}
      \sigma^{\bA}_{\bi}{}_{\bB} & 0 \\
                    0  & \sigma^{{\bA}'}_{\bi}{}_{{\bB}'} \\
      \end{array}\right), \hskip 22pt
\tilde a_{\bi}:=\frac{\rm i}{\sqrt{2}}\left(\begin{array}{cc}
      \sigma^{\bA}_{\bi}{}_{\bB} & 0 \\
                    0  & -\sigma^{{\bA}'}_{\bi}{}_{{\bB}'} \\
      \end{array}\right), \\
&{}&c_{\ua}:=\frac{1}{\sqrt{2}}\left(\begin{array}{cc}
                    0 & \sigma^{\bA}_{\ua{\bB}'}  \\
        -\sigma^{{\bA}'}_{\ua{\bB}} & 0 \\
      \end{array}\right), \hskip 20pt
\tilde c_{\ua}:=\frac{\rm i}{\sqrt{2}}\left(\begin{array}{cc}
                    0 & \sigma^{\bA}_{\ua{\bB}'}  \\
        \sigma^{{\bA}'}_{\ua{\bB}} & 0 \\
      \end{array}\right).
\end{eqnarray*}
Clearly, $d$ commutes with all the basis elements and generates the change of
the determinant of $\Phi$. The Lie algebra elements $a_0$, $a_{\bi}$, $\tilde
a_{\bi}$, $c_{\ua}$ and $\tilde c_{\ua}$ form a 15 real dimensional Lie algebra
with the Lie products

\begin{eqnarray*}
&{}& [a_0,a_{\bi}]=[a_0,\tilde a_{\bi}]=0, \hskip 20pt
  [a_0,c_{\ua}]=\tilde c_{\ua}, \hskip 20pt
  [a_0,\tilde c_{\ua}]=-c_{\ua}, \\
&{}& [a_{\bi},a_{\bj}]=\epsilon_{\bi\bj}{}^{\bk}\tilde a_{\bk},
  \hskip 30pt
  [a_{\bi},\tilde a_{\bj}]=-\epsilon_{\bi\bj}{}^{\bk}a_{\bk}, \\
&{}& [a_{\bi},c_0]=c_{\bi}, \hskip 20pt
  [a_{\bi},c_{\bj}]=-\eta_{\bi\bj}c_0, \hskip 20pt
  [a_{\bi},\tilde c_0]=\tilde c_{\bi}, \hskip 20pt
  [a_{\bi},\tilde c_{\bj}]=-\eta_{\bi\bj}\tilde c_0, \\
&{}& [\tilde a_{\bi},\tilde a_{\bj}]=-\epsilon_{\bi\bj}{}^{\bk}\tilde a_{\bk},
  \hskip 20pt
  [\tilde a_{\bi},c_0]=[\tilde a_{\bi},\tilde c_0]=0, \hskip 20pt
  [\tilde a_{\bi},c_{\bj}]=-\epsilon_{\bi\bj}{}^{\bk}c_{\bk},  \hskip 20pt
  [\tilde a_{\bi},\tilde c_{\bj}]=-\epsilon_{\bi\bj}{}^{\bk}\tilde c_{\bk},\\
&{}& [c_0,c_{\bi}]=-a_{\bi}, \hskip 20pt
  [c_0,\tilde c_0]=-a_0, \hskip 20pt
  [c_0,\tilde c_{\bi}]=0, \\
&{}& [c_{\bi},c_{\bj}]=-\epsilon_{\bi\bj}{}^{\bk}\tilde a_{\bk}, \hskip 20pt
  [c_{\bi},\tilde c_0]=0, \hskip 20pt
  [c_{\bi},\tilde c_{\bj}]=-\eta_{\bi\bj}a_0, \\
&{}& [\tilde c_0,\tilde c_{\bi}]=-a_{\bi}, \hskip 20pt
 [\tilde c_{\bi},\tilde c_{\bj}]=-\epsilon_{\bi\bj}{}^{\bk}\tilde a_{\bk}.
\end{eqnarray*}
(Note that we lower the small boldface indices by the \emph{negative
definite} $\eta_{\bi\bj}=-\delta_{\bi\bj}$, the capital boldface indices by
the anti-symmetric Levi-Civita symbol $\epsilon_{\bA\bB}$; and we raise them
by their inverses.)

To see the structure of this Lie algebra, consider its non-trivial
subalgebras. First, $\{\tilde a_{\bi},a_{\bi}\}$ spans the Lorentz Lie algebra
$so(1,3)\simeq sl(2,\mathbb{C})$, and both $\{\tilde a_{\bi},c_{\bi}\}$ and
$\{\tilde a_{\bi},\tilde c_{\bi}\}$ (which are isomorphic with each other) span
$so(4)\simeq su(2)\oplus su(2)$. The subalgebras spanned by $\{\tilde a_{\bi},
a_{\bi},c_{\bi},c_0\}$ and by $\{\tilde a_{\bi},a_{\bi},\tilde c_{\bi},\tilde c_0\}$
are also isomorphic, and since they contain $so(1,3)$ and $so(4)$, it is
natural to expect them to be just the de Sitter algebra $so(1,4)\simeq
sp(1,1)$. Similarly, $\{\tilde a_{\bi},c_{\bi},\tilde c_{\bi},a_0\}$ contains
two (overlapping) copies of $so(4)$, hence this can be expected to be just
$so(5)\simeq sp(2)$. It is straightforward (e.g. by explicit calculations)
to show that these are, in fact, isomorphic to the $so(1,4)$ and $so(5)$
Lie algebras, respectively. In addition, $\{a_0,c_0,\tilde c_0\}$ spans
$so(1,2)\simeq sl(2,\mathbb{R})$. Therefore, the Lie algebra of the symmetry
group $G(\ker{\cal T},\nu)$ should be $so(1,5)\oplus\mathbb{R}$.

Since $so(1,5)$ is isomorphic to $sl(2,\mathbb{H})$, the Lie algebra of the
special linear group in two dimensions over the quaternions \cite{LM},
$G(\ker{\cal T},\nu)$ factorized by the determinants is locally isomorphic
to $SL(2,\mathbb{H})$. In fact, the explicit form (\ref{eq:4.7.5}) of the
symmetry group $G(\ker{\cal T},\nu)$ is just the complex realization of
$GL(2,\mathbb{H})$ \cite{Ro}, i.e. the quotient of $G(\ker{\cal T},\nu)$ and
the multiplicative group $(0,\infty)$ of the determinants $\det(\Phi)$ is
\emph{precisely} $SL(2,\mathbb{H})$ (which in its actual complex form is also
denoted by $SU^*(4)$ \cite{Ro}), the spin group of $SO(1,5)$.

\subsubsection{The general form of the total energy-momentum}
\label{sub-4.7.2}

Since the solution $(\hat\alpha^A,\bar{\hat\beta}^{A'})$ of (\ref{eq:4.5.1}) is
completely controlled e.g. by the boundary value of $\alpha^A$ on the cut of
the conformal boundary, the functional $H$ yields a well defined positive
definite quadratic form on the space $\ker{\cal T}$ by

\begin{equation}
H^*:\ker{\cal T}\rightarrow[0,\infty):\alpha^A\mapsto\frac{1}{2}\Bigl(H\bigl[
\alpha,\bar\alpha\bigr]+H\bigl[\nu(\alpha),\overline{\nu(\alpha)}\bigr]\Bigr).
\label{eq:4.7.6}
\end{equation}
Then the polarization formula (\ref{eq:4.1p}) (applied to $H^*$) makes it
possible to extend $H^*$ to be a positive Hermitian bilinear form on $\ker
{\cal T}$. This gives, for any $\lambda^A,\mu^A\in\ker{\cal T}$, that

\begin{equation}
H^*\bigl[\lambda,\bar\mu\bigr]=\frac{1}{2}\Bigl(H\bigl[\lambda,\bar\mu\bigr]
+H\bigl[\nu(\mu),\overline{\nu(\lambda)}\bigr]\Bigr), \label{eq:4.7.7}
\end{equation}
by means of which it is easy to see that

\begin{equation}
H^*\bigl[\lambda,\overline{\nu(\mu)}\bigr]=-H^*\bigl[\mu,\overline{\nu(
\lambda)}\bigr], \hskip 20pt
H^*\bigl[\nu(\lambda),\overline{\nu(\mu)}\bigr]=\overline{H^*\bigl[\lambda,
\bar\mu\bigr]} \label{eq:4.7.8}
\end{equation}
hold. In particular, by the first of these $H^*[\lambda,\overline{\nu(
\lambda)}]=0$ for any $\lambda^A\in\ker{\cal T}$.

Let us fix a basis $\{\alpha^A_{\bA},\nu(\alpha_{{\bA}'})^A\}$ in $\ker{\cal T}$
(see subsection \ref{sub-4.7.1}). Then by (\ref{eq:4.7.8}) $H^*$ in this
basis is a $4\times4$ complex matrix

\begin{equation}
H^*=\left(\begin{array}{cc}
      {\tt P} & {\tt Q} \\
 -\bar{\tt Q} & \bar{\tt P} \\
      \end{array}\right), \label{eq:4.7.9}
\end{equation}
where ${\tt P}:=H^*[\alpha_{\bA},\bar\alpha_{{\bB}'}]$ (or, rather ${\tt P}
_{{\bA}{\bB}'}$) is a $2\times2$ Hermitian, while ${\tt Q}:=H^*[\alpha_{\bA},
\overline{\nu(\alpha_{{\bB}'})}]=H^*[\alpha_{\bA},\overline{\nu}(\alpha_{\bB})]$
is a $2\times2$ complex anti-symmetric matrix. By the positivity and rigidity
results of subsection \ref{sub-4.6} the matrix ${\tt P}$ is positive, and
one of its diagonal elements (i.e. the product ${\tt P}_{{\bf 0}{\bf 0}'}{\tt P}
_{{\bf 1}{\bf 1}'}$) is vanishing if and only if the domain of dependence of the
hypersurface $\Sigma$ is locally isometric to the de Sitter spacetime, in
which case the whole $H^*$ is vanishing. By the anti-symmetry of ${\tt Q}$ it
can also be written as $Q\epsilon_{\bA\bB}$, where $Q:=H^*[\alpha_{\bf 0},\bar
\nu(\alpha_{\bf 1})]\in\mathbb{C}$.

Since the choice for ${\bf S}_{\bA}$ is \emph{not} canonical, neither the
basis nor ${\tt P}_{{\bA}{\bB}'}$ is canonically defined. Under the action of
a basis transformation with $\Phi\in G(\ker{\cal T},\nu)$ given by
(\ref{eq:4.7.5}), the components ${\tt P}_{{\bA}{\bB}'}$ and $Q$ change as

\begin{eqnarray*}
&{}&{\tt P}_{{\bA}{\bA}'}\mapsto{\tt P}_{{\bB}{\bB}'}\bigl(A^{\bB}{}_{\bA}
 \bar{A}^{{\bB}'}{}_{{\bA}'}+B^{\bB}{}_{{\bA}'}\bar{B}^{{\bB}'}{}_{\bA}\bigr)
 -Q\,\epsilon_{\bC\bD}A^{\bC}{}_{\bA}B^{\bD}{}_{{\bA}'}-\bar{Q}\,
 \epsilon_{{\bC}'{\bD}'}\bar{A}^{{\bC}'}{}_{{\bA}'}\bar{B}^{{\bD}'}{}_{\bA}, \\
&{}&Q\mapsto Q\,\det(A)+\bar{Q}\,\det(B)+{\tt P}_{{\bA}{\bA}'}A^{\bA}{}_{\bC}
 \bar{B}^{{\bA}'}{}_{\bD}\epsilon^{\bC\bD}.
\end{eqnarray*}
Hence, in the lack of any \emph{further} extra structure on $\ker{\cal T}$,
it does not seem to be able to extract ${\tt P}_{{\bA}{\bB}'}$ as the
energy-momentum 4-vector in a \emph{canonical} way from $H^*$, even though it
shares the positivity and rigidity properties. Such an extra structure would
be needed to reduce $G(\ker{\cal T},\nu)$ to $GL(1,\mathbb{C})\times SL(2,
\mathbb{C})=\{{\rm diag}(A,\bar A)\}$. In this case we would have a well
defined energy-momentum ${\tt P}_{{\bA}{\bB}'}$. However, even if we had such
a reduction, still we would not have any natural symplectic metric to define
mass as the length of this energy-momentum. Still we would have to rule out
the factor $GL(1,\mathbb{C})$.

A simple calculation shows that

\begin{equation}
\det(H^*)=\bigl(\det({\tt P}_{{\bA}{\bB}'})-Q\bar Q\bigr)^2, \label{eq:4.7.10}
\end{equation}
and hence it might be tempting to introduce the concept of mass by $\det
(H^*)$, even if ${\tt P}_{{\bA}{\bB}'}$ cannot be defined in a canonical way.
Since $H^*$ is positive definite, by the rigidity property $\det(H^*)=0$ (i.e.
$\det({\tt P}_{{\bA}{\bB}'})=Q\bar Q$) is equivalent to the vanishing of $H^*$,
i.e. to the local de Sitter nature of the domain of dependence of the
spacelike hypersurface $\Sigma$ in the spacetime. Unfortunately, however,
$\det(\Phi)$ is not one, and hence this determinant is still \emph{not} an
invariant, it is only its sign (positive or zero) that is invariant. We would
need an extra structure on the 2-surface twistor space, e.g. a geometrically
defined \emph{volume 4-form}, by means of which the symmetry group could be
reduced to $SL(2,\mathbb{H})$.

Nevertheless, still $H^*$ is a well defined observable with useful properties,
but in general asymptotically de Sitter spacetimes this cannot be
\emph{interpreted} as energy-momentum in a natural way. In the next
subsection we consider special cases when additional extra structures are
present on the 2-surface twistor space by means of which the analog of the
Bondi mass can be defined by (\ref{eq:4.7.10}).

\subsubsection{Further extra structures on $\ker{\cal T}$ on non-contorted
cuts}
\label{sub-4.7.3}

The 2-surface twistor space can also be considered as the space of the pairs
${\tt Z}^\alpha:=(\lambda^A,{\rm i}\Delta_{A'B}\lambda^B)$, $\lambda^A\in\ker
{\cal T}$; and $\pi_{A'}:={\rm i}\Delta_{A'B}\lambda^B$ is called the secondary
part of the 2-surface twistor ${\tt Z}^\alpha$ (see \cite{PR2}). For any
four twistors ${\tt Z}^\alpha_i=(\lambda^A_i,\pi^i_{A'})$, $i=1,...,4$, one can
define

\begin{equation*}
\varepsilon:=\frac{1}{4}\varepsilon_{\alpha\beta\gamma\delta}{\tt Z}^\alpha_1
{\tt Z}^\beta_2{\tt Z}^\gamma_3{\tt Z}^\delta_4:=\epsilon^{ij}{}_{kl}
\lambda^0_i\lambda^1_j\pi^k_{0'}\pi^l_{1'}=\epsilon^{ijkl}\lambda^0_i
\lambda^1_j\bigl({\edth}\lambda^0_k\bigr)\bigl({\edth}'\lambda^1_l\bigr),
\end{equation*}
where $\epsilon^{ijkl}$ is the Levi-Civita alternating symbol, and we used the
GHP form of the secondary part of the twistors. In general, this $\varepsilon$
is a complex valued \emph{function} on the cut. If, however, $\varepsilon$
were \emph{constant} on the cut, then this $\varepsilon_{\alpha\beta\gamma\delta}$
would define a volume 4-form, a further \emph{extra structure}, on the
2-surface twistor space $\ker{\cal T}$. The presence of this volume form
would reduce the symmetry group to its volume-preserving subgroup, i.e. to
$SL(2,\mathbb{H})$. Then we could define the mass in an invariant way by
(\ref{eq:4.7.10}), where the basis $\{\alpha^A_{\bA},\nu(\alpha_{{\bA}'})^A\}$
would be chosen such that the components of $\varepsilon_{\alpha\beta\gamma\delta}$
are those of the Levi-Civita symbol. It is still not known what are the
necessary and sufficient conditions on the geometrical properties of the cut
that could ensure the existence of such a volume 4-form.

If a Hermitian metric exists on $\ker{\cal T}$, then it defines a
geometrically given volume 4-form (see \cite{PR2}). To introduce this, for
any two twistors ${\tt Z}^\alpha=(\lambda^A,\pi_{A'})$ and ${\tt W}^\alpha=
(\mu^A,\rho_{A'})$ let us define

\begin{equation*}
{\tt h}_{\alpha\beta'}{\tt Z}^\alpha\bar{\tt W}^{\beta'}:=\lambda^A\bar\rho_A+
\pi_{A'}\bar\mu^{A'}=: h(\lambda,\bar\mu).
\end{equation*}
In general, this is \emph{not} constant on the cut, but when it is, then it
defines a (conformally invariant) Hermitian metric on $\ker{\cal T}$ with
signature $(+,+,-,-)$. By definition, the group of the linear transformations
of $\ker{\cal T}$ preserving this metric is $SU(2,2)$, the spin group of
$SO(2,4)$. Hence, in the presence of such a Hermitian metric, the symmetry
group of the 2-surface twistor space is $SL(2,\mathbb{H})\cap SU(2,2)\simeq
SP(1,1)$, just the spin group of the de Sitter group $SO(1,4)=SO(1,5)\cap
SO(2,4)$. Thus, the de Sitter group (in its spinor representation) emerges
as the (reduced) symmetry group of the 2-surface twistor space. The
2-surfaces for which such a Hermitian scalar product ${\tt h}_{\alpha\alpha'}$
exists are called \emph{non-contorted}, and these are known to be just the
2-surfaces which can be embedded, at least locally, into a conform Minkowski
spacetime with their first and second fundamental forms \cite{Tod86,Je84,Je87}.
For example, the cuts of an \emph{intrinsically locally conformally flat}
conformal boundary are all non-contorted (see also \cite{Tod83}). Thus, we
have a large class of radiative spacetimes in which the cuts are
non-contorted.

Thus, suppose that the cut ${\cal S}$ is non-contorted and hence the
Hermitian metric ${\tt h}_{\alpha\alpha'}$ exists on $\ker{\cal T}$. Then a
straightforward calculation yields that

\begin{equation}
h(\nu(\lambda),\overline{\nu(\mu)})=\overline{h(\lambda,\bar\mu)},
\hskip 20pt
h(\lambda,\overline{\nu(\mu)})=-h(\mu,\overline{\nu(\lambda)});
\label{eq:4.7.11}
\end{equation}
i.e. the Hermitian metric is anti-invariant under the action of $\nu$. (It
might be worth noting that these equations already ensure the existence of
the infinity twistor on $\ker{\cal T}$ \cite{Tod15}.) Thus, in particular,
$h(\lambda,\overline{\nu(\lambda)})=0$ for any $\lambda^A$; i.e. $\lambda^A$
and $\nu(\lambda)^A$ are orthogonal to each other with respect to $h$, too.
The properties (\ref{eq:4.7.11}) make it possible to choose the basis
$\{\alpha^A_{\bA},\nu(\alpha_{{\bA}'})^A\}$ in $\ker{\cal T}$ in a more specific
way. Namely, let us choose $\alpha^A_{\bf 0}$ such that $h(\alpha_{\bf 0},\bar
\alpha_{{\bf 0}'})=1$. Then by (\ref{eq:4.7.11}) $\nu(\alpha_{{\bf 0}'})^A$ is
$h$-orthogonal to $\alpha^A_{\bf 0}$; and its norm is also 1. Thus the spinors
in the 2-plane $[\alpha^A_{\bf 0}]$ have positive norm. Then let us choose
$\alpha^A_{\bf 1}$ to be $h$-orthogonal to the 2-plane $[\alpha^A_{\bf 0}]$.
Because of the signature of ${\tt h}_{\alpha\alpha'}$ its norm is negative, and
we choose it to be $-1$. Then by (\ref{eq:4.7.11}) the norm of $\nu(\alpha
_{{\bf 1}'})^A$ is $-1$, and it is $h$-orthogonal not only to $\alpha^A_{\bf 1}$,
but to the whole 2-plane $[\alpha^A_{\bf 0}]$, too. Hence, the resulting basis
$\{\alpha^A_{\bA},\nu(\alpha_{{\bA}'})^A\}$ is $h$--orthonormal. Thus, in this
basis, $h$ and its contravariant form $h^{-1}$ have the form ${\rm diag}
(1,-1,1,-1)$. Hence, the invariants of $H^*$ are given by the trace of the
first four powers of $h^{-1}H^*$. These are $2({\tt P}_{{\bf 0}{\bf 0}'}-{\tt P}
_{{\bf 1}{\bf 1}'})$, $4(\det({\tt P}_{{\bA}{\bB}'})-Q\bar Q)+2({\tt P}
_{{\bf 0}{\bf 0}'}-{\tt P}_{{\bf 1}{\bf 1}'})^2$, $6({\tt P}_{{\bf 0}{\bf 0}'}
-{\tt P}_{{\bf 1}{\bf 1}'})(\det({\tt P}_{{\bA}{\bB}'})-Q\bar Q)+2({\tt P}
_{{\bf 0}{\bf 0}'}-{\tt P}_{{\bf 1}{\bf 1}'})^3$ and $2(2(\det({\tt P}_{{\bA}{\bB}'})-Q
\bar Q)+({\tt P}_{{\bf 0}{\bf 0}'}-{\tt P}_{{\bf 1}{\bf 1}'})^2)^2$, respectively.
Therefore, $H^*$ has two independent invariants, both of them real:

\begin{equation}
{\tt M}^2:=\det\bigl({\tt P}_{{\bA}{\bB}'}\bigr)-Q\bar Q, \hskip 30pt
{\tt N}:={\tt P}_{{\bf 0}{\bf 0}'}-{\tt P}_{{\bf 1}{\bf 1}'}.
\end{equation}
Clearly, ${\tt M}$ is analogous to the Bondi mass of asymptotically flat
spacetimes, but the meaning of the other invariant, the difference ${\tt N}$
of the value of the functional $H^*$ on the positive and negative norm
elements of the subspace ${\bf S}_{\bA}$, is still unclear. Further
investigation of these quantities, viz. their group theoretical properties,
their alternative expressions, the analog of `mass-loss', etc. will be given
in a separate paper.

\section{Appendix: The Sen--Witten operator in weighted
function spaces}
\label{sec-A}

The aim of this appendix is to introduce and develop the necessary functional
analytic tools by means of which we can prove the existence and uniqueness
of the solution of the renormalized Witten equation
(\ref{eq:4.6.5a})-(\ref{eq:4.6.5b}) on \emph{asymptotically hyperboloidal}
hypersurfaces in a rigorous way. The key ideas and statements are motivated
by those of \cite{ChCh} developed for \emph{asymptotically flat} Riemannian
manifolds. First we introduce the appropriate weighted function spaces and
state two of their properties. Then we prove a number of estimates for the
renormalized Sen--Witten operator $\tilde{\cal D}$, by means of which finally
we prove that $\tilde{\cal D}$ is an \emph{isomorphism}. In this appendix
all the quantities and objects are in the \emph{physical} spacetime, but,
for the sake of simplicity, we leave the `hats' off of them.

\subsection{The weighted function spaces on asymptotically
hyperboloidal hypersurfaces}
\label{sub-A.1}

Following the general ideas of \cite{ChCh} we define the weighted Lebesgue
spaces of pairs of Weyl spinor fields, i.e. of Dirac spinor fields, on the
asymptotically hyperboloidal hypersurfaces discussed in subsection
\ref{sub-4.4}. (The following concepts can be generalized in a natural way to
cross-sections of Hermitian vector bundles over asymptotically hyperboloidal
$n$-manifolds, even with more than one asymptotic end.) Thus, we assume that
all the geometric structures of those hypersurfaces $\Sigma$ are present and
are smooth. In particular, $h_{ab}$ is an asymptotically hyperboloidal
(negative definite) metric and $\chi_{ab}$ is the extrinsic curvature with the
asymptotic form (\ref{eq:4.4.2}) and (\ref{eq:4.4.3}), respectively. We also
consider the future pointing unit timelike normal (and hence the positive
definite Hermitian metric on the spinor spaces), the conformal factor
$\Omega$ (or radial coordinate $1/\Omega$) and the foliation ${\cal S}_\Omega$
etc. to be given. Clearly, the conformal factor can be assumed to be one on
some `large enough' compact subset $K\subset\Sigma$ with smooth boundary
$\partial K\approx S^2$ and strictly monotonically decreasing on $\Sigma-K$.
${\cal D}_e$ will denote the covariant derivative operator acting on the
Dirac spinor fields, determined by the Sen connection.

For $\delta\in\mathbb{R}$ and a measurable Dirac spinor field $\Phi^\alpha=
(\sigma^A,\bar\pi^{A'})$ on $\Sigma$ we define

\begin{equation}
\bigl(\Vert\Phi^\alpha\Vert_\delta\bigr)^2:=\int_\Sigma\Omega^{-2\delta}\sqrt{2}
t_{AA'}\bigl(\sigma^A\bar\sigma^{A'}+\pi^A\bar\pi^{A'}\bigr){\rm d}\Sigma
\label{eq:A.1.1}
\end{equation}
and let $L^\delta_2(\Sigma,\mathbb{D}^\alpha)$, or shortly $L^\delta_2$, denote
the space of the spinor fields $\Phi^\alpha$ for which $\Vert\Phi^\alpha\Vert
_\delta<\infty$. This space is a Banach space with the norm $\Vert\, .\,\Vert
_\delta$, which is, in fact, a Hilbert space with the obvious Hermitian scalar
product: $\langle\Phi^\alpha,\Psi^\alpha\rangle:=\int_\Sigma\Omega^{-2\delta}
\sqrt{2}t_{AA'}(\sigma^A\bar\alpha^{A'}+\bar\pi^{A'}\beta^A){\rm d}\Sigma$ for
any $\Phi^\alpha=(\sigma^A,\bar\pi^{A'})$ and $\Psi^\alpha=(\alpha^A,\bar\beta
^{A'})$.

With the convention $\Vert\Phi^\alpha\Vert_{0,\delta}:=\Vert\Phi^\alpha\Vert
_\delta$, for any $s=0,1,2,...$ let us define\footnote{See the footnote to
equation (\ref{eq:4.6.11}) in subsection \ref{sub-4.6.2}.}

\begin{equation}
\bigl(\Vert\Phi^\alpha\Vert_{s,\delta}\bigr)^2:=\sum^s_{k=0}\int_\Sigma\Omega
^{-2\delta}\vert {\cal D}_{e_k}...{\cal D}_{e_1}\Phi^\alpha\vert^2{\rm d}\Sigma
\label{eq:A.1.2}
\end{equation}
for any measurable spinor field $\Phi^\alpha$ with measurable ${\cal D}
_e$--derivatives (in the weak sense) up to order $s$, where

\begin{eqnarray}
\vert{\cal D}_{e_k}...{\cal D}_{e_1}\Phi^\alpha\vert^2:\!\!\!\!&=\!\!\!\!&
 \sqrt{2}t_{AA'}(-h^{e_1f_1})...(-h^{e_kf_k})\Bigl(({\cal D}_{e_k}...{\cal D}_{e_1}
 \sigma^A)({\cal D}_{f_k}...{\cal D}_{f_1}\bar\sigma^{A'})+ \nonumber \\
\!\!\!\!&+\!\!\!\!&({\cal D}_{e_k}...{\cal D}_{e_1}\pi^A)({\cal D}_{f_k}
 ...{\cal D}_{f_1}\bar\pi^{A'})\Bigr) \label{eq:A.1.3}
\end{eqnarray}
is the positive definite pointwise norm of the $k$th derivative. Then the
space $H_{s,\delta}(\Sigma,\mathbb{D}^\alpha)$, or shortly $H_{s,\delta}$, is
defined to be the space of those spinor fields $\Phi^\alpha$ for which $\Vert
\Phi^\alpha\Vert_{s,\delta}<\infty$. This is a Hilbert space with the obvious
scalar product. Note that the spaces $H_{s,\delta}$ are \emph{not} the familiar
weighted Sobolev spaces (see e.g. \cite{ChCh}), rather these are the classical
Sobolev spaces with the overall weighted volume element. The concept of this
kind of spaces is motivated by the observation that the fall-off rate of both
the solution of the renormalized Witten equation and its derivatives are
the same (see subsection \ref{sub-4.6.2}). By definition $H_{0,\delta}=L^\delta
_2$, and $H_{s,0}$ is just the classical Sobolev space $H_s(\Sigma,\mathbb{D}
^\alpha)$.

We also need to define the space $C^s_\delta(\Sigma,\mathbb{D}^\alpha)$ (or
simply $C^s_\delta$) of the $C^s$ Dirac spinor fields $\Phi^\alpha$ on $\Sigma$
for which\footnote{See the footnote to equation (\ref{eq:4.6.11}) in
subsection \ref{sub-4.6.2}.}

\begin{equation}
\Vert\Phi^\alpha\Vert_{C^s_\delta}:=\sup\Bigl\{\sum^s_{k=0}\Omega^{-\delta}\vert
{\cal D}_{e_k}...{\cal D}_{e_1}\Phi^\alpha\vert(p)\,\vert\,p\in\Sigma\,\Bigr\}
<\infty. \label{eq:A.1.4}
\end{equation}
$C^s_\delta$ is a Banach space with the norm $\Vert\,.\,\Vert_{C^s_\delta}$.
The concept of these spaces is motivated by the weighted function spaces of
$C^s$ tensor fields of \cite{ChCh}, but note that the norms are different and
hence the spaces $C^s_\delta$ are different from those of \cite{ChCh}. Here
the weight functions in front of the different order terms are the
\emph{same}. The significance of these spaces is that they give a control on
the fall-off properties of the spinor fields. In particular, $\Phi^\alpha$
is continuous and $\Phi^\alpha=O(\Omega^k)$ (in the sense that $\Omega^{-k}\Phi
^\alpha$ can be extended to the conformally compactified $\Sigma$ as a
continuous spinor field) precisely when $\Phi^\alpha\in C^0_\delta$ with
$\delta=k-\frac{1}{2}$.

It follows immediately from the definitions that $H_{s,\delta}\subset
H_{s',\delta'}$ if $\delta'\leq\delta$ and $s'\leq s$. The next lemma is
analogous to the Rellich lemma for the classical Sobolev spaces over compact
domains:

\begin{lemma}
If $\delta'<\delta$ and $s'<s$, then the injection $i:H_{s,\delta}\rightarrow
H_{s',\delta'}$ is compact.  \label{l:A.1.1}
\end{lemma}
\begin{proof}
The statement is the adaptation of Lemma 2.1 of \cite{ChCh} to the actual
spaces, and the proof is similar to that.
\end{proof}

Also, as a simple consequence of the definitions, $C^{s'}_{\delta'}\subset H_{s,
\delta}$ holds if $s'\geq s$ and $\delta'>\delta+1$. (For tensor fields on
$n$ dimensional $\Sigma$ the latter condition would be $\delta'>\delta+(n-1)
/2$, while on asymptotically flat $\Sigma$ it is known to be $\delta'>\delta
+n/2$.) In particular, the spinor fields with $o(\Omega^{\frac{3}{2}})$
fall-off, i.e. the elements of $C^0_{\delta'}$ with $\delta'>1$, are square
integrable, while those with $O(\Omega^2)$ fall-off belong to $L^\delta_2$ with
$\delta<\frac{1}{2}$. The next lemma states that, for appropriate indices,
the inclusion holds in the opposite direction, too. This statement is
analogous to the classical Sobolev lemma, and is the adaptation of Lemma 2.4
of \cite{ChCh} to the present case:

\begin{lemma}
If $s\geq s'+2$ and $\delta>\delta'-1$, then $H_{s,\delta}\subset C^{s'}
_{\delta'}$.
\label{l:A.1.2}
\end{lemma}

\begin{proof} The proof can be based on the proof of Lemma 2.3 and Lemma 2.4
of \cite{ChCh}, and on the classical Sobolev lemma $H_s\subset C^{s'}$, $s>s'+
n/2$, for compact domains and its extension from the $n$ dimensional
asymptotically Euclidean to asymptotically hyperboloidal geometries. The only
deviation from the asymptotically flat case is that in Lemma 2.3 of
\cite{ChCh} the map $T_\varepsilon: H_{s,\delta}\rightarrow H_{s,\delta'}$ on $n$
dimensional asymptotically hyperboloidal $\Sigma$ is a topological vector
space isomorphism for $\delta>\delta'-(n-1)/2$, rather than for $\delta>
\delta'-n/2$.
\end{proof}

\subsection{The basic estimates for $\tilde{\cal D}$}
\label{sub-A.2}

Since the Sen--Witten operator, given explicitly by  ${\cal D}^\alpha{}_\beta
\Phi^\beta=({\cal D}^A{}_{B'}\bar\pi^{B'},{\cal D}^{A'}{}_{B}\sigma^B)$, is
elliptic, $\tilde{\cal D}$, given explicitly by $\tilde{\cal D}^\alpha{}_\beta
\Phi^\beta:={\cal D}^\alpha{}_\beta\Phi^\beta+\frac{3}{2}K\Phi^\alpha$, is also
elliptic. If $\Phi^\alpha\in H_{1,\delta}$ (with arbitrary $\delta\in
\mathbb{R}$), then

\begin{eqnarray*}
\Bigl(\Vert{\cal D}^\alpha{}_\beta\Phi^\beta\Vert_{0,\delta}\Bigr)^2\!\!\!\!&{}
 \!\!\!\!&=\int_\Sigma\Omega^{-2\delta}\Bigl(\vert{\cal D}_{AA'}\bar\pi^{A'}\vert^2
 +\vert{\cal D}_{A'A}\sigma^A\vert^2\Bigr){\rm d}\Sigma \\
\leq\!\!\!\!&{}\!\!\!\!&\int_\Sigma\Omega^{-2\delta}\Bigl(\vert{\cal D}_{AA'}
 \bar\pi^{A'}\vert^2+\frac{3}{2}\vert{\cal D}_{(A'B'}\bar\pi_{C')}\vert^2+\vert
 {\cal D}_{A'A}\sigma^A\vert^2+\frac{3}{2}\vert{\cal D}_{(AB}\sigma_{C)}\vert^2
 \Bigr){\rm d}\Sigma \\
=\!\!\!\!&{}\!\!\!\!&\frac{3}{2}\int_\Sigma\Omega^{-2\delta}\Bigl(\vert{\cal D}
 _e\bar\pi^{A'}\vert^2+\vert{\cal D}_e\sigma^A\vert^2\Bigr){\rm d}\Sigma \leq
 \frac{3}{2}\Bigl(\Vert\Phi^\alpha\Vert_{1,\delta}\Bigr)^2
\end{eqnarray*}
holds, i.e. ${\cal D}:H_{1,\delta}\rightarrow L^\delta_2$ is a \emph{bounded}
linear operator. (Here we used the orthogonal decomposition (\ref{eq:4.6.2})
of the ${\cal D}_e$-derivative of the spinor fields.) Then

\begin{eqnarray*}
\Vert\tilde{\cal D}^\alpha{}_\beta\Phi^\beta\Vert_{0,\delta}\!\!\!\!&=\!\!\!\!&
 \Vert{\cal D}^\alpha{}_\beta\Phi^\beta+\frac{3}{2}K\Phi^\alpha\Vert_{0,\delta}\leq
 \Vert{\cal D}^\alpha{}_\beta\Phi^\beta\Vert_{0,\delta}+\frac{3}{2}
 \sqrt{\frac{\Lambda}{6}}\Vert\Phi^\alpha\Vert_{0,\delta} \\
\!\!\!\!&\leq\!\!\!\!&\sqrt{\frac{3}{2}}\Vert\Phi^\alpha\Vert_{1,\delta}+
 \frac{3}{2}\sqrt{\frac{\Lambda}{6}}\Vert\Phi^\alpha\Vert_{0,\delta}\leq
 \sqrt{\frac{3}{2}}\Bigl(1+\sqrt{\frac{\Lambda}{4}}\Bigr)\Vert\Phi^\alpha
 \Vert_{1,\delta};
\end{eqnarray*}
i.e. the renormalized Sen--Witten operator

\begin{equation}
\tilde{\cal D}:H_{1,\delta}\rightarrow L^\delta_2 \label{eq:A.2.1}
\end{equation}
is also bounded, and hence \emph{continuous}, for any $\delta\in\mathbb{R}$.

Next we prove a number of lemmas that we need in the proof of the isomorphism
theorem for $\tilde{\cal D}$ in the next subsection. The first of these is
the so-called fundamental elliptic estimate:

\begin{lemma} Let the dominant energy condition hold on $\Sigma$ and let $B
\subset\Sigma$ be an open set with compact closure and smooth boundary. Then

\begin{equation}
\Vert\Phi^\alpha_A\Vert_{1,0}\leq\sqrt{2}\Vert\tilde{\cal D}^\alpha{}_\beta\Phi
^\beta_A\Vert_{0,0}+\sqrt{1+\frac{3}{4}\Lambda}\Vert\Phi^\alpha_A\Vert_{0,0}
\label{eq:A.2.2a}
\end{equation}
for any $\Phi^\alpha_A\in H_{1,0}$, $\supp(\Phi^\alpha_A)\subset\Sigma-B$. Also,

\begin{equation}
\Vert\Phi^\alpha_B\Vert_{H_1(B)}\leq\sqrt{2}\Vert\tilde{\cal D}^\alpha{}_\beta\Phi
^\beta_B\Vert_{L_2(B)}+\sqrt{1+\frac{3}{4}\Lambda}\Vert\Phi^\alpha_B\Vert_{L_2(B)}
\label{eq:A.2.2b}
\end{equation}
for any $\Phi^\alpha_B\in H_{1,0}$, $\supp(\Phi^\alpha_B)\subset B$.
\label{l:A.2.1}
\end{lemma}

\smallskip
\noindent
{\bf Scholium:} The index $A$ of $\Phi^\alpha_A$ and $B$ of $\Phi^\alpha_B$ are
not spinor indices. The former indicates that the spinor field is localized
in the asymptotic region on $\Sigma$, while the latter is localized in a
bounded domain in $\Sigma$. (See also the first paragraph of the proof of
Lemma \ref{l:A.2.4} below.) $H_s(B)$, $s=0,1,2,...$, denotes the classical
Sobolev spaces on the domain $B$ with \emph{compact} closure.
\smallskip

\begin{proof}
For any smooth spinor field $\sigma^A$ on $\Sigma$ we have the Sen--Witten
identity

\begin{eqnarray}
\vert{\cal D}_e\sigma^A\vert^2\!\!\!\!&=\!\!\!\!&2\vert{\cal D}_{A'A}\sigma^A
 \vert^2-\frac{1}{\sqrt{2}}t_a\bigl(\varkappa T^a{}_b+\Lambda\delta^a_b\bigr)
 \sigma^B\bar\sigma^{B'} \nonumber \\
\!\!\!\!&-\!\!\!\!&\sqrt{2}D_{AA'}\Bigl(t^{AB'}\bar\sigma^{A'}{\cal D}_{B'B}
 \sigma^B-t^{BA'}\bar\sigma^{B'}{\cal D}_{BB'}\sigma^A\Bigr). \label{eq:A.2.3}
\end{eqnarray}
If $\sigma^A$ is square integrable, then the integral of the total divergence
at infinity is vanishing (see its 2-surface integral form (\ref{eq:4.6.7b})
and the argumentation in the second paragraph of subsection \ref{sub-4.6.2}),
and if $\supp(\sigma^A)\subset\Sigma-B$ or $\supp(\sigma^A)\subset B$, then
the boundary integral on $\partial B$ is also vanishing. Hence, by $\Lambda
>0$ and the dominant energy condition this and the analogous argument
for $\pi^A$ in $\Phi^\alpha=(\sigma^A,\bar\pi^{A'})$ we obtain

\begin{eqnarray}
\bigl(\Vert\Phi^\alpha\Vert_{1,0}\bigr)^2\!\!\!\!&\leq\!\!\!\!&2\bigl(
 \Vert{\cal D}^\alpha{}_\beta\Phi^\beta\Vert_{0,0}\bigr)^2+\bigl(\Vert\Phi^\alpha
 \Vert_{0,0}\bigr)^2=2\bigl(\Vert\tilde{\cal D}^\alpha{}_\beta\Phi^\beta-
 \frac{3}{2}K\Phi^\alpha\Vert_{0,0}\bigr)^2+\bigl(\Vert\Phi^\alpha\Vert_{0,0}
 \bigr)^2 \nonumber \\
\!\!\!\!&\leq\!\!\!\!&\Bigl(\sqrt{2}\Vert\tilde{\cal D}^\alpha{}_\beta\Phi
 ^\beta\Vert_{0,0}+\sqrt{1+\frac{3}{4}\Lambda}\Vert\Phi^\alpha\Vert_{0,0}\Bigr)^2.
 \label{eq:A.2.3a}
\end{eqnarray}
Thus, recalling that for $\Phi^\alpha=\Phi^\alpha_B$ the $H_{0,0}$ and $H_{1,0}$
norms coincide with the $L_2(B)$ and $H_1(B)$ norms, respectively, the
inequalities hold for smooth $\Phi^\alpha_A$ and $\Phi^\alpha_B$, respectively.

Finally, let $\Phi^\alpha$ ($=\Phi^\alpha_A$ or $\Phi^\alpha_B$) be an arbitrary
element of $H_{1,0}$. Since the smooth spinor fields form a dense subspace in
$H_{1,0}$, there exists a sequence $\{\Phi^\alpha_i\}$, $i\in\mathbb{N}$, of
smooth spinor fields in $H_{1,0}$ which converges to $\Phi^\alpha$ strongly.
Applying the estimate (\ref{eq:A.2.3a}) to the smooth spinor fields $\Phi
^\alpha_i$ and recalling that the norms $\Vert \,.\,\Vert_{s,\delta}:H_{s,\delta}
\rightarrow[0,\infty)$ are continuous, the inequalities follow.
\end{proof}
\noindent
N.B.: The statement holds true even for $\Lambda\leq0$; and in the $\Lambda
<0$ case the cosmological constant term gives additional contribution to the
constant in front of $\Vert\Phi^\alpha\Vert_{0,0}$. Note also that although in
this appendix we assume that the dominant energy condition holds, estimates
of the form (\ref{eq:A.2.2a}), (\ref{eq:A.2.2b}) could be derived from the
ellipticity of $\tilde{\cal D}$ alone, without the use of the dominant energy
condition. In this case the forthcoming statements that depend on the
fundamental elliptic estimate would hold independently of the dominant energy
condition. The price that we had to pay for the simplicity of the proof of
this estimate is the requirement of the dominant energy condition.

Another consequence of the Sen--Witten identity is given by the next lemma:

\begin{lemma} Let $B\subset\Sigma$ be an open subset with compact closure
and smooth boundary. Then

\begin{equation}
\Bigl(\Vert\tilde{\cal D}_e\Phi^\alpha_A\Vert_{0,0}\Bigr)^2\leq2\Bigl(\Vert
\tilde{\cal D}^\alpha{}_\beta\Phi^\beta_A\Vert_{0,0}\Bigr)^2+4\sqrt{\frac{2}{3}
\Lambda}\,\Vert\tilde{\cal D}^\alpha{}_\beta\Phi^\beta_A\Vert_{0,0}\,\Vert\Phi
^\alpha_A\Vert_{0,0} \label{eq:A.2.4}
\end{equation}
holds for any $\Phi^\alpha_A\in H_{1,0}$, $\supp(\Phi^\alpha_A)\subset\Sigma-B$.
\label{l:A.2.2}
\end{lemma}

\begin{proof}
By the Sen--Witten identity for $\Phi^\alpha_A=(\sigma^A,\bar\pi^{A'})$ and the
form (\ref{eq:4.5B}) of the boundary term

\begin{eqnarray*}
2\Bigl(\Vert\tilde{\cal D}^\alpha{}_\beta\Phi^\beta_A\Vert_{0,0}\Bigr)^2
 \!\!\!\!&{}\!\!\!\!&=\Bigl(\Vert\tilde{\cal D}_e\Phi^\alpha_A\Vert_{0,0}
 \Bigr)^2+\frac{\varkappa}{\sqrt{2}}\int_{\Sigma-B}t_aT^a{}_b\bigl(\sigma^B
 \bar\sigma^{B'}+\pi^B\bar\pi^{B'}\bigr){\rm d}\Sigma \\
 +4K\int_{\Sigma-B}\sqrt{2}t_{AA'}\!\!\!\!&{}\!\!\!\!&\Bigl(\bigl(
 \overline{\tilde{\cal D}}{}^A{}_{B'}\bar\sigma^{B'}\bigr)\bar\pi^{A'}-
 \bigl(\tilde{\cal D}^{A'}{}_B\sigma^B\bigr)\pi^A+\bigl(\overline{\tilde
 {\cal D}}{}^{A'}{}_B\pi^B\bigr)\sigma^A-\bigl(\tilde{\cal D}^A{}_{B'}
 \bar\pi^{B'}\bigr)\bar\sigma^{A'}\Bigr){\rm d}\Sigma
\end{eqnarray*}
holds, where we used that the integral of the boundary terms vanishes. Then
by $K=-\bar K$ and the dominant energy condition it follows that

\begin{equation*}
2\Bigl(\Vert\tilde{\cal D}^\alpha{}_\beta\Phi^\beta_A\Vert_{0,0}\Bigr)^2
\geq\Bigl(\Vert\tilde{\cal D}_e\Phi^\alpha_A\Vert_{0,0}\Bigr)^2+4\Bigl(
\langle\tilde{\cal D}^\alpha{}_\beta\Phi^\beta_A,K\Phi^\alpha_A\rangle+
\overline{\langle\tilde{\cal D}^\alpha{}_\beta\Phi^\beta_A,K\Phi^\alpha_A
\rangle}\Bigr),
\end{equation*}
where $\langle\, ,\,\rangle$ denotes the natural $L_2$ scalar product of the
Dirac spinors. Then by the Cauchy--Schwarz inequality this yields
(\ref{eq:A.2.4}).
\end{proof}

The inequality in our next lemma is analogous to the Hardy inequality:

\begin{lemma} Let $\delta\in\mathbb{R}$ and let $\Sigma$ be an asymptotically
hyperboloidal hypersurface for which the `boost parameter function' $W$
satisfies $\frac{1}{3}\Lambda W<2\delta+1$. Then there exists an open set
$B\subset\Sigma$ with compact closure and smooth boundary, and a positive
constant $c$ such that

\begin{equation}
c\Vert\Phi^\alpha_A\Vert_{0,\delta}\leq\Vert\tilde{\cal D}_e\Phi^\alpha_A\Vert
_{0,\delta} \label{eq:A.2.5}
\end{equation}
for any $\Phi^\alpha_A\in H_{1,\delta}$, $\supp(\Phi^\alpha_A)\subset\Sigma-B$.
\label{l:A.2.3}
\end{lemma}

\begin{proof}
Let $\Phi^\alpha=(\sigma^A,\bar\pi^{A'})\in H_{1,\delta}$ and $l\in\mathbb{R}$.
Then

\begin{eqnarray*}
\vert\tilde{\cal D}_e\Phi^\alpha\vert^2:=\!\!\!\!&{}\!\!\!\!&-h^{ef}\sqrt{2}
  t_{AA'}\Bigl(\bigl(\tilde{\cal D}_e\sigma^A\bigr)\bigl(\overline{\tilde
  {\cal D}}_f\bar\sigma^{A'}\bigr)+\bigl(\tilde{\cal D}_e\bar\pi^{A'}\bigr)
  \bigl(\overline{\tilde{\cal D}}_f\pi^A\bigr)\Bigr) \\
=\!\!\!\!&{}\!\!\!\!&-h^{ef}\sqrt{2}t_{AA'}\Bigl\{\Bigl(\Omega^{-l}\tilde{\cal
  D}_e\bigl(\Omega^l\sigma^A\bigr)-l\sigma^A\Omega^{-1}D_e\Omega\Bigr)\Bigl(
  \Omega^{-l}\overline{\tilde{\cal D}}_f\bigl(\Omega^l\bar\sigma^{A'}\bigr)-
  l\bar\sigma^{A'}\Omega^{-1}D_f\Omega\Bigr) \\
\!\!\!\!&{}\!\!\!\!&+\Bigl(\Omega^{-l}\tilde{\cal D}_e\bigl(\Omega^l\bar\pi
  ^{A'}\bigr)-l\bar\pi^{A'}\Omega^{-1}D_e\Omega\Bigr)\Bigl(\Omega^{-l}
  \overline{\tilde{\cal D}}_f\bigl(\Omega^l\pi^A\bigr)-l\pi^A\Omega^{-1}D_f
  \Omega\Bigr)\Bigr\} \\
=\!\!\!\!&{}\!\!\!\!&\Omega^{-2l}\vert\tilde{\cal D}_e\bigl(\Omega^l\Phi
  ^\alpha\bigr)\vert^2+l^2\Omega^{-2}\vert D_e\Omega\vert^2\vert\Phi^\alpha
  \vert^2 \\
\!\!\!\!&{}\!\!\!\!&+l\Omega^{-2l-1}\sqrt{2}\bigl(D^e\Omega\bigr){\cal D}_e
  \Bigl(\Omega^{2l}\bigl(\sigma^A\bar\sigma^{A'}+\pi^A\bar\pi^{A'}\bigr)\Bigr)
  t_{AA'},
\end{eqnarray*}
where in the last step we used the expression (\ref{eq:4.3}) of $\tilde
{\cal D}_e$ in terms of ${\cal D}_e$ and $K$. Multiplying this by $\Omega
^{-2\delta}$, choosing $\Phi^\alpha$ to be $\Phi^\alpha_A$ (i.e. in $H_{1,\delta}$
and such that $\supp(\Phi^\alpha_A)\subset\Sigma-B$), integrating on $\Sigma-B$
and using that the boundary terms both on $\partial B$ and at infinity give
zero, we obtain

\begin{eqnarray*}
\int_{\Sigma-B}\Omega^{-2\delta}\vert\tilde{\cal D}_e\Phi^\alpha_A\vert^2
  {\rm d}\Sigma\geq\!\!\!\!&{}\!\!\!\!&-l\int_{\Sigma-B}\Omega^{-2\delta}
  \Bigl((2\delta+1+l)\Omega^{-2}\vert D_e\Omega\vert^2\vert\Phi^\alpha_A
  \vert^2\\
\!\!\!\!&{}\!\!\!\!&+\Omega^{-1}\bigl(D_eD^e\Omega\bigr)\vert\Phi^\alpha_A
  \vert^2+\Omega^{-1}\vert D_e\Omega\vert \sqrt{2}v^a\chi_{ab}\bigl(\sigma^B
  \bar\sigma^{B'}+\pi^B\bar\pi^{B'}\bigr)\Bigr){\rm d}\Sigma.
\end{eqnarray*}
If, however, $B$ is chosen to be large enough, then by (\ref{eq:4.4.DO}) and
(\ref{eq:4.4.3})

\begin{equation*}
\vert D_e\Omega\vert=\frac{\Omega}{\vert T_e\vert}+O(\Omega^2),
\hskip 15pt
D_eD^e\Omega=\frac{\Omega}{\vert T_e\vert^2}+O(\Omega^2),
\hskip 15pt
v^a\chi_{ab}=\frac{1}{\vert T_e\vert}\bigl(1+\frac{1}{3}\Lambda W\bigr)v_b+
O(\Omega)
\end{equation*}
hold on $\Sigma-B$ in which the corrections to the leading terms are already
small, and hence the sign of these expressions on $\Sigma-B$ is the sign of
their leading term. Taking into account these asymptotic expressions, and
that since $\supp(\Phi^\alpha_A)\subset\Sigma-B$ the left hand side is
$(\Vert\tilde{\cal D}_e\Phi^\alpha_A\Vert_{0,\delta})^2$, we obtain

\begin{eqnarray*}
\bigl(\Vert\tilde{\cal D}_e\Phi^\alpha_A\Vert_{0,\delta}\bigr)^2\geq-l\int
  _{\Sigma-B}\bigl(\frac{1}{\vert T_e\vert^2}+O(\Omega)\bigr)\Bigl(\!\!\!\!&
  {}\!\!\!\!&(2\delta+2+l)\vert\Phi^\alpha_A\vert^2\\
\!\!\!\!&{}\!\!\!\!&+(1+\frac{1}{3}\Lambda W)\sqrt{2}v_a\bigl(\sigma^A\bar
  \sigma^{A'}+\pi^A\bar\pi^{A'}\bigr)\Bigr){\rm d}\Sigma.
\end{eqnarray*}
Since $\frac{1}{3}\Lambda W<2\delta+1$, we can always choose $l$ to be
\emph{negative} such that $\frac{1}{3}\Lambda W-1-2\delta<l<0$. Moreover,
since $(\sigma^A\bar\sigma^{A'}+\pi^A\bar\pi^{A'})$ is future pointing and
non-spacelike, $\vert\Phi^\alpha_A\vert^2=\sqrt{2}t_{AA'}(\sigma^A\bar\sigma^{A'}
+\pi^A\bar\pi^{A'})\geq\sqrt{2}\vert v_a(\sigma^A\bar\sigma^{A'}+\pi^A\bar
\pi^{A'})\vert$ holds, and hence we find that

\begin{equation*}
\bigl(\Vert\tilde{\cal D}_e\Phi^\alpha_A\Vert_{0,\delta}\bigr)^2\geq\vert l\vert
\bigl(2\delta+1-\vert l\vert-\frac{1}{3}\Lambda\max_{\cal S}\{ W\}\bigr)
\inf_{\Sigma-B}\bigl\{\frac{1}{\vert T_e\vert^2}+O(\Omega)\bigr\}\bigl(\Vert
\Phi^\alpha_A\Vert_{0,\delta}\bigr)^2;
\end{equation*}
i.e. for some positive constant $c$ and for all $\Phi^\alpha_A$ the inequality
$\Vert\tilde{\cal D}_e\Phi^\alpha_A\Vert_{0,\delta}\geq c\Vert\Phi^\alpha_A\Vert
_{0,\delta}$ holds.
\end{proof}

\begin{corollary} Under the conditions of Lemma \ref{l:A.2.3} with $\delta=0$

\begin{equation*}
\Bigl(\sqrt{c^2+8\vert K\vert^2}-2\sqrt{2}\vert K\vert\Bigr)\Vert\Phi^\alpha_A
\Vert_{0,0}\leq\sqrt{2}\Vert\tilde{\cal D}^\alpha{}_\beta\Phi^\beta_A\Vert_{0,0}
\end{equation*}
holds for any $\Phi^\alpha_A\in H_{1,0}$, $\supp(\Phi^\alpha_A)\subset\Sigma-B$.
\label{c:A.2.1}
\end{corollary}

\begin{proof}
Combining the inequalities of Lemma \ref{l:A.2.3} and Lemma \ref{l:A.2.2}
we find that

\begin{equation*}
c^2\bigl(\Vert\Phi^\alpha_A\Vert_{0,0}\bigr)^2\leq\bigl(\Vert\tilde{\cal D}_e
\Phi^\alpha_A\Vert_{0,0}\bigr)^2\leq2\Bigl(\Vert\tilde{\cal D}^\alpha{}_\beta
\Phi^\beta_A\Vert_{0,0}+2\vert K\vert \Vert\Phi^\alpha_A\Vert_{0,0}\Bigr)^2-8
\vert K\vert^2\bigl(\Vert\Phi^\alpha_A\Vert_{0,0}\bigr)^2,
\end{equation*}
which is just the inequality that we wanted to prove.
\end{proof}
Hence, on $\Sigma-B$, the $L_2$-norm of $\tilde{\cal D}^\alpha{}_\beta\Phi^\beta
_A$ is bounded from below by the $L_2$-norm of $\Phi^\alpha_A$ itself.

\begin{corollary} Under the conditions of Lemma \ref{l:A.2.3} with $\delta=0$
there is a positive constant $\tilde c$ such that

\begin{equation*}
\Vert\Phi^\alpha_A\Vert_{1,0}\leq\tilde c\Vert\tilde{\cal D}^\alpha{}_\beta
\Phi^\beta_A\Vert_{0,0}
\end{equation*}
holds for any $\Phi^\alpha_A\in H_{1,0}$, $\supp(\Phi^\alpha_A)\subset\Sigma-B$.
\label{c:A.2.2}
\end{corollary}

\begin{proof} By Corollary \ref{c:A.2.1} there is a positive constant $c_1$
such that $\Vert\Phi^\alpha_A\Vert_{0,0}\leq c_1\Vert\tilde{\cal D}^\alpha{}
_\beta\Phi^\beta_A\Vert_{0,0}$. Combining this with the inequality of Lemma
\ref{l:A.2.2}, we obtain

\begin{equation*}
\bigl(\Vert\tilde{\cal D}_e\Phi^\alpha_A\Vert_{0,0}\bigr)^2\leq2(1+4\vert K
\vert c_1)\bigl(\Vert\tilde{\cal D}^\alpha{}_\beta\Phi^\beta_A\Vert_{0,0}\bigr)^2.
\end{equation*}
Clearly, for any $\delta\in\mathbb{R}$ the $H_{s,\delta}$-norms defined with
${\cal D}_e$ and with $\tilde{\cal D}_e$ are equivalent. In fact, in
particular,

\begin{eqnarray*}
\bigl(\Vert\Phi^\alpha\Vert_{1,\delta}\bigr)^2\!\!\!\!&=\!\!\!\!&\bigl(\Vert
  \Phi^\alpha\Vert_{0,\delta}\bigr)^2+\bigl(\Vert{\cal D}_e\Phi^\alpha\Vert
  _{0,\delta}\bigr)^2 \\
\!\!\!\!&\leq\!\!\!\!&\bigl(\Vert\Phi^\alpha\Vert_{0,\delta}\bigr)^2+\bigl(
  \Vert\tilde{\cal D}_e\Phi^\alpha\Vert_{0,\delta}+\frac{\sqrt{3}}{2}\vert K
  \vert\Vert\Phi^\alpha\Vert_{0,\delta}\bigr)^2 \\
\!\!\!\!&\leq\!\!\!\!&(1+\frac{3}{4}\vert K\vert^2)\Bigl(\bigl(\Vert\Phi
  ^\alpha\Vert_{0,\delta}\bigr)^2+\bigl(\Vert\tilde{\cal D}_e\Phi^\alpha
  \Vert_{0,\delta}\bigr)^2\Bigr)+\sqrt{3}\vert K\vert\Vert\tilde{\cal D}_e
  \Phi^\alpha\Vert_{0,\delta}\,\Vert\Phi^\alpha\Vert_{0,\delta} \\
\!\!\!\!&\leq\!\!\!\!&(1+\frac{\sqrt{3}}{2}\vert K\vert)^2\Bigl(\bigl(\Vert
  \Phi^\alpha\Vert_{0,\delta}\bigr)^2+\bigl(\Vert\tilde{\cal D}_e\Phi^\alpha
  \Vert_{0,\delta}\bigr)^2\Bigr).
\end{eqnarray*}
Hence

\begin{eqnarray*}
\bigl(\Vert\Phi^\alpha_A\Vert_{1,0}\bigr)^2\!\!\!\!&\leq\!\!\!\!&(1+
  \frac{\sqrt{3}}{2}\vert K\vert)^2\Bigl(\bigl(\Vert\Phi^\alpha_A\Vert_{0,0}
  \bigr)^2+\bigl(\Vert\tilde{\cal D}_e\Phi^\alpha_A\Vert_{0,0}\bigr)^2\Bigr) \\
\!\!\!\!&\leq\!\!\!\!&(1+\frac{\sqrt{3}}{2}\vert K\vert)^2\bigl(c_1^2+8\vert
K\vert c_1+2\bigr)\bigl(\Vert\tilde{\cal D}^\alpha{}_\beta\Phi^\beta_A
  \Vert_{0,0}\bigr)^2;
\end{eqnarray*}
i.e. the inequality of the corollary holds with $\tilde c^2=(1+
\frac{\sqrt{3}}{2}\vert K\vert)^2(c_1^2+8\vert K\vert c_1+2)$.
\end{proof}

The next lemma and its corollary are the adaptation of Theorem 6.2 of
\cite{ChCh} to the present situation:

\begin{lemma}
Let $\Sigma$ be an asymptotically hyperboloidal hypersurface for which the
`boost parameter function' $W$ satisfies $\frac{1}{3}\Lambda W<1$. Then, for
any $\delta'\leq0$, there exists a positive constant $C$ such that

\begin{equation}
\Vert\Phi^\alpha\Vert_{1,0}\leq C\Bigl(\Vert\tilde{\cal D}^\alpha{}_\beta\Phi
^\beta\Vert_{0,0}+\Vert\Phi^\alpha\Vert_{0,\delta'}\Bigr) \label{eq:A.2.6}
\end{equation}
holds for any $\Phi^\alpha\in H_{1,0}$.
\label{l:A.2.4}
\end{lemma}

\begin{proof}
Let $0<a<1$, and $B_a:=\{p\in\Sigma\,\vert\,\Omega(p)>a\,\}$. Clearly, its
closure, $\overline{B_a}$, is compact with smooth boundary. Let $\beta:\Sigma
\rightarrow[0,1]$ be a smooth function such that $\beta(p)=1$ for $p\in B_a$
and $\beta(p)=0$ for $p\in\Sigma- B_{\frac{1}{2}a}$; i.e. in particular $\supp
(\beta)\subset\overline{B_{\frac{1}{2}a}}$. For any $\Phi^\alpha\in H_{1,0}$ let 
us define$\Phi^\alpha_B:=\beta\Phi^\alpha$ and $\Phi^\alpha_A:=(1-\beta)\Phi
^\alpha$, thebounded (or localized) and the asymptotic part of $\Phi^\alpha$, 
respectively.Clearly, $\supp(\Phi^\alpha_B)\subset\overline{B_{\frac{1}{2}a}}$, 
$\supp(\Phi^\alpha_A)\subset\Sigma-B_a$ and

\begin{equation}
\Vert\Phi^\alpha\Vert_{1,0}\leq\Vert\Phi^\alpha_B\Vert_{1,0}+\Vert\Phi^\alpha_A
\Vert_{1,0} \label{eq:A.2.7}
\end{equation}
hold. We derive the estimate (\ref{eq:A.2.6}) for $\Phi^\alpha_B$ and
$\Phi^\alpha_A$ separately, which by (\ref{eq:A.2.7}) yield
(\ref{eq:A.2.6}) for $\Phi^\alpha$.

Since $\supp(\Phi^\alpha_B)\subset\overline{B_{\frac{1}{2}a}}$, by
(\ref{eq:A.2.2b}) in Lemma \ref{l:A.2.1} there is a positive constant $c'$
such that

\begin{equation}
\Vert\Phi^\alpha_B\Vert_{H_1(B_{\frac{1}{2}a})}\leq c'\Bigl(\Vert\tilde{\cal D}
^\alpha{}_\beta\Phi^\beta_B\Vert_{L_2(B_{\frac{1}{2}a})}+\Vert\Phi^\alpha_B\Vert
_{L_2(B_{\frac{1}{2}a})}\Bigr). \label{eq:A.2.8}
\end{equation}
Since $\supp(\Phi^\alpha_B)\subset\overline{B_{\frac{1}{2}a}}$, the norm on the
left hand side is in fact $\Vert\Phi^\alpha_B\Vert_{1,0}$. In the first term on
the right we can write

\begin{equation*}
\tilde{\cal D}^\alpha{}_\beta\Phi^\beta_B=\tilde{\cal D}^\alpha{}_\beta\bigl(\beta
\Phi^\beta\bigr)=\beta\tilde{\cal D}^\alpha{}_\beta\Phi^\beta+\Bigl(\bar\pi^{B'}
D_{B'}{}^A\beta,\sigma^BD_B{}^{A'}\beta\Bigr),
\end{equation*}
and a straightforward calculation shows that the square of the pointwise norm
of its second term is

\begin{equation*}
\vert\bigl(\bar\pi^{B'}D_{B'}{}^A\beta,\sigma^BD_B{}^{A'}\beta\bigr)\vert^2=
\frac{1}{2}\vert D_e\beta\vert^2\,\vert\Phi^\alpha\vert^2.
\end{equation*}
Thus the first term on the right in (\ref{eq:A.2.8}) can be estimated as

\begin{eqnarray*}
\Vert\tilde{\cal D}^\alpha{}_\beta\Phi^\beta_B\Vert_{L_2(B_{\frac{1}{2}a})}
  \!\!\!\!&\leq\!\!\!\!&\Vert\tilde{\cal D}^\alpha{}_\beta\Phi^\beta\Vert
  _{L_2(B_{\frac{1}{2}a})}+\frac{1}{\sqrt{2}}\Bigl(\int_{B_{\frac{1}{2}a}}
  \vert D_e\beta\vert^2\vert\Phi^\alpha\vert^2{\rm d}\Sigma\Bigr)
  ^{\frac{1}{2}} \\
\!\!\!\!&\leq\!\!\!\!&\Vert\tilde{\cal D}^\alpha{}_\beta\Phi^\beta\Vert_{0,0}
  +\tilde C\Vert\Phi^\alpha\Vert_{L_2(B_{\frac{1}{2}a})}.
\end{eqnarray*}
Here, in the first step we used the triangle inequality and that $\beta\leq1$,
and then, in the second, we used the notation $\sqrt{2}\tilde C:=\sup\{
\vert D_e\beta\vert(p)\, \vert \, p\in B_{\frac{1}{2}a}\,\}$, the Cauchy--Schwarz
inequality and $\Vert\tilde{\cal D}^\alpha{}_\beta\Phi^\beta\Vert
_{L_2(B_{\frac{1}{2}a})}\leq\Vert\tilde{\cal D}^\alpha{}_\beta\Phi^\beta\Vert_{0,0}<
\infty$. Combining this inequality with (\ref{eq:A.2.8}) we find that

\begin{equation}
\Vert\Phi^\alpha_B\Vert_{1,0}\leq c'\Bigl(\Vert\tilde{\cal D}^\alpha{}_\beta\Phi
^\beta\Vert_{0,0}+(1+\tilde C)\Vert\Phi^\alpha_B\Vert_{L_2(B_{\frac{1}{2}a})}\Bigr).
\label{eq:A.2.9}
\end{equation}
Since $\Omega^{-2\delta'}\leq1$ for $\delta'\leq0$, $\Vert\Phi^\alpha\Vert
_{0,\delta'}\leq\Vert\Phi^\alpha\Vert_{0,0}\leq\Vert\Phi^\alpha\Vert_{1,0}<\infty$
holds; and hence the second norm on the right hand side can be estimated as

\begin{eqnarray}
\bigl(\Vert\Phi^\alpha_B\Vert_{L_2(B_{\frac{1}{2}a})}\bigr)^2\!\!\!\!&\leq\!\!\!\!&
  \bigl(\Vert\Phi^\alpha\Vert_{L_2(B_{\frac{1}{2}a})}\bigr)^2=\int_{B_{\frac{1}{2}a}}
  \Omega^{2\delta'}\Omega^{-2\delta'}\vert\Phi^\alpha\vert^2{\rm d}\Sigma
  \nonumber \\
\!\!\!\!&\leq\!\!\!\!&\sup\{\Omega^{2\delta'}(p)\,\vert\,p\in B_{\frac{1}{2}a}\,\}
  \bigl(\Vert\Phi^\alpha\Vert_{0,\delta'}\bigr)^2. \label{eq:A.2.10}
\end{eqnarray}
Combining this estimate with (\ref{eq:A.2.9}) we obtain that there exists
a positive constant $C_1$ such that

\begin{equation}
\Vert\Phi^\alpha_B\Vert_{1,0}\leq C_1\Bigl(\Vert\tilde{\cal D}^\alpha{}_\beta\Phi
^\beta\Vert_{0,0}+\Vert\Phi^\alpha\Vert_{0,\delta'}\Bigr) \label{eq:A.2.11}
\end{equation}
holds.

By Lemma \ref{l:A.2.3} and its Corollary \ref{c:A.2.2} there exist a
small enough $a\in(0,1)$ and a positive constant $\tilde c$ such that

\begin{equation}
\Vert\Phi^\alpha_A\Vert_{1,0}\leq\tilde c\Vert\tilde{\cal D}^\alpha{}_\beta
\Phi^\beta_A\Vert_{0,0} \label{eq:A.2.12}
\end{equation}
holds, where $\supp(\Phi^\alpha_A)\subset\Sigma-B_{\frac{1}{2}a}$. Since $\Phi
^\alpha_A:=(1-\beta)\Phi^\alpha$,

\begin{equation*}
\tilde{\cal D}^\alpha{}_\beta\Phi^\beta_A=(1-\beta)\tilde{\cal D}^\alpha{}_\beta
\Phi^\beta-\Bigl(\bar\pi^{B'}D_{B'}{}^A\beta,\sigma^BD_B{}^{A'}\beta\Bigr)
\end{equation*}
follows, and hence $\vert\tilde{\cal D}^\alpha{}_\beta\Phi^\beta_A\vert\leq\vert
\tilde{\cal D}^\alpha{}_\beta\Phi^\beta\vert+\frac{1}{\sqrt{2}}\vert D_e\beta
\vert\,\vert\Phi^\alpha\vert$. Using $\supp(\beta)\subset\overline{B
_{\frac{1}{2}a}}$ and the Cauchy--Schwarz inequality we obtain

\begin{eqnarray*}
\bigl(\Vert\tilde{\cal D}^\alpha{}_\beta\Phi^\beta_A\Vert_{0,0}\bigr)^2
\!\!\!\!&\leq\!\!\!\!&\int_\Sigma\vert\tilde{\cal D}^\alpha{}_\beta\Phi
  ^\beta\vert^2{\rm d}\Sigma+\int_{B_{\frac{1}{2}a}}\Bigl(\vert\tilde{\cal D}
  ^\alpha{}_\beta\Phi^\beta\vert2\tilde C\vert\Phi^\alpha\vert+\tilde C^2
  \vert\Phi^\alpha\vert^2\Bigr){\rm d}\Sigma \\
\!\!\!\!&\leq\!\!\!\!&\bigl(\Vert\tilde{\cal D}^\alpha{}_\beta\Phi^\beta
  \Vert_{0,0}\bigr)^2+2\tilde C\Vert\tilde{\cal D}^\alpha{}_\beta\Phi^\beta
  \Vert_{0,0}\,\Vert\Phi^\alpha\Vert_{L_2(B_{\frac{1}{2}a})}+\tilde C^2
  \bigl(\Vert\Phi^\alpha\Vert_{L_2(B_{\frac{1}{2}a})}\bigr)^2,
\end{eqnarray*}
where $\tilde C$ has been defined above. But by (\ref{eq:A.2.10}) and
(\ref{eq:A.2.12}) this implies that there exists a positive constant $C_2$
such that

\begin{equation}
\Vert\Phi^\alpha_A\Vert_{1,0}\leq C_2\Bigl(\Vert\tilde{\cal D}^\alpha{}_\beta\Phi
^\beta\Vert_{0,0}+\Vert\Phi^\alpha\Vert_{0,\delta'}\Bigr) \label{eq:A.2.13}
\end{equation}
holds. Finally, (\ref{eq:A.2.7}), (\ref{eq:A.2.11}) and
(\ref{eq:A.2.13}) yield the estimate (\ref{eq:A.2.6}).
\end{proof}

\begin{corollary} Under the conditions of Lemma \ref{l:A.2.4} there exists
a positive constant $C'$ such that

\begin{equation}
\Vert\Phi^\alpha\Vert_{1,0}\leq C'\Vert\tilde{\cal D}^\alpha{}_\beta\Phi^\beta
\Vert_{0,0}
\end{equation}
for any $\Phi^\alpha\in(\ker\tilde{\cal D})^\bot\cap H_{1,0}$, where $(\ker\tilde
{\cal D})^\bot$ denotes the orthogonal complement of $\ker\tilde{\cal D}$ in
$L_2$.
\label{c:A.2.3}
\end{corollary}

\begin{proof}
Suppose, on the contrary, that for any $i\in\mathbb{N}$ there exists a spinor
field $\Phi^\alpha_i\in(\ker\tilde{\cal D})^\bot\cap H_{1,0}$ for which $\Vert
\Phi^\alpha_i\Vert_{1,0}>i\Vert\tilde{\cal D}^\alpha{}_\beta\Phi^\beta_i\Vert_{0,0}$.
Then $\hat\Phi^\alpha_i:=(\Vert\Phi^\alpha_i\Vert_{1,0})^{-1}\Phi^\alpha_i$ defines
a sequence in $(\ker\tilde{\cal D})^\bot\cap H_{1,0}$ such that $\Vert\hat
\Phi^\alpha_i\Vert_{1,0}=1$, and hence

\begin{equation*}
1>i\Vert\tilde{\cal D}^\alpha{}_\beta\hat\Phi^\beta_i\Vert_{0,0}.
\end{equation*}
Thus, in particular, the sequence $\{\tilde{\cal D}^\alpha{}_\beta\hat\Phi^\beta
_i\}$, $i\in\mathbb{N}$, is Cauchy in $L_2$ and converges to zero. Since
$\{\hat\Phi^\alpha_i\}$, $i\in\mathbb{N}$, is bounded in $H_{1,0}$ and by Lemma
\ref{l:A.1.1} the injection $H_{1,0}\rightarrow H_{0,\delta'}$ is compact for
$\delta'<0$, there is a subsequence of $\{\hat\Phi^\alpha_i\}$, for the sake of
simplicity $\{\hat\Phi^\alpha_i\}$ itself, which is Cauchy in $H_{0,\delta'}$.
Applying (\ref{eq:A.2.6}) to $\hat\Phi^\alpha_i-\hat\Phi^\alpha_j$ we obtain

\begin{equation*}
\Vert\hat\Phi^\alpha_i-\hat\Phi^\alpha_j\Vert_{1,0}\leq C\Bigl(\Vert\tilde{\cal D}
^\alpha{}_\beta\hat\Phi^\beta_i-\tilde{\cal D}^\alpha{}_\beta\hat\Phi^\beta_j\Vert
_{0,0}+\Vert\hat\Phi^\alpha_i-\hat\Phi^\alpha_j\Vert_{0,\delta'}\Bigr);
\end{equation*}
i.e. $\{\hat\Phi^\alpha_i\}$, $i\in\mathbb{N}$, is Cauchy in $H_{1,0}$. Hence it
converges strongly in $H_{1,0}$ to some $\hat\Phi^\alpha$. Since the norm $\Vert
\,.\,\Vert_{1,0}:H_{1,0}\rightarrow[0,\infty)$ is continuous, $\Vert\hat\Phi
^\alpha\Vert_{1,0}=\lim_{i\rightarrow\infty}\Vert\hat\Phi^\alpha_i\Vert_{1,0}=1$, and
by the continuity of the scalar product $\hat\Phi^\alpha\in(\ker\tilde{\cal
D})^\bot$ also holds. Thus $\hat\Phi^\alpha$ is a non-zero vector which does
not belong to the kernel of $\tilde{\cal D}$. However, by the continuity of
$\tilde{\cal D}:H_{1,0}\rightarrow L_2$ and of the norm we have that $\Vert
\tilde{\cal D}^\alpha{}_\beta\hat\Phi^\beta\Vert_{0,0}=\lim_{i\rightarrow\infty}\Vert
\tilde{\cal D}^\alpha{}_\beta\hat\Phi^\beta_i\Vert_{0,0}=0$, which would yield
that $\hat\Phi^\alpha\in\ker\tilde{\cal D}$.
\end{proof}

Our last lemma is the so-called elliptic regularity estimate, which is just
Lemma A.2 of \cite{Sz12} applied now to Dirac spinors on the asymptotically
hyperboloidal $\Sigma$:

\begin{lemma} There exist positive constants $C_1$ and $C_2$ such that for
any $\Phi^\alpha\in H_{s,0}$, $s\geq1$, for which ${\cal D}^\alpha{}_\beta
\Phi^\beta\in H_{s,0}$ is also true, the inequality

\begin{equation*}
\Vert\Phi^\alpha\Vert_{s+1,0}\leq C_1\Vert{\cal D}^\alpha{}_\beta\Phi^\beta\Vert
_{s,0}+C_2\Vert\Phi^\alpha\Vert_{s,0}
\end{equation*}
holds.
\label{l:A.2.5}
\end{lemma}

\begin{proof}
The proof is a straightforward modification of that of Lemma A.2 of
\cite{Sz12}. The only essential difference between the two proofs is that in
the present case we should use orthonormal dual bases $\{e^a_{\bi},\vartheta
^{\bi}_a\}$, ${\bi}=1,2,3$, for which the connection 1-form $\gamma^{\bi}
_{e{\bj}}:=\vartheta^{\bi}_a{\cal D}_ee^a_{\bj}=\vartheta^{\bi}_aD_ee^a_{\bj}$ is
bounded on the whole $\Sigma$. Since $\Sigma$ is asymptotically hyperboloidal
with bounded intrinsic curvature (see subsection \ref{sub-4.4}), such a frame
field always exists. Also, though in the present case $\Sigma$ is not compact,
its intrinsic and extrinsic curvatures and their finitely many derivatives
(up to order $(s-1)$) are also bounded.
\end{proof}

\subsection{The isomorphism theorem for $\tilde{\cal D}$}
\label{sub-A.3}

In subsection \ref{sub-4.6.2} we showed that $\tilde{\cal D}^\alpha{}_\beta
\Phi^\beta=0$ does not have any smooth square integrable solution. First we
show that it does not have any solution even in $H_{1,0}$.

\begin{proposition}
$\ker\tilde{\cal D}\subset H_{1,0}$ is empty. \label{p:A.3.1}
\end{proposition}

\begin{proof}
If $\Phi^\alpha\in\ker\tilde{\cal D}$ such that $\Phi^\alpha\in H_{1,0}$, then
${\cal D}^\alpha{}_\beta\Phi^\beta=-\frac{3}{2}K\Phi^\alpha\in H_{1,0}$. Thus by
the elliptic regularity estimate, Lemma \ref{l:A.2.5}, it follows that
$\Phi^\alpha$ belongs to $H_{2,0}$ too, which implies that it belongs to
$H_{3,0}$, ... etc; i.e. it belongs to $H_{s,0}$ for any $s\in\mathbb{N}$. But
by the Sobolev lemma, Lemma \ref{l:A.1.2}, this implies that $\Phi^\alpha$
is smooth. However, in subsection \ref{sub-4.6.2} we showed that $\tilde
{\cal D}^\alpha{}_\beta\Phi^\beta=0$ does not have any \emph{smooth} square
integrable solution, and hence $\ker\tilde{\cal D}=\emptyset$.
\end{proof}

We also need the following two propositions:

\begin{proposition}
Under the conditions of Lemma \ref{l:A.2.4} ${\rm Im}\,\tilde{\cal D}
\subset L_2$ is a closed subspace. \label{p:A.3.2}
\end{proposition}

\begin{proof}
Let $\{\chi^\alpha_i\}$, $i\in\mathbb{N}$, be any Cauchy sequence in ${\rm Im}
\,\tilde{\cal D}\subset L_2$. By the previous proposition there is a uniquely
determined sequence $\{\Phi^\alpha_i\}$ in $H_{1,0}$ such that $\tilde{\cal D}
^\alpha{}_\beta\Phi_i^\beta=\chi^\alpha_i$. Then by Corollary \ref{c:A.2.3}
there is a positive constant $C'$ such that $\Vert\Phi^\alpha_i-\Phi^\alpha_j
\Vert_{1,0}\leq C'\Vert\chi^\alpha_i-\chi^\alpha_j\Vert_{0,0}$. Hence $\{\Phi
^\alpha_i\}$ is Cauchy in $H_{1,0}$, converging to some $\Phi^\alpha\in H_{1,0}$.
Since $\tilde{\cal D}$ is continuous, $\chi^\alpha_i=\tilde{\cal D}^\alpha
{}_\beta\Phi_i^\beta\rightarrow\tilde{\cal D}^\alpha{}_\beta\Phi^\beta\in{\rm Im}
\,\tilde{\cal D}$ if $i\rightarrow\infty$; i.e. ${\rm Im}\,\tilde{\cal D}
\subset L_2$ is a closed subspace.
\end{proof}

\begin{proposition}
${\rm i}\,\tilde{\cal D}:H_{1,0}\rightarrow L_2$ is self-adjoint.
\label{p:A.3.3}
\end{proposition}

\begin{proof}
To calculate the formal adjoint of $\tilde{\cal D}:C^\infty(\Sigma,\mathbb{D}
^\alpha)\rightarrow C^\infty(\Sigma,\mathbb{D}^\alpha)$ (with respect to the
$L_2$-scalar product), let $\Phi^\alpha=(\sigma^A,\bar\pi^{A'})\in C^\infty(
\Sigma,\mathbb{D}^\alpha)\cap L_2$ and $\chi^\alpha=(\lambda^A,\bar\mu^{A'})\in
C^\infty(\Sigma,\mathbb{D}^\alpha)$ be arbitrary. Then by integration by parts

\begin{eqnarray*}
\langle\tilde{\cal D}^\alpha{}_\beta\Phi^\beta,\chi^\alpha\rangle=\int_\Sigma
  \sqrt{2}t_{AA'}\!\!\!\!&{}\!\!\!\!&\Bigl(\bigl(D^A{}_{B'}\bar\pi^{B'}+
  \frac{1}{2}\chi t^A{}_{B'}\bar\pi^{B'}+\frac{3}{2}K\sigma^A\bigr)\bar
  \lambda^{A'} \\
\!\!\!\!&{}\!\!\!\!&+\bigl(D^{A'}{}_B\sigma^B+\frac{1}{2}\chi t^{A'}{}_B\sigma
  ^B+\frac{3}{2}K\bar\pi^{A'}\bigr)\mu^A\Bigr){\rm d}\Sigma \\
=\sqrt{2}\oint_{\cal S}\!\!\!\!&{}\!\!\!\!&\bigl(\sigma^At_A{}^{A'}v_{A'B}\mu^B+
 \bar\pi^{A'}t_{A'}{}^Av_{AB'}\bar\lambda^{B'}\bigr){\rm d}{\cal S}-\langle
 \Phi^\alpha,\tilde{\cal D}^\alpha{}_\beta\chi^\beta\rangle.
\end{eqnarray*}
Thus, the formal adjoint $\tilde{\cal D}^*:C^\infty(\Sigma,\mathbb{D}^\alpha)$
$\rightarrow C^\infty(\Sigma,\mathbb{D}^\alpha)$ of $\tilde{\cal D}$ is just
$-\tilde{\cal D}$, i.e. ${\rm i}\,\tilde{\cal D}$ is \emph{formally}
self-adjoint. Since $\Phi^\alpha$ is square integrable, its Weyl spinor parts
fall off as $o(\Omega^{3/2})$, and hence to ensure the vanishing of the
boundary integral $\chi^\alpha$ should be required to fall off at least as
$O(\Omega^{3/2})$. In fact, since $\tilde{\cal D}$ maps $H_{1,0}$ into $L_2$,
we require $\chi^\alpha$ to be square integrable, too, i.e. to fall off as $o(
\Omega^{3/2})$.

Thus, what we should show is only that the domain

\begin{equation*}
{\rm Dom}\bigl(\tilde{\cal D}^*\bigr):=\bigl\{\chi^\alpha\in L_2\vert\,\exists
\Omega^\alpha\in L_2:\,\langle\tilde{\cal D}^\alpha{}_\beta\Phi^\beta,\chi^\alpha
\rangle=\langle\Phi^\alpha,\Omega^\alpha\rangle\,\,\forall\,\Phi^\alpha\in H_{1,0}
\,\bigr\}
\end{equation*}
is just $H_{1,0}$. In fact, since the formal adjoint $({\rm i}\tilde{\cal D})
^*$ coincides with ${\rm i}\tilde{\cal D}$, their (unique) extension from the
space of the smooth square integrable spinor fields to ${\rm Dom}(\tilde
{\cal D}^*)$ and $H_{1,0}$, respectively, coincide if the domains coincide.
Note that $\Omega^\alpha$ in the definition of ${\rm Dom}(\tilde{\cal D}^*)$
is uniquely determined by $\chi^\alpha$.

Let $\chi^\alpha\in H_{1,0}$. Then $\langle{\rm i}\tilde{\cal D}^\alpha{}_\beta
\Phi^\beta,\chi^\alpha\rangle=\langle\Phi^\alpha,{\rm i}\tilde{\cal D}^\alpha
{}_\beta\chi^\beta\rangle$ for any $\Phi^\alpha\in H_{1,0}$. Thus, for any
$\chi^\alpha\in H_{1,0}$, the spinor field $\Omega^\alpha:={\rm i}\tilde{\cal D}
^\alpha{}_\beta\chi^\beta\in L_2$ is such that $\langle{\rm i}\tilde{\cal D}^\alpha
{}_\beta\Phi^\beta,\chi^\alpha\rangle=\langle\Phi^\alpha,\Omega^\alpha\rangle$ for
any $\Phi^\alpha\in H_{1,0}$. Hence, $H_{1,0}\subset{\rm Dom}(\tilde{\cal D}^*)$.

Conversely, let $\chi^\alpha\in{\rm Dom}(\tilde{\cal D}^*)$. Since $H_{1,0}
\subset L_2$ is dense, there exists a sequence $\{\chi^\alpha_i\}$, $i\in
\mathbb{N}$, in $H_{1,0}\cap{\rm Dom}(\tilde{\cal D}^*)$ such that $\chi^\alpha
_i\rightarrow\chi^\alpha$ in the $L_2$-norm. Then for any $\Phi^\alpha\in H_{1,0}$
$\langle\Phi^\alpha,{\rm i}\tilde{\cal D}^\alpha{}_\beta\chi_i^\beta\rangle=
\langle{\rm i}\tilde{\cal D}^\alpha{}_\beta\Phi^\beta,\chi^\alpha_i\rangle
\rightarrow\langle{\rm i}\tilde{\cal D}^\alpha{}_\beta\Phi^\beta,\chi^\alpha
\rangle$ when $i\rightarrow\infty$. By $\chi^\alpha\in{\rm Dom}(\tilde{\cal
D}^*)$ and the definition of ${\rm Dom}(\tilde{\cal D}^*)$ there exists a
spinor field $\Omega^\alpha\in L_2$ such that the limit on the right is
$\langle{\rm i}\tilde{\cal D}^\alpha{}_\beta\Phi^\beta,\chi^\alpha\rangle=\langle
\Phi^\alpha,\Omega^\alpha\rangle$. Hence, $\langle\Phi^\alpha,{\rm i}\tilde
{\cal D}^\alpha{}_\beta\chi_i^\beta-\Omega^\alpha\rangle\rightarrow0$ if
$i\rightarrow\infty$ for any $\Phi^\alpha\in H_{1,0}$. But since $H_{1,0}$ is
dense in $L_2$, this implies $\langle\Psi^\alpha,{\rm i}\tilde{\cal D}^\alpha
{}_\beta\chi_i^\beta-\Omega^\alpha\rangle\rightarrow0$ if $i\rightarrow\infty$
for any $\Psi^\alpha \in L_2$. Therefore, ${\rm i}\tilde{\cal D}^\alpha{}_\beta
\chi_i^\beta\rightarrow\Omega^\alpha\in L_2$ in the \emph{weak} topology of
$L_2$. Since every weakly convergent sequence is bounded, there exist positive
constants $C_1$ and $C_2$ such that $\Vert\chi^\alpha_i\Vert_{0,0}\leq C_1$ and
$\Vert\tilde{\cal D}^\alpha{}_\beta\chi_i^\beta\Vert_{0,0}\leq C_2$ holds. Hence
$\{\chi^\alpha_i\}$ is bounded in $H_{1,0}$. But every bounded sequence contains
a \emph{weakly} convergent subsequence, i.e. there is a subsequence $\{\chi
^\alpha_{i_k}\}$, $k\in\mathbb{N}$, which converges weakly to some $\chi^\alpha
_w\in H_{1,0}\subset L_2$. However, $\{\chi^\alpha_i\}$ was assumed to converge
\emph{strongly} to $\chi^\alpha\in{\rm Dom}(\tilde{\cal D}^*)\subset L_2$, the
strong and the weak limits must coincide: $\chi^\alpha=\chi^\alpha_w\in H_{1,0}$;
i.e. ${\rm Dom}(\tilde{\cal D}^*)\subset H_{1,0}$.
\end{proof}

\begin{theorem}
Under the conditions of Lemma \ref{l:A.2.4} $\tilde{\cal D}:H_{1,0}
\rightarrow L_2$ is a topological vector space isomorphism. \label{t:A.3.1}
\end{theorem}

\begin{proof}
We saw that $\tilde{\cal D}$ is continuous, and by Proposition \ref{p:A.3.1}
it is injective. Thus we should show only that it is surjective, too, because,
by the open mapping theorem, any continuous bijection between two Banach
spaces is a topological vector space isomorphism.

Since by Proposition \ref{p:A.3.3} ${\rm i}\tilde{\cal D}$ is self-adjoint,
${\rm Im}\,\tilde{\cal D}={\rm Im}\,\tilde{\cal D}^*$ holds, and suppose, on
the contrary, that ${\rm Im}\,\tilde{\cal D}\not=L_2$. Since by Proposition
\ref{p:A.3.2} ${\rm Im}\,\tilde{\cal D}$ is closed, this implies that
$({\rm Im}\,\tilde{\cal D})^\perp$, the orthogonal complement of ${\rm Im}\,
\tilde{\cal D}$ in $L_2$, is not empty. Thus, let $\chi^\alpha\in({\rm Im}\,
\tilde{\cal D})^\perp$ be a non-zero vector. Since $H_{1,0}\subset L_2$ is
dense, there is a sequence $\{\chi^\alpha_i\}$, $i\in\mathbb{N}$, in $H_{1,0}
\cap({\rm Im}\,\tilde{\cal D})^\perp$ such that $\chi^\alpha_i\rightarrow\chi
^\alpha$ if $i\rightarrow\infty$. Then, by the self-adjointness of ${\rm i}
\tilde{\cal D}$, for any $\Phi^\alpha\in H_{1,0}$ we have that $0=\langle\chi
^\alpha_i,{\rm i}\tilde{\cal D}^\alpha{}_\beta\Phi^\beta\rangle=\langle{\rm i}
\tilde{\cal D}^\alpha{}_\beta\chi^\beta_i,\Phi^\alpha\rangle$. Since $H_{1,0}
\subset L_2$ is dense, this implies that $\chi^\alpha_i\in\ker\tilde{\cal D}
=\emptyset$. However, this yields $\chi^\alpha_i=0$ for any $i\in\mathbb{N}$,
and hence that $\chi^\alpha=0$, which is a contradiction.
\end{proof}

This completes the proof of the existence and uniqueness of the solution of
the renormalized Witten equation.

\bigskip
\noindent
The authors are grateful to one of the referees for the remark in footnote 1
and the reference \cite{Fr13}.
LBSz is grateful to J\"org Frauendiener for the discussions on the geometry
of asymptotically hyperboloidal hypersurfaces. PT gratefully acknowledges
the Wigner Research Centre for Physics for hospitality while this work was
in progress. The work was partially supported by NEFIM.


\end{document}